\documentclass[smallextended,numbook]{svjour3}     
\smartqed  
\usepackage{setspace}
\usepackage{bm}
\usepackage{latexsym}
\usepackage{graphicx}
\usepackage{graphics}
\usepackage{subeqnarray}
\usepackage{amsmath}
\usepackage{amssymb}
\usepackage{subfigure}
\usepackage{multirow}
\usepackage{ifthen}
\usepackage{color}
\usepackage{times}
\usepackage{enumerate}
\usepackage[stable]{footmisc}
\usepackage{cite}
\newcommand{\vnorm}[1]{\|#1\|}
\newcommand{\mcal}{\mathcal}

\newcommand{\bmd}[1]{\dot{\bm{#1}}}
\newcommand{\lag}{\mathcal{L}}
\newcommand{\ham}{\mathcal{H}}
\newcommand{\abs}[1]{\vert#1\vert}
\newcommand{\divr}{\operatorname{div}}
\newcommand{\tr}{\operatorname{tr}}
\newcommand{\vol}{\operatorname{vol}}
\newcommand{\dfm}{\bm{F}}
\newcommand{\stress}{\bm{\sigma}}
\newcommand{\stressave}{\bm{\sigma}_{\rm av}}
\newcommand{\euclid}{\mathcal{E}}
\newcommand{\body}{\mathcal{B}}
\newcommand{\pot}{\mathcal{V}}
\newcommand{\poteam}{\mathcal{V}^{\rm EAM}}
\newcommand{\potfisher}{\widehat{\mathcal{V}}}
\newcommand{\kin}{\mathcal{T}}

\newcommand{\inv}[1]{#1^{-1}}

\newcommand{\murdoch}{DA}
\newcommand{\Qz}{\bm{Q}_{\bm{z}}}

\newcommand{\rhomh}{\tilde{\rho}}
\newcommand{\stressmh}{\tilde{\stress}}
\newcommand{\bmh}{\tilde{\bm{b}}}
\newcommand{\fmh}{\tilde{\bm{f}}}
\newcommand{\vmh}{\tilde{\bm{v}}}
\newcommand{\pmh}{\tilde{\bm{p}}}

\newcommand{\eref}[1]{(\ref{#1})}
\newcommand{\sref}[1]{Section~\ref{#1}}
\newcommand{\fref}[1]{Fig.~\ref{#1}}
\newcommand{\real}[1]{\mathbb{R}^#1}
\newcommand{\tbm}[1]{\tilde{\bm{#1}}}
\newcommand{\ol}[1]{\overline{#1}}
\newcommand{\T}{{\rm T}}
\newcommand{\kb}{\ensuremath{k_{\rm B}}}
\newcommand{\ang}{\text{\AA}}

\journalname{J Elast}
\begin{document}

\title{A unified interpretation of stress in molecular systems\thanks{This work
was partly supported through NSF (DMS-0757355). This article has drawn heavily
upon material from Ellad~Tadmor and Ronald~Miller, {\em Modeling Materials:
Continuum, Atomistic and Multiscale Techniques}, \copyright 2010 Ellad~Tadmor
and Ronald Miller, forthcoming Cambridge University Press, reproduced with
permission. The final publication is available at www.springerlink.com. \\
}
}
\dedication{The authors would like to dedicate this article to Jack~Irving, who passed away in 2008 at the age of 87. Irving, while a graduate student on leave from Princeton, worked with Prof. John Kirkwood at Caltech on the fundamental non-equilibrium statistical mechanics theory which serves as the basis for the present article.}

\author{Nikhil~Chandra~Admal \and
        E.~B.~Tadmor
}


\institute{Nikhil~Chandra~Admal \at
              Department of Aerospace Engineering and Mechanics, The University of Minnesota.\\
              \email{admal002@umn.edu}           
           \and
           Ellad~B.~Tadmor \at
              Department of Aerospace Engineering and Mechanics, The University of Minnesota.\\
              \email{tadmor@aem.umn.edu}
}

\date{Received: 11 January 2010/ Published online: 27 May 2010}

\maketitle
\begin{abstract}
The microscopic definition for the Cauchy stress tensor has been examined in the past from many different perspectives. This has led to different expressions for the stress tensor and consequently the ``correct'' definition has been a subject of debate and controversy.  In this work, a unified framework is set up in which all existing definitions can be derived, thus establishing the connections between them. The framework is based on the non-equilibrium statistical mechanics procedure introduced by Irving, Kirkwood and Noll, followed by spatial averaging. The Irving--Kirkwood--Noll procedure is extended to multi-body potentials with continuously differentiable extensions and generalized to \emph{non-straight bonds}, which may be important for particles with internal structure. Connections between this approach and the direct spatial averaging approach of Murdoch and Hardy are discussed and the Murdoch--Hardy procedure is systematized.  Possible sources of non-uniqueness of the stress tensor, resulting separately from both procedures, are identified and addressed. Numerical experiments using molecular dynamics and lattice statics are conducted to examine the behavior of the resulting stress definitions including their convergence with the spatial averaging domain size and their symmetry properties.
\end{abstract}

\section{Introduction}
\label{ch:intro}
Continuum mechanics provides an efficient theoretical framework for modeling materials science phenomena.  To characterize the behavior of materials, \emph{constitutive relations} serve as an input to the continuum theory. These constitutive models have functional forms which must be consistent with material frame-indifference and the laws of thermodynamics and include parameters that are fitted to reproduce experimental observations.  With the advent of modern computing power, atomistic simulations through ``numerical experiments'' offer the potential for studying different materials and arriving at their constitutive laws from first principles. This could make it possible to design new materials and to improve the properties of existing materials in a systematic fashion. To use the data obtained from an atomistic simulation to build a constitutive law, which is framed in the language of continuum mechanics, it is necessary to understand the connection between continuum fields and the underlying microscopic dynamics.  

Another arena where the connection between continuum and atomistic concepts is important is the field of \emph{multiscale modeling}. This discipline involves the development of computational tools for studying problems where two or more length and/or time scales play a major role in determining macroscopic behavior. A prototypical example is fracture mechanics where the behavior of a crack is controlled by atomic-scale phenomena at the crack-tip, while at the same time long-range elastic stress fields are set up in the body.  Many advances have been made in the area of multiscale modeling in recent years. Some common atomistic/continuum coupling methods are quasicontinuum \cite{tadmor1996,shenoy1999}, coupling of lengthscales \cite{rudd2000}, cluster quasicontinuum \cite{knap2001}, bridging domain \cite{xiao2004}, coupled atomistics and discrete dislocations \cite{shilkrot2004}, and heterogeneous multiscale methods \cite{e2007}, to name just a few. Refer to \cite{tadmor2009} for a comparison of some prominent atomistic/continuum coupling multiscale methods. In a multiscale method, a key issue involves the transfer of information between the discrete model and the continuum model. It is therefore of practical interest to understand how to construct definitions of continuum fields for an atomistic system, to ensure a smooth transfer of information between the discrete and continuum domains.

In this paper, we focus on just one aspect of the continuum-atomistic connection, namely the interpretation of the (Cauchy) stress tensor in a discrete system. This question has been explored from many different perspectives for nearly two hundred years and this has led to various definitions that do not appear to be consistent with each other. As a result, the ``correct'' definition for the stress tensor has been a subject of great debate and controversy.  We begin with a brief historical survey. 

\subsubsection*{A brief history of microscopic definitions for the stress tensor}
Interest in microscopic definitions for the stress tensor dates back at least to Cauchy in the 1820s \cite{Cauchy1828a,Cauchy1828b} with his aim to define stress in a crystalline solid. Cauchy's original definition emerges from the intuitive idea of identifying stress with the force per unit area carried by the bonds that cross a given surface. A comprehensive derivation of Cauchy's approach is given in Note B of Love's classic book on the theory of elasticity \cite{love}. Since this approach is tied to the particular surface being considered, it actually constitutes a definition for the {\em traction} (or {\em stress vector}) and not for the stress tensor.  The first definition of stress as a tensorial quantity follows from the works of Clausius \cite{clausius1870} and Maxwell \cite{maxwell1870,maxwell1874} in the form of \emph{virial theorem}. Clausius states the virial theorem as
\begin{quotation}
\emph{The mean vis viva of a system is equal to its virial.}
\end{quotation}
By ``vis viva'' (literally ``living force''), Clausius means kinetic energy, while the term ``virial'' comes from the Latin ``vis'' (pl. ``vires'') meaning force.  The virial theorem leads to a definition for pressure in a gas. Maxwell \cite{maxwell1870,maxwell1874} extended Clausius' work and showed the existence of a tensorial version of the virial theorem (see Appendix \ref{ch:virial}). The {\em virial stress} resulting from the virial theorem is widely used even today in many atomistic simulations due to its simple form and ease of computation.  Unlike Cauchy's original definition for stress, the virial stress includes a contribution due to the kinetic energy of the particles. This discrepancy was addressed by Tsai \cite{tsai1979}, who extended the definition given by Cauchy to finite temperature by taking into consideration the momentum flux passing through the surface. Let us refer to this stress vector as the \emph{Tsai traction}. 

An alternative approach for defining the stress tensor was pioneered in the landmark paper of Irving and Kirkwood \cite{ik1950}. Irving and Kirkwood derived the equations of hydrodynamics from the principles of non-equilibrium classical statistical mechanics and in the process established a pointwise definition for various continuum fields including the stress tensor. Although their work was indeed noteworthy, the stress tensor obtained involved a series expansion of the Dirac delta distribution which is not mathematically rigorous. Continuing their work, Noll \cite{noll1955} proved two lemmas, which allowed him to avoid the use of the Dirac delta distribution, and thus arrive at a closed-form expression for the stress tensor which does not involve a series expansion. We refer to the procedure introduced by Irving and Kirkwood and extended by Noll as the \emph{Irving--Kirkwood--Noll procedure}. Schofield and Henderson \cite{schofield1982} highlighted the non-uniqueness present in the stress tensor derived by Irving and Kirkwood, and pointed out that it could result in a non-symmetric stress tensor. There have been several attempts to improve on the Irving and Kirkwood procedure. In particular, Lutsko \cite{lutsko1988} reformulated this procedure in Fourier space.  An error in Lutsko's derivation was corrected by Cormier et al. \cite{cormier2001}. 

Due to the stochastic nature of the Irving and Kirkwood stress, many difficulties arise when one tries to use their expression in atomistic simulations. To avoid these difficulties, Hardy and co-workers \cite{hardy1982,hardy2002} and independently Murdoch \cite{murdoch1982,murdoch1993,murdoch1994,murdoch2003,murdoch2007} developed a new approach that bypasses the mathematical complexity of the Irving and Kirkwood procedure. This is done by defining continuum fields as direct spatial averages of the discrete equations of motion using weighting functions with compact support. In particular, this approach leads to the so-called \emph{Hardy stress} \cite{hardy1982} often used in molecular dynamics simulations. Murdoch in \cite{murdoch2007} provides an excellent description of the spatial averaging approaches currently being used and discusses the non-uniqueness of the stress tensor resulting from the spatial averaging procedure. We refer to the direct spatial averaging approach as the \emph{Murdoch--Hardy procedure}.

Another approach, which leads to a stress tensor very similar to that obtained by Irving and Kirkwood is the reformulation of elasticity theory using peridynamics \cite{silling2000}. Lehoucq and Silling \cite{lehoucq2008} have recently shown that Noll's solution is a minimum solution in a variational sense. Morante et al. \cite{morante2006} proposed a new approach for defining the stress tensor using the invariance of partition function under infinitesimal canonical point transformations. However, their approach is limited to equilibrium statistical mechanics and involves taking derivatives of delta distributions.

We can summarize the ``state of the art'' for the microscopic definition of the stress tensor as follows.  There are currently at least three definitions for the stress tensor which are commonly used in atomistic simulations: the virial stress, the Tsai traction, and the Hardy stress \cite{zim2004}. The importance of the Irving and Kirkwood formulation is recognized, however, it is not normally used in practice and its connection with the other stress definitions is not commonly understood. The difference between \emph{pointwise} stress measures and temporal and/or spatially-averaged quantities is often not fully appreciated. The result is that the connection between the Cauchy stress tensor defined in continuum mechanics and its analogue, defined for a discrete system, remains controversial and continues to be a highly-debated problem.

\subsubsection*{A unified framework for the microscopic definition for the stress tensor}
In this paper, a unified framework based on the Irving--Kirkwood--Noll procedure is established which \emph{leads to all of the major stress definitions} discussed above and identifies additional possible definitions. Since all of the definitions are obtained from a common framework the connections between them can be explored and analyzed and the uniqueness of the stress tensor can be established. An overview of the approach and the organization of the paper are described below.

Before turning to the general framework, we begin in \sref{ch:canonical} with a derivation of the virial stress tensor within the framework of equilibrium statistical mechanics using the technique of canonical transformations.  Although this derivation is quite different from the Irving--Kirkwood--Noll procedure, it provides insight into how the geometric ideas of mechanics can be used to derive the stress tensor. It also provides a limit to which the general non-equilibrium stress tensor must converge under equilibrium conditions in the thermodynamic limit. This is used later to establish the uniqueness of the stress tensor obtained from our general unified framework.

Next, we turn to the construction of the new unified framework. In \sref{ch:phase}, we extend the Irving--Kirkwood--Noll procedure \cite{ik1950,noll1955}, originally derived for pair potential interactions, to multi-body potentials. Due to the invariance of the potential energy function with respect to the Euclidean group, it can be shown that any multi-body potential can be expressed as a function of distances between particles. When expressed in this form, we note that for a system of more than 4 particles, this function is only defined on a manifold since the $N(N-1)/2$ distances between $N$ particles in $\real{3}$ are not independent for $N\ge5$. To apply the Irving--Kirkwood--Noll procedure to multi-body potentials, we recognize that the potential energy function must be {\em extended} from its manifold to a higher-dimensional Euclidean space as a continuously differentiable function.  We show that if such an extension exists, then an infinite number of equivalent extensions can be constructed using \emph{Cayley-Menger determinants}, which describe the constraints that the distances between particles embedded in $\real{3}$ must satisfy. Then for multi-body potentials that possess continuously differentiable extensions (which is the case for most practical interatomic potentials), we establish the key result that due to the balance of linear and angular momentum, \emph{the force on a particle in a discrete system can always be decomposed as a sum of central forces between particles}, i.e., forces that are parallel to the lines connecting the particles. In other words, the \emph{strong law of action and reaction} is always satisfied for such multi-body potentials.  We show, that although the net force on a particle calculated using \emph{any} extension is the same, its decomposition into central forces is generally different for different extensions. Using this result we show that the pointwise stress tensor resulting from the Irving--Kirkwood--Noll procedure is non-unique and symmetric. We also show, that a generalization of Noll's lemmas \cite{noll1955} to \emph{non-straight bonds} gives a non-symmetric stress tensor that may be important for particles with internal structure, such as liquid crystals.

The {\em macroscopic} stress tensor corresponding to the pointwise stress tensor described above is obtained in \sref{ch:spatial} through a procedure of spatial averaging.  The connection between this stress and the stress tensors obtained via the direct spatial averaging procedure introduced by Murdoch \cite{murdoch1982,murdoch1993,murdoch1994,murdoch2003,murdoch2007} and Hardy \cite{hardy1982} is explored and in the process the Murdoch--Hardy procedure is systematized and generalized to multi-body potentials using the results of \sref{ch:phase}.  The non-uniqueness of the stress tensor, inherent in the Murdoch--Hardy procedure is studied and a general class of possible definitions under this procedure are identified. The connection between the non-uniqueness in the Murdoch--Hardy procedure and the non-uniqueness mentioned in \sref{ch:phase} is addressed. 

In \sref{ch:compare}, various stress definitions including the Hardy stress, the Tsai traction and the virial stress are shown to be special cases of the macroscopic stress tensor derived from the extended Irving--Kirkwood--Noll procedure in \sref{ch:spatial}. The original definitions for these measures are generalized in this manner to multi-body potentials. The existence of different extensions for the potential energy function, which led to non-uniqueness of the pointwise stress tensor discussed in \sref{ch:phase}, also result in the non-uniqueness of these definitions. However it is shown that the difference in the macroscopic stress tensor resulting from this non-uniqueness tends to zero in the \emph{thermodynamic limit}\footnote{\label{foot:tdlimit} The thermodynamic limit is the state obtained as the number of particles, $N$, and the volume, $V$, of the system tend to infinity in such a way that the ratio $N/V$ is constant.}. Another source of non-uniqueness explored in this section is that given any definition for the stress tensor, a new definition, which also satisfies the balance of linear momentum, can be obtained by adding to it an arbitrary tensor field with zero divergence.  It is shown that in the thermodynamic limit the macroscopic stress tensor obtained in \sref{ch:spatial} converges to the virial stress derived in \sref{ch:canonical}.

To address practical aspects of the different definitions obtained within the unified framework, \sref{ch:experiment} describes several ``numerical experiments'' involving molecular dynamics and lattice statics. These simulations are designed to examine the behavior of these stress definitions, including their convergence with averaging domain size and their symmetry properties. Our conclusions and directions for future research are presented in \sref{ch:conclusions}.

\subsubsection*{Notation}
In this paper, vectors are denoted by lower case letters in bold font and tensors of higher order are denoted by capital letters in bold font. The tensor product of two vectors is denoted by the symbol ``$\otimes$'' and the inner product of two vectors is denoted by a dot ``$\cdot$''. The inner product of two second-order tensors is denoted by ``:''. A second-order tensor operating on a vector is denoted by juxtaposition, e.g., $\bm{T}\bm{v}$.  The gradient of a vector field, $\bm{v}(\bm{x})$, is denoted by $\nabla_{\bm{x}} \bm{v}(\bm{x})$, which in indicial notation is given by
$[ \nabla_{\bm{x}} \bm{v} ]_{ij} = \partial \bm{v}_i/\partial \bm{x}_j$.
The divergence of a tensor field, $\bm{T}(\bm{x})$, is denoted by $\divr_{\bm{x}} \bm{T}(\bm{x})$. The divergence of a vector field is defined as the trace of its gradient. The divergence of a second-order tensor field in indicial notation (with Einstein's summation convention) is given by
$[ \divr_{\bm{x}} \bm{T} ]_i = \partial \bm{T}_{ij}/\partial \bm{x}_j $. The notation described above is followed unless otherwise explicitly stated.

\section{Stress in an equilibrium system}
\label{ch:canonical}
In this section, we obtain expressions for the Cauchy stress in an equilibrium system using the technique of canonical transformations. The basic philosophy behind canonical transformation is explained in the next section.

\subsection{Canonical transformations}
\label{sec:can_trans}
Consider a system consisting of $N$ point masses whose behavior is governed by classical mechanics. Let $\bm{q}_\alpha(t)$ and $\bm{p}_\alpha(t)$ ($\alpha = 1,2,\dots,N)$ denote the generalized coordinates and momenta of the system.\footnote{In a general theory of canonical transformations, $\bm{q}_\alpha$ and $\bm{p}_\alpha$ need not denote the actual position and momentum of particle $\alpha$.} For brevity, we sometimes use $\bm{q}(t)$ and $\bm{p}(t)$ to denote the vectors $(\bm{q}_1(t),\bm{q}_2(t),\dots,\bm{q}_N(t))$ and $(\bm{p}_1(t), \bm{p}_2(t),\dots, \bm{p}_N(t))$, respectively. The time evolution of the system can be studied through three well-known approaches, referred to as the {\em Newtonian formulation}, the {\em Lagrangian formulation}, and the {\em Hamiltonian formulation}. The first approach is used in molecular dynamics simulations, while the latter two approaches are more elegant and can sometimes be used to obtain useful information from systems in the absence of closed-form solutions.

In the Lagrangian formulation, a system is characterized by the vector $\bm{q}(t)$ and a Lagrangian function $\lag$, given by
\begin{equation}
\lag(\bm{q},\dot{\bm{q}};t) = \kin(\dot{\bm{q}}) - \pot(\bm{q}),
\end{equation}
where $\kin$ is the kinetic energy of the system, $\pot$ is the potential energy of the system, and $\dot{\bm{q}}(t)$ represents the time derivative of $\bm{q}(t)$. It is useful to think of $\bm{q}$ as a point in a $3N$-dimensional {\em configuration space}.  The time evolution of $\bm{q}(t)$ in configuration space is described by a variational principle called \emph{Hamilton's principle}. Hamilton's principle states that the time evolution of $\bm{q}(t)$ corresponds to the extremum of the action integral defined as a functional of $\bm{q}$ by
\begin{equation}
\mathcal{A}[\bm{q}] = \int_{t_1}^{t_2} \lag(\bm{q},\dot{\bm{q}};t) \, dt,
\label{eqn:action}
\end{equation}
where $t_1$, $t_2$, $\bm{q}(t_1)$ and $\bm{q}(t_2)$ are held fixed with respect to the class of variations being considered \cite[Section V.1]{lanczos}. In mathematical terms, we require that
\begin{equation}
\delta \mcal{A} = 0,
\label{eqn:var_lag}
\end{equation}
while keeping the ends fixed as described above. The Euler--Lagrange equation arising from \eref{eqn:var_lag} is
\begin{equation}
\frac{d}{dt} \left (\frac{\partial \lag}{\partial \bmd{q}_\alpha} \right) - \frac{\partial \lag}{\partial \bm{q}_\alpha} = \bm{0}.
\label{eqn:lagrange}
\end{equation}
The Lagrangian formulation is commonly used as a calculation tool in solving simple problems. 

Next, we note that the Lagrangian is the Legendre transform of the Hamiltonian $\ham$, \cite[Section VI.2]{lanczos}, 
\begin{equation}
\lag(\bm{q},\dot{\bm{q}}; t) 
= \sup_{\bm{p}} [ \bm{p} \cdot \bmd{q} - \ham(\bm{p},\bm{q};t) ].
\end{equation}
The Hamiltonian is the total energy of the system. Using the Hamiltonian, equation~\eref{eqn:var_lag} can be rewritten as 
\begin{equation}
\delta \int_{t_1}^{t_2} \left [ \bm{p} \cdot \bmd{q} - \ham(\bm{p},\bm{q};t) \right ] \, dt = 0.
\label{eqn:var_hamilton}
\end{equation}
Note that in \eref{eqn:var_lag}, the variation is only with respect to $\bm{q}$, whereas in \eref{eqn:var_hamilton}, the functional depends on the functions $\bm{q}$ and $\bm{p}$, and variations are taken with respect to both $\bm{q}$ and $\bm{p}$ independently. In both cases, $t_1$, $t_2$, $\bm{q}(t_1)$ and $\bm{q}(t_2)$ are held fixed. The variational principle given in \eref{eqn:var_hamilton} is commonly referred as the \emph{modified Hamilton's principle} \cite{goldstein} or simply as the ``Hamiltonian formulation''.  The advantage of the Hamiltonian formulation lies not in its use as a calculation tool, but rather in the deeper insight it affords into the formal structure of mechanics. The Euler--Lagrange equations associated with \eref{eqn:var_hamilton} are
\begin{align}
\bmd{q}_\alpha &= \nabla_{\bm{p}_\alpha} \ham,\label{eqn:ham1}\\
\bmd{p}_\alpha &= -\nabla_{\bm{q}_\alpha} \ham,\label{eqn:ham2}
\end{align}
commonly called Hamilton's equations. The above equations are also referred to as the \emph{canonical equations of motion}\footnote{The term ``canonical'' in this context has nothing to do with the canonical ensemble of statistical mechanics.  The terminology was introduced by Jacobi to indicate that Hamilton's equations constitute the simplest form of the equations of motion.}. 

It is important to note that the Hamiltonian formulation is more general than the Lagrangian formulation, since it accords the coordinates and momenta independent status, thus providing the analyst with far greater freedom in selecting generalized coordinates.  We now think of $(\bm{q},\bm{p})$ as a point in a $6N$-dimensional \emph{phase space}, as opposed to the $3N$-dimensional configuration space of the Lagrangian formulation.  The choice of $\bm{q}$ and $\bm{p}$ is not arbitrary, however, since the selected variables must satisfy the canonical equations of motion. For this reason $\bm{q}$ and $\bm{p}$ are called \emph{canonical variables}. 

The requirement that the generalized coordinates and momenta must be canonical means that new sets of generalized coordinates can be derived from a given set through a special kind of transformation defined below.

\newtheorem{definitions}{Definition}
\begin{definitions}
Any transformation of generalized coordinates that preserves the canonical form of Hamilton's equations is said to be a canonical transformation.\footnote{This definition suffices for our purpose, but a more correct definition can be found in \cite{arnold} using \emph{differential forms}.}
\end{definitions}

The construction of canonical transformations is facilitated by the introduction of \emph{generating functions} as explained below.

\subsubsection*{Generating functions}
Consider two sets of canonical variables $(\bm{Q},\bm{P})$ and $(\bm{q},\bm{p})$, related to each other through a canonical transformation given by
\begin{equation}
\bm{Q} = \bm{Q}(\bm{q},\bm{p},t), \qquad 
\bm{P} = \bm{P}(\bm{q},\bm{p},t).
\label{eqn:trans_pq_PQ}
\end{equation}
Since the variables are canonical, they satisfy the modified Hamilton's principle in \eref{eqn:var_hamilton},
\begin{align}
\delta \int_{t_1}^{t_2} \left [ \bm{p} \cdot \bmd{q} - \ham(\bm{p},\bm{q};t) \right ] \, dt &= 0, \label{eqn:actionvarone} \\
\delta \int_{t_1}^{t_2} \left [ \bm{P} \cdot \bmd{Q} - \hat{\ham}(\bm{P},\bm{Q};t) \right ] \, dt &= 0, \label{eqn:actionvartwo}
\end{align}
where $\hat{\ham}$ is defined later as part of the canonical transformation. The integrands of \eref{eqn:actionvarone} and \eref{eqn:actionvartwo} can therefore only differ by a quantity whose variation after integration is identically zero. A possible solution is
\begin{equation}
\delta \int_{t_1}^{t_2} \left [ \bm{p} \cdot \bmd{q} - \bm{P} \cdot \bmd{Q} -  (\ham - \hat{\ham}) \right ] \, dt = \delta \int_{t_1}^{t_2} \frac{dG}{dt} \, dt,
\end{equation}
where $G$ is an arbitrary scalar function of the canonical variables and time, with continuous second derivatives. The integral on the right is only evaluated at fixed integration bounds and its variation is zero. This is not obvious since there is no restriction on the variation of the momenta at the ends. We assume this to be true to avoid the introduction of differential forms. For a mathematically rigorous argument refer to \cite[Section 45]{arnold}\footnote{Briefly the proof is based on the symmetry present in the geometry of any Hamiltonian system commonly called \emph{symplectic geometry}.}. The difference between the integrands of \eref{eqn:actionvarone} and \eref{eqn:actionvartwo} therefore satisfies,
\begin{equation}
dG - \bm{p} \cdot d\bm{q} + \bm{P} \cdot d\bm{Q} + (\ham - \hat{\ham})dt = 0.
\label{eqn:diff_G_1}
\end{equation}
Now, consider the case where $G = G_1(\bm{q},\bm{Q},t)$. The total differential of $G$ is then
\begin{equation}
dG = \nabla_{\bm{q}} G_1 \cdot d\bm{q} + \nabla_{\bm{Q}} G_1 \cdot d\bm{Q} + \frac{\partial G_1}{\partial t} dt.
\label{eqn:diff_G_3}
\end{equation}
Substituting \eref{eqn:diff_G_3} into \eref{eqn:diff_G_1} gives
\begin{equation}
\left ( \nabla_{\bm{q}} G_1 - \bm{p} \right ) \cdot d\bm{q} + \left ( \nabla_{\bm{Q}} G_1 + \bm{P} \right ) \cdot d\bm{Q} + \left ( \frac{\partial G_1}{\partial t} + \ham - \hat{\ham} \right ) dt = 0.
\end{equation}
Since $\bm{q}$, $\bm{Q}$ and $t$ are independent, the above equation is satisfied provided that
\begin{equation}
\bm{p}_\alpha = \frac{\partial G_1}{\partial \bm{q}_\alpha}, \qquad \bm{P}_\alpha = -\frac{\partial G_1}{\partial \bm{Q}_\alpha}, \qquad \hat{\ham} = \ham + \frac{\partial G_1}{\partial t}.
\label{eqn:tranform_1}
\end{equation}
The above relations define the canonical transformation. Since $G_1$ generates the transformation, it is commonly called the \emph{generating function} of the canonical transformation. Note that if $G_1$ does not depend on time $t$, then $\hat{\ham} = \ham$.

The generating functions of the form, $G = G_1(\bm{q},\bm{Q},t)$, does not generate all possible canonical transformations. In general, there are four primary classes of generating functions where the functional dependence is $(\bm{q},\bm{Q})$, $(\bm{q},\bm{P})$, $(\bm{p},\bm{Q})$ and $(\bm{p},\bm{P})$.\footnote{In addition to these four classes of transformation, it is possible to have a mixed dependence, where each degree of freedom can belong to a different class \cite{goldstein}.} We have already encountered the first class, where $G = G_1(\bm{q},\bm{Q},t)$. The remaining classes can be obtained from the first through Legendre transformations. Consider for example, the following definition,
\begin{equation}
G = G_3(\bm{p},\bm{Q},t) + \bm{q} \cdot \bm{p}.
\end{equation}
The total differential of this expression is
\begin{equation}
dG = \nabla_{\bm{p}} G_3 \cdot d\bm{p} + \nabla_{\bm{Q}} G_3 \cdot d\bm{Q} + \frac{\partial G_3}{\partial t} dt + \bm{q} \cdot d\bm{p} + \bm{p} \cdot d\bm{q}.
\end{equation}
Substituting the above equation into \eref{eqn:diff_G_1} gives
\begin{equation}
\left ( \nabla_{\bm{p}} G_3  + \bm{q} \right ) \cdot d\bm{p} + \left( \nabla_{\bm{Q}} G_3 + \bm{P} \right ) \cdot d\bm{Q} + \left( \frac{\partial G_3}{\partial t} + \ham - \hat{\ham} \right ) dt = 0,
\end{equation}
which leads to the following canonical transformation:
\begin{equation}
\bm{q}_\alpha = -\frac{\partial G_3}{\partial \bm{p}_\alpha}, \qquad \bm{P}_\alpha = -\frac{\partial G_3}{\partial \bm{Q}_\alpha}, \qquad \hat{\ham} = \ham + \frac{\partial G_3}{\partial t}.
\label{eqn:transform_2}
\end{equation}
The other two classes of transformation can be derived in a similar way. 

Finally, an important property of a canonical transformation is that it preserves the volume of any element in phase space, i.e., $d\bm{q}d\bm{p}=d\bm{Q}d\bm{P}$ \cite[page 402]{goldstein}. This means that for a change of variables between $(\bm{p},\bm{q})$ and $(\bm{P},\bm{Q})$, the Jacobian of the transformation is unity.

\subsection{A derivation of the stress tensor under equilibrium conditions}
In this section, we use the method of canonical transformations to derive an expression for the Cauchy stress tensor. In continuum mechanics, a body is identified with a regular region of Euclidean space $\euclid$ referred to as the reference configuration. Any point $\bm{X} \in \body$ is referred to as a material point. The body $\body$ is deformed via a smooth, one-to-one mapping $\bm{\varphi}:\euclid \to \euclid$, which maps each $\bm{X} \in \body$ to a point,
\begin{equation}
\bm{x} = \bm{\varphi}(\bm{X}),
\end{equation}
in the deformed configuration,\footnote{We adopt the continuum mechanics convention of denoting variables in the reference configuration with upper-case letters, and variables in the deformed configuration with lower-case letters.} where we have assumed that the deformation is independent of time. The deformation gradient $\dfm$ is defined as
\begin{equation}
\dfm(\bm{X}) = \nabla_{\bm{X}} \bm{\varphi}.
\end{equation}
The mapping $\varphi$ is assumed to satisfy the condition that $\det \bm{F}$ is strictly positive. The Cauchy stress, $\stress$, is defined by \cite{malvern}
\begin{equation}
\label{eqn:cauchy_psi}
\stress(T,\bm{F}) = \frac{1}{\det \dfm}\nabla_{\dfm} \bm{\psi} \dfm^{\T},
\end{equation}
where $\psi(T,\bm{F})$ is the Helmholtz free energy density function relative to the reference configuration. We are only focusing on a conservative elastic body.

A system in thermodynamic equilibrium\footnote{\label{foot:tdequil} A system is said to be in a state of thermodynamic equilibrium when all of its properties are independent of time and all of its intensive properties are independent of position \cite{weiner}. To stress this, the term {\em uniform state of thermodynamic equilibrium} is sometimes used to describe this state.} can by definition only support a uniform state of deformation. Therefore, our material system is deformed via the 
affine mapping\footnote{To understand this mapping, consider a system of $N$ particles with positions $\bm{q}_\alpha$ ($\alpha=1,2,\dots,N$) confined to a parallelepiped container defined by the three linearly independent vectors $\bm{l}_1$, $\bm{l}_2$ and $\bm{l}_3$, which need not be orthogonal. This selection is done for convenience and does not limit the generality of the derivation as explained below.  The position of a particle in the container can be expressed in terms of scaled coordinates $\xi^\alpha_i\in[0,1]$ as
\begin{equation*}
\bm{q}_\alpha = \xi^\alpha_i\bm{l}_i,
\tag{*}
\label{eq:atomposnorm}
\end{equation*}
where Einstein's summation convention is applied to spatial indices.  The deformation of the container is defined relative to a reference configuration where the cell vectors are $\bm{L}_1$, $\bm{L}_2$ and $\bm{L}_3$. The current and reference cell vectors are related through an affine mapping defined by $\dfm$,
\begin{equation*}
\bm{l}_i = \dfm \bm{L}_i.
\tag{**}
\label{eq:cellaffine}
\end{equation*}
Equations \eref{eq:atomposnorm} and \eref{eq:cellaffine} can be combined to relate the position $\bm{q}_\alpha$ of particle $\alpha$ in the deformed configuration with its position in the reference configuration $\bm{Q}_\alpha$,
\begin{equation*}
\bm{q}_\alpha 
= \xi^\alpha_i(\dfm\bm{L}_i)
= \dfm(\xi^\alpha_i\bm{L}_i) 
= \dfm\bm{Q}_\alpha.
\tag{***}
\label{eq:posaffinemap}
\end{equation*}
This is exactly the mapping defined in \eref{eqn:map}. It provides a direct relationship between the positions of particles in the reference configuration and their position in the deformed configuration. Note that the assumed (parallelepiped) shape of the container does not enter into the relation, $\bm{q}_\alpha=\dfm\bm{Q}_\alpha$, which means that this relation holds for a container of any shape.  

It is important to note that \eref{eq:posaffinemap} does {\em not} impose a kinematic constraint that dictates the position of particle $\alpha$ in the deformed configuration based on its position in the reference configuration (as does the Cauchy--Born rule used in multiscale methods \cite{tadmor2009}). We will see later that this will merely be used as a change of variables, where, instead of integrating over the deformed configuration with the variables $\bm{q}$, the integration is carried out over a given reference configuration using the variables $\bm{Q}$. In both cases the same result is obtained. However, by using the referential variables the dependence on the deformation gradient is made explicit.
}
\begin{equation}
\label{eqn:map}
\bm{q}_\alpha = \dfm \bm{Q}_\alpha.
\end{equation}
It is clear that if we enforce this mapping on our system, with no change in the momentum coordinates, then the newly obtained variables will not satisfy Hamilton's equations. Therefore any change of variables should be governed by a canonical transformation. The following generator function provides the desired canonical transformation
\begin{equation}
G_3(\bm{p},\bm{Q}) = -\sum_{\alpha} \bm{p}_{\alpha} \cdot \dfm \bm{Q}_{\alpha}.
\end{equation}
Substituting this generating function into \eref{eqn:transform_2} gives
\begin{equation}
\label{eqn:transform_3}
\bm{q}_\alpha = -\frac{\partial G_3}{\partial \bm{p}_\alpha} = \dfm \bm{Q}_\alpha, \qquad 
\bm{P}_\alpha = -\frac{\partial G_3}{\partial \bm{Q}_\alpha} = \dfm^{\T}\bm{p}_\alpha, \qquad 
\hat{\ham} = \ham.
\end{equation}
The first relation in the above equation is the desired transformation in \eref{eqn:map}. The second relation is the corresponding transformation that the momentum degrees of freedom must satisfy, so that the new set of coordinates $(\bm{Q},\bm{P})$ are canonical. The third relation refers to the Hamiltonian of the system, which is assumed to be given by
\begin{equation}
\ham(\bm{p},\bm{q}) 
= \sum_{\alpha=1}^N \frac{\bm{p}_\alpha\cdot\bm{p}_\alpha}{2m_\alpha} 
+ \pot(\bm{q}_1,\dots,\bm{q}_N),
\label{eq:ham}
\end{equation}
where $\mcal{V}$ denotes the potential energy of the system. Expressed in terms of the reference variables, \eref{eq:ham} becomes
\begin{align}
\hat{\ham}(\bm{P},\bm{Q},\dfm) &= \ham(\bm{p}(\bm{Q},\bm{P},\dfm), \bm{q}(\bm{Q},\bm{P},\dfm)) \notag \\
&= \sum_{\alpha=1}^{N} \frac{\dfm^{-\T} \bm{P}_\alpha \cdot \dfm^{-\T} \bm{P}_\alpha}{2m_\alpha} + \pot(\dfm \bm{Q}_1,\dots,\dfm \bm{Q}_N).
\end{align}
We now proceed to derive the expression for the Cauchy stress tensor using \eref{eqn:cauchy_psi}. The Helmholtz free energy density for the canonical ensemble is given by \cite{huang}
\begin{equation}
\psi(T,\dfm) = -\frac{k_B T \ln{Z}}{V_0},
\label{eqn:psi}
\end{equation}
where $k_B$ is the Boltzmann's constant, $T$ is the absolute temperature, $V_0$ is the volume of the body in the reference configuration, and $Z(T,\dfm)$ is the \emph{partition function} defined as
\begin{equation}
Z(T,\bm{F}) := \frac{1}{N! h^{3N}} \int_{\Gamma_0} e^{-\hat{\ham}/k_B T} \, d\bm{P} d\bm{Q},
\label{eqn:part_func}
\end{equation}
where $h$ is Planck's constant and $\Gamma_0$ denotes the phase space in the reference configuration. With this definition, the statistical mechanics phase average of a function $A(\bm{P},\bm{Q})$ in the canonical ensemble is
\begin{equation}
\langle A \rangle(T,\dfm) = \int_{\Gamma_0} A(\bm{P},\bm{Q}) W_{\rm{c}}(\bm{P},\bm{Q},T,\dfm) \, d\bm{P} \, d\bm{Q},
\label{eqn:phaseave}
\end{equation}
where
\begin{equation}
\label{eqn:w_canonical}
W_{\rm{c}}(\bm{P},\bm{Q},T,\dfm) = \frac{1}{N! h^{3N} Z} e^{-\hat{\ham}(\bm{P},\bm{Q},\dfm)/k_B T}
\end{equation}
is the canonical distribution function.
Substituting \eref{eqn:psi} and \eref{eqn:part_func} into \eref{eqn:cauchy_psi}, we obtain
\begin{equation}
\label{eqn:cauchy_psi_2}
\stress = -\frac{k_B T}{(\det \dfm) V_0 Z} \nabla_{\dfm} Z \dfm^{\T} = \frac{1}{V} \left \langle \nabla_{\dfm} \hat{\ham} \right \rangle \dfm^{\T},
\end{equation}
where in the last step we have used the identity $V = (\det \dfm) V_0$ and where
\begin{align}
\nabla_{\dfm} Z &= \frac{\partial}{\partial \dfm} \left [ \frac{1}{N! h^{3N}} \int_{\Gamma_0} e^{-\hat{\ham}/k_B T} \, d\bm{Q} d\bm{P} \right ] \notag \\
&= -\frac{1}{k_B T N! h^{3N}} \int_{\Gamma_0} \nabla_{\dfm} \hat{\ham} e^{-\hat{\ham}/k_B T} \, d\bm{Q} d\bm{P}.
\end{align}
Next, we compute $\nabla_{\dfm} \hat{\ham}$. In our derivation, we make use of indicial notation and the Einstein summation rule. To accommodate for the spatial indices, we push $\alpha$ representing the particle to the superscript position. Following this adjustment, we have 
\begin{equation}
\label{eqn:dH_dF}
\frac{\partial \hat{\ham}}{\partial F_{iJ}}  = \frac{\partial}{\partial F_{iJ}} \left [ \sum_{\alpha} \frac{p_{k}^{\alpha} p_{k}^{\alpha}}{2m^\alpha} + \mcal{V}(\bm{q}^1,\dots,\bm{q}^N) \right ] = \sum_\alpha \left [ \frac{1}{m^\alpha}\frac{\partial p_{k}^{\alpha}}{\partial F_{iJ}} p_{k}^{\alpha} + \frac{\partial \pot}{\partial q_{k}^\alpha} \frac{\partial q_{k}^{\alpha}}{\partial F_{iJ}} \right ].
\end{equation}
From \eref{eqn:transform_3}, we have
\begin{align}
\frac{\partial q_{k}^{\alpha}}{\partial F_{iJ}} &= \frac{\partial}{\partial F_{iJ}} (F_{kL} Q_{L}^{\alpha}) = \delta_{ik} Q_{J}^{\alpha}, \label{eqn:dr_dF} \\
\frac{\partial p_{k}^{\alpha}}{\partial F_{iJ}} &= \frac{\partial}{\partial F_{iJ}} (\inv{F}_{Lk} P_{L}^{\alpha}) = -\inv{F_{Jk}} \inv{F_{Li}} P_{L}^{\alpha} = -\inv{F_{Jk}}p^\alpha_{i}, \label{eqn:dp_dF}
\end{align}
where in \eref{eqn:dp_dF}, we have used the following identity:
\begin{equation}
\frac{\partial \inv{F_{Lk}}}{\partial F_{iJ}} = -\inv{F_{Li}} \inv{F_{Jk}}.
\end{equation}
Substituting \eref{eqn:dr_dF} and \eref{eqn:dp_dF} into \eref{eqn:dH_dF}, we have
\begin{equation}
\frac{\partial \hat{\ham}}{\partial F_{iJ}} = -\sum_{\alpha} \left [\frac{p_{i}^{\alpha} \inv{F_{Jk}} p_{k}^\alpha}{m^\alpha} + f^{\rm int}_{\alpha,i} Q_{J}^\alpha  \right ],
\end{equation}
where $\bm{f}^{\rm int}_\alpha = -\partial \mcal{V}/\partial \bm{r}_\alpha$ is the internal force, defined in the deformed configuration, on particle $\alpha$.\footnote{There is a subtle point here. Since we are using the canonical ensemble, the Hamiltonian $\ham$ neglects the interaction term of the system with the surrounding ``heat bath''. This means that the potential energy $\pot$ in $\ham$ only includes the {\em internal} energy of the system and, therefore, its derivative with respect to the position of particle $\alpha$ gives the force $\bm{f}^{\rm int}_{\alpha}$ on this particle due to its interactions with other particles in the system.} In direct notation, we have
\begin{equation}
\nabla_{\dfm} \hat{\ham} = -\sum_\alpha \left [\frac{\bm{p}_\alpha \otimes \inv{\dfm} \bm{p}_\alpha}{m_\alpha} + \bm{f}^{\rm int}_\alpha \otimes \bm{Q}_\alpha \right ].
\end{equation}
Substituting the above equation into \eref{eqn:cauchy_psi_2} and using \eref{eqn:transform_3}, we obtain an expression for the Cauchy stress:
\begin{equation}
\stress(T,\dfm) = -\frac{1}{V} \sum_\alpha \left \langle \frac{\bm{p}_\alpha \otimes \bm{p}_\alpha}{m_\alpha}  + \bm{f}^{\rm int}_\alpha \otimes \bm{q}_\alpha \right \rangle,
\label{eqn:cauchy_canonical}
\end{equation}
where the phase averaging is now being performed with respect to the variables $\bm{p}$ and $\bm{q}$. The switch from phase averaging over $\bm{P}$ and $\bm{Q}$ in \eref{eqn:phaseave} to $\bm{p}$ and $\bm{q}$ above can be made because canonical transformations preserve the volume element in phase space as explained at the end of \sref{sec:can_trans}.

The expression in \eref{eqn:cauchy_canonical} for the Cauchy stress tensor is called the \emph{virial stress}. A simpler derivation of the virial stress, based on time averages, is given in Appendix \ref{ch:virial}. Although, the derivation here made use of the canonical ensemble, it is expected to apply to any ensemble in the thermodynamic limit (see footnote~\ref{foot:tdlimit} on page~\pageref{foot:tdlimit}) where all ensembles are equivalent. Continuum mechanics also tells us that the Cauchy stress tensor is symmetric, something that is not evident from the above equation. Discussion of the symmetry of the stress tensor, which hinges on an important property of $\bm{f}^{\rm int}_\alpha$, is postponed to \sref{ch:compare}.

\medskip
The virial stress defined above corresponds to the macroscopic stress tensor only under conditions of thermodynamic equilibrium in the thermodynamic limit. We now show that this expression for the stress tensor, as well as all other expressions in common use, can be derived as limiting cases of a more general formulation which begins with the Irving--Kirkwood--Noll procedure. We refer to this as the ``unified framework'' for the stress tensor.

\section{Continuum fields as phase averages}
\label{ch:phase}
In this section, we discuss the Irving and Kirkwood procedure \cite{ik1950}, which laid the foundation for the microscopic definition of continuum fields for non-equilibrium systems. This work was later extended by Walter Noll \cite{noll1955}\footnote{An English translation of this article appears in the current issue of the {\em Journal of Elasticity}.}, who showed how closed-form analytical solutions can be obtained for the definition of certain continuum fields, which otherwise involved a non-rigorous\footnote{The derivation is non-rigorous in the sense that expressing the stress tensor as a series expansion is only possible when the probability density function, which is used in the derivation, is an analytic function of the spatial variables \cite{noll1955}.} series expansion of the Dirac delta distribution in the original procedure. We refer to the procedure proposed by Noll in \cite{noll1955} as the \emph{Irving--Kirkwood--Noll procedure}.  The derivation presented in this section largely follows that of Noll \cite{noll1955}, but extends it to more general atomistic models. 

Consider a system $\mathcal{M}$ modeled as a collection of $N$ point masses/particles, each particle referred to as $\alpha$ $(\alpha=1,2,\ldots,N)$. We use the terms ``particle'' and ``atom'' interchangeably. The position, mass and velocity of particle $\alpha$ are denoted by $\bm{x}_\alpha$, $m_\alpha$ and $\bm{v}_\alpha$, respectively. The complete microscopic state of the system  is known, at any instant of time, from the knowledge of position and velocity of each particle in $\real{3}$. Hence, the state of the system  at time $t$, may be represented by a point $\bm{\Xi}(t)$ in a $6N$-dimensional phase space\footnote{The usual convention is to represent the phase space via positions and momenta of the particles. For convenience, in this section, we represent the phase space via positions and velocities of the particles.}. Let $\Gamma$ denote the phase space. Therefore any point $\bm{\Xi}(t) \in \Gamma$, can be represented as,
\begin{align}
\bm{\Xi}(t) &= (\bm{x}_1(t),\bm{x}_2(t),\dots,\bm{x}_N(t);\bm{v}_1(t),\bm{v}_2(t),\dots,\bm{v}_N(t)) \notag \\
&=: (\bm{x}(t);\bm{v}(t)).  \label{eqn:define_X}
\end{align}
In reality, the microscopic state of the system is never known to us, and the only observables identified are the macroscopic fields as defined in continuum mechanics. We identify the continuum fields with macroscopic observables obtained in a two-step process: (1) a pointwise field is obtained as a statistical mechanics phase average; (2) a macroscopic field is obtained as a spatial average over the pointwise field.  The phase averaging in step (1) is done with respect to a probability density function $W:\Gamma \times \mathbb{R}^+ \rightarrow \mathbb{R}^+$ of class $C^1$ defined on all phase space for all $t$ ($W_{\rm{c}}$, defined in \eref{eqn:w_canonical}, is an example of a stationary (time-independent) probability density function defined for the canonical ensemble). The explicit dependence of $W$ on time $t$, indicates that our system need not be in thermodynamic equilibrium.

As discussed in \sref{ch:canonical}, the evolution of $\bm{\Xi}(t)$ in the phase space is given by the following set of $2N$ first-order equations (Hamilton's equations of \eref{eqn:ham1}--\eref{eqn:ham2}):
\begin{subequations} \label{eqn:hamilton}
\begin{align}
\dot{\bm{p}} &= -\nabla_{\bm{x}}\ham, \\
\dot{\bm{x}} &= \nabla_{\bm{p}}\ham, 
\end{align}
\end{subequations}
where $\bm{p} := (\bm{p}_1,\bm{p}_2,\dots,\bm{p}_N)$, $\bm{p}_\alpha$ denotes the momentum of each particle, and $\ham(\bm{p},\bm{x})$ is the Hamiltonian of the system.

The basic idea behind the original Irving and Kirkwood procedure is to prescribe/derive microscopic definitions for continuum fields, such that they are consistent with the balance laws of mass, momentum and energy. To arrive at these definitions, we repeatedly use the following theorem, commonly referred to as \emph{Liouville's theorem}, which relates to the conservation of volume in phase space.  

As a  system evolves, the phase space $\Gamma$ is mapped into itself at every instant of time, and this mapping is governed by \eref{eqn:hamilton}. If $g_t$ denotes this mapping, then Liouville's theorem essentially says that for any subset $U$ of $\Gamma$, the volume of $U$ remains invariant under the mapping $g_t$. This can be be formally stated as,
\spnewtheorem*{liouville}{Liouville's Theorem}{\bfseries}{\itshape}
\begin{liouville} 
For any $U \subseteq \Gamma$, volume is preserved under the one-parameter group of transformations of phase space, $g_t:U \rightarrow \Gamma$, given by the mapping 
\[
(\bm{x}(0),\bm{p}(0)) \mapsto (\bm{x}(t),\bm{p}(t)),
\]
where $\bm{x}(t)$ and $\bm{p}(t)$ are solutions of the Hamilton's system of equations {\rm \eref{eqn:hamilton}}, i.e., \\
\begin{equation}
\label{eqn:volume}
\vol(U) = \vol(g_tU).
\end{equation}
\end{liouville}
\begin{proof}
Let $\dot{\overline{\vol(g_tU)}}$ denote the material time derivative of $\vol(g_tU)$ in the sense that $\bm{\Xi}(0)$ is held fixed while performing this differentiation. Then we have,
\[
\dot{\overline{\vol(g_tU)}} = \dot{\overline{\int_{g_tU} d\bm{\Xi}(t)}}
= \int_{U} (\dot{\overline{\det \bm{F}}}) d \bm{\Xi}_0,
\]
where $\bm{F}(\bm{\Xi}_0,t):= \nabla_{\bm{\Xi}_0} \bm{\Xi}(\bm{\Xi}_0,t)$, $\bm{\Xi}(\bm{\Xi}_0,t) = g_t(\bm{\Xi_0})$ and $\bm{\Xi}_0 = \bm{\Xi}(0)$. Using the fact that
\[
\dot{\overline{\det \bm{F}}} = (\det \bm{F}) \tr(\dot{\bm{F}}\bm{F}^{-1}),
\]
we obtain
\begin{equation}
\dot{\overline{\vol(g_tU)}} = \int_{U} (\det \bm{F}) \tr(\dot{\bm{F}}\bm{F}^{-1}) d\bm{\Xi}_0.
\label{eqn:liouville_proof_1}
\end{equation}
Let 
\begin{equation}
\dot{\bm{\Xi}}(\bm{\Xi}) := \left. \frac{d\bm{\Xi}}{dt}  \right |_{\bm{\Xi}_0 = g_t^{-1}(\bm{\Xi})}.
\end{equation}
From chain rule, we have
\begin{equation}
\nabla \dot{\bm{\Xi}} = \left. \frac{d (\nabla \bm{\Xi})}{dt} \right |_{\bm{\Xi}_0 = g_t^{-1}(\bm{\Xi})} \nabla_{\bm{\Xi}} \bm{\Xi}_0 =  \dot{\bm{F}}\bm{F}^{-1}.
\end{equation}
Therefore $\divr \, \dot{\bm{\Xi}} = \tr(\dot{\bm{F}}\bm{F}^{-1})$. Equation \eref{eqn:liouville_proof_1} can now be rewritten as
\begin{equation}
\dot{\overline{\vol(g_tU)}} = \int_{U} (\det \bm{F}) (\divr \, \dot{\bm{\Xi}})\mid_{\bm{\Xi}(t) = g_t(\bm{\Xi}_o)} d\bm{\Xi}_0.
\end{equation}
But from \eref{eqn:define_X} and \eref{eqn:hamilton} we also have, 
\[
\divr \dot{\bm{\Xi}} = \divr_{\bm{x}} \dot{\bm{x}} + \divr_{\bm{p}} \dot{\bm{p}} = \divr_{\bm{x}} (\nabla_{\bm{p}} \ham) - \divr_{\bm{p}}(\nabla_{\bm{x}} \ham) = 0.
\]
Therefore $\dot{\overline{\vol(g_tU)}} = 0$ for arbitrary $t$. Thus \eref{eqn:volume} holds.\qed   
\end{proof}

Let $W(\bm{\Xi};t)$ denote the probability density function defined on $g_t(\Gamma)$. Hence, we have
\begin{equation}
\int_{g_tU} W(\bm{\Xi}(t);t) d \bm{\Xi}(t) = \int_{U} W(\bm{\Xi}_0;0) (\det \dfm) d \bm{\Xi}_0.
\end{equation}
As a consequence of Liouville's theorem, we have, $\det \dfm =1$. Therefore
\begin{equation}
\label{eqn:corr}
\frac{d}{dt}\int_{g_tU} W(\bm{\Xi}(t);t) d \bm{\Xi}(t) = 0.
\end{equation}
Since \eref{eqn:corr} holds for all $U \subseteq \Gamma$, we have $\dot{W}(\bm{\Xi(t)};t) = 0$. Hence, the time evolution of the probability density function is given by
\begin{equation}
\label{eqn:liouville}
\frac{\partial W}{\partial t} + \sum_{\alpha=1}^{N} \left [ \bm{v}_\alpha \cdot \nabla_{\bm{x}_\alpha}W + \dot{\bm{v}}_\alpha \cdot \nabla_{\bm{v}_\alpha}W \right ] = 0.
\end{equation}
The above equation can be rewritten as
\begin{equation}
\frac{\partial W}{\partial t} + \sum_{\alpha=1}^{N} \left [ \bm{v}_\alpha \cdot \nabla_{\bm{x}_\alpha}W - \frac{\nabla_{\bm{x}_\alpha} \pot }{m_\alpha} \cdot \nabla_{\bm{v}_\alpha}W \right ] = 0,
\label{eqn:useful_liouville}
\end{equation}
where, as before, $\pot(\bm{x}_1,\bm{x}_2,\dots,\bm{x}_N)$ denotes the potential energy of the system. Equation~\eref{eqn:useful_liouville} is called \emph{Liouville's equation}.

\subsection{Phase averaging}
\label{sec:phase}
Under the Irving--Kirkwood--Noll procedure, pointwise fields are defined as phase averages. This phase averaging is expressed via weighted marginal densities. For example, the pointwise mass density field is defined as
\begin{equation}
\rho(\bm{x},t) := \sum_{\alpha} m_\alpha \int_{\real{{3N}} \times \real{{3N}}} W \delta (\bm{x}_\alpha - \bm{x}) \, d\bm{x} d\bm{v},
\label{eqn:define_density_delta}
\end{equation}
where the integral represents a marginal density defined on $\real{3}$, $\delta$ denotes the Dirac delta distribution, and $\sum_\alpha$ denotes summation from $\alpha=1$ to $N$.  To avoid the Dirac delta distribution and for greater clarity we adopt Noll's notation as originally used in \cite{noll1955}. Hence \eref{eqn:define_density_delta} can be rewritten as
\begin{align}
\rho(\bm{x},t) &= \sum_{\alpha} m_\alpha \int W \, d\bm{x}_1 \dots d\bm{x}_{\alpha-1} d\bm{x}_{\alpha+1}\dots d\bm{x}_N d\bm{v} \notag \\
&=: \sum_{\alpha} m_\alpha \left \langle W \mid \bm{x}_\alpha = \bm{x} \right \rangle,
\label{eqn:define_density}
\end{align}
where $\left \langle W \mid \bm{x}_\alpha = \bm{x} \right \rangle$ denotes an integral of $W$ over all its arguments except $\bm{x}_\alpha$ and $\bm{x}_\alpha$ is substituted with $\bm{x}$.
Now consider the continuum velocity field. Unlike the definition of pointwise density field, which appears unambiguous, the pointwise velocity field can be defined in different ways. It may seem more natural to define the continuum velocity in an analogous fashion to the density field, i.e.,
\begin{equation}
\bm{v}(\bm{x},t) = \frac{\sum_{\alpha} \left \langle W \bm{v}_\alpha \mid \bm{x}_\alpha = \bm{x} \right \rangle}{\sum_{\alpha} \left \langle W \mid \bm{x}_\alpha = \bm{x} \right \rangle}.
\label{eqn:define_velocity_alt}
\end{equation}
Alternatively, the pointwise velocity field can be defined via the momentum density field, $\bm{p}(\bm{x},t)$, as follows:
\begin{align}
\bm{p}(\bm{x},t) &:= \sum_{\alpha} m_\alpha \left \langle W \bm{v}_\alpha \mid \bm{x}_\alpha = \bm{x} \right \rangle,\label{eqn:define_mom_density}\\ 
\bm{v}(\bm{x},t) &:= \frac{\bm{p}(\bm{x},t)}{\rho(\bm{x},t)}. \label{eqn:define_velocity}
\end{align}
Note that definitions \eref{eqn:define_velocity_alt} and \eref{eqn:define_velocity} are equivalent for a single species material, but are not so in general. The definition given by \eref{eqn:define_velocity} is the one used in practice. There are two reasons for this. First, the definition in \eref{eqn:define_velocity} makes more physical sense since, following spatial averaging, it associates the continuum velocity with the velocity of the center of mass of the system of particles. Second, the definition in \eref{eqn:define_velocity} satisfies the continuity equation as shown in \sref{sec:continuity}, whereas \eref{eqn:define_velocity_alt} does not.\footnote{It would be interesting to explore how the equation of continuity fails for the definition in \eref{eqn:define_velocity_alt} by identifying the regions that act as sinks and sources. This is difficult to do for a general $N$ particle system because the continuity equation quickly becomes unwieldy. Even for the much simpler case of a two-particle system, the answer is not trivial. A quick examination shows that the distribution of sinks and sources depends not only on the ratio of the masses but also on the probability density function $W$.}

\subsection{Regularity assumptions for the probability density function}
It is clear from the definitions in \eref{eqn:define_density}, \eref{eqn:define_mom_density} and \eref{eqn:define_velocity} that the integrals in these equations converge under appropriate decay conditions on $W$. The following two conditions are sufficient for the convergence of all the integrals and the validity of the results in this section \cite{noll1955}:
\begin{enumerate}
\item There exists a $\delta>0$ such that the function
\label{cond_1}
\begin{equation}
W(\bm{\Xi};t) \prod_{\alpha=1}^{N} \vnorm{\bm{x}_\alpha}^{3+\delta} \prod_{\beta=1}^{N} \vnorm{\bm{v}_\beta}^{6+\delta}
\label{eqn:w_decay}
\end{equation}
and its first derivatives are bounded by a constant that only depends on time.
\item $\pot(\bm{x}_1,\bm{x}_2,\cdots \bm{x}_N)$ is a bounded $C^1$ function defined on the phase space, and having bounded first derivatives.\footnote{If any two particles overlap, we would normally expect $\pot \to \infty$. By specifying additional decay conditions for $W$, the case of unbounded $\pot$ can be handled. For simplicity, we assume $\pot$ to be bounded.}
\label{cond_2}
\end{enumerate}
Conditions (\ref{cond_1}) and (\ref{cond_2}) ensure the convergence of all the integrals considered in this section and swapping of integration and differentiation. Furthermore, let $\bm{G}(\bm{\Xi};t)$ be any vector or tensor-valued function of class $C^1$ defined on the phase space for all $t$, and which, for suitable functions $g(t)$ and $h(t)$, satisfies the condition
\begin{equation}
\sup_{\bm{x}_1 \in \real{3},\bm{x}_2 \in \real{3},\cdots,\bm{x}_N \in \real{3}}(\vnorm{\bm{G}},\vnorm{\divr_{\bm{v}_\alpha} \bm{G}},\vnorm{\divr_{\bm{x}_\alpha}\bm{G}}) < g(t) \prod_{\beta=1}^{N}\vnorm{\bm{v}_\beta}^3 +h(t),
\end{equation}
where $\vnorm{\cdot}$ refers to the norm defined through the inner product. Since the space of all tensors has a natural inner product defined as
\begin{equation}
\bm{S} : \bm{T} = \tr(\bm{S}^{\T} \bm{T}),
\end{equation}
we have $\vnorm{\bm{S}} = \sqrt{\bm{S} : \bm{S}}$.
Under these conditions on $\bm{G}(\bm{\Xi};t)$, we have\footnote{If $\bm{G}$ is a second-order tensor or higher, then the dot product indicates tensor operating on a vector. Note that in \eref{eqn:reg}, in the interest of brevity, we are breaking our notation of denoting a second-order tensor operating on a vector by juxtaposition.}
\begin{subequations}
\label{eqn:reg}
\begin{align}
\int_{\real{3}}\bm{G} \cdot \nabla_{\bm{x}_\alpha}W \, d\bm{x}_\alpha &= -\int_{\real{3}} W \divr_{\bm{x}_\alpha} \bm{G} \, d\bm{x}_\alpha,\label{eqn:reg_1} \\ 
\int_{\real{3}}\bm{G} \cdot \nabla_{\bm{v}_\alpha}W \, d\bm{v}_\alpha &= -\int_{\real{3}} W \divr_{\bm{v}_\alpha} \bm{G} \, d\bm{v}_\alpha. \label{eqn:reg_2}
\end{align}
\end{subequations}
The above identities are repeatedly used in deriving the equation of continuity and the equation of motion in the following sections.

\subsection{Equation of continuity}
\label{sec:continuity}
Let us demonstrate that the pointwise fields defined in \sref{sec:phase} satisfy the equation of continuity. The equation of continuity from continuum mechanics is given by \cite{malvern}
\begin{equation}
\label{eqn:continuity}
\frac{\partial \rho}{\partial t} + \divr_{\bm{x}}(\rho \bm{v}) = 0.
\end{equation}
From \eref{eqn:define_density} we have
\[
\frac{\partial \rho} {\partial t} (\bm{x},t) = \sum_{\alpha}m_\alpha \left \langle \left. \frac{\partial W}{\partial t} \right | \bm{x}_\alpha = \bm{x} \right \rangle.
\]
Using Liouville's equation in \eref{eqn:useful_liouville}, we have
\[
\frac{\partial \rho} {\partial t} (\bm{x},t) = \sum_\alpha m_\alpha \left \langle \left. \sum_\beta \left ( -\bm{v}_\beta \cdot \nabla_{\bm{x}_\beta}W + \frac{\nabla_{\bm{x}_\beta} \pot}{m_\beta}\cdot \nabla_{\bm{v}_\beta}W \right ) \right | \bm{x}_\alpha=\bm{x} \right \rangle.
\]
Now, consider the summand on the right-hand side of the above equation for a fixed $\alpha$. From \eref{eqn:reg_2}, it is clear that $\left \langle \left. \frac{\nabla_{\bm{x}_\beta} \pot}{m_\beta}\cdot \nabla_{\bm{v}_\beta}W \right | \bm{x}_\alpha=\bm{x} \right \rangle=0$, for $\beta=1,2,\cdots N$, and from \eref{eqn:reg_1}, we also have $\left \langle \bm{v}_\beta \cdot \nabla_{\bm{x}_\beta}W \mid \bm{x}_\alpha=\bm{x} \right \rangle = 0$, for $\beta \ne \alpha$. Therefore the above equation simplifies to
\[
\frac{\partial \rho} {\partial t} (\bm{x},t) = -\sum_{\alpha} m_\alpha \left \langle \bm{v}_\alpha \cdot \nabla_{\bm{x}_\alpha} W\mid \bm{x}_\alpha =\bm{x} \right \rangle.
\]
Using the identity,
\begin{equation}
\label{eqn:id1}
\divr_{\bm{x}}(a \bm{w}) = \nabla_{\bm{x}} a \cdot \bm{w},
\end{equation}
where $a(\bm{x})$ is any $C^1$ scalar function of $\bm{x}$, and $\bm{w}$ is any vector independent of $\bm{x}$, we obtain
\[
\frac{\partial \rho} {\partial t} (\bm{x},t) = -\sum_{\alpha} m_\alpha \divr_{\bm{x}} \left \langle W \bm{v}_\alpha \mid \bm{x}_\alpha = \bm{x} \right \rangle.
\]
Using \eref{eqn:define_mom_density} and \eref{eqn:define_velocity} for the definition of the pointwise momentum density field, we have
\[
\frac{\partial \rho} {\partial t} (\bm{x},t) + \divr_{\bm{x}} (\rho \bm{v}) = 0,
\]
which is the continuity equation. We have established that the definitions given in \eref{eqn:define_density_delta} and \eref{eqn:define_mom_density} identically satisfy conservation of mass.

\subsection{Equation of Motion}
\label{sec:s_motion}
The equation of motion from continuum mechanics is given by \cite{malvern}
\begin{eqnarray}
\label{eqn:motion}
\frac{\partial(\rho \bm{v})}{\partial t} + \divr_{\bm{x}}(\rho \bm{v} \otimes \bm{v}) = \divr_{\bm{x}}\bm{\sigma} + \bm{b}.
\end{eqnarray}
Here we identify $\bm{\sigma}$ with the pointwise stress tensor.  From \eref{eqn:define_mom_density}, we have
\[
\frac{\partial \bm{p}}{\partial t}(\bm{x},t) = \sum_{\alpha} m_\alpha \left \langle \left. \bm{v}_\alpha \frac{\partial W}{\partial t} \right | \bm{x}_\alpha = \bm{x} \right \rangle. 
\]
Again, using \eref{eqn:useful_liouville} we obtain,
\begin{align}
\frac{\partial \bm{p}}{\partial t} (\bm{x},t) &=  \sum_{\alpha} m_\alpha \left \langle \left. \bm{v}_\alpha \sum_\beta \left (-\nabla_{\bm{x}_\beta} W \cdot \bm{v}_\beta + \frac{\nabla_{\bm{x}_\beta} \pot}{m_\beta} \cdot \nabla_{\bm{v}_\beta} W \right ) \right | \bm{x}_\alpha = \bm{x} \right \rangle \notag \\
&= \sum_\alpha m_\alpha \sum_\beta \left \langle \left. -\left ( \bm{v}_\alpha \otimes \bm{v}_\beta \right ) \nabla_{\bm{x}_\beta} W + \left ( \bm{v}_\alpha \otimes \frac{\nabla_{\bm{x}_\beta} \pot}{m_\beta} \right ) \nabla_{\bm{v}_\beta} W \right |  \bm{x}_\alpha = \bm{x} \right \rangle.
\label{eqn:motion_alt}
\end{align}
Now, consider the summand on the right-hand side of the above equation for fixed $\alpha$ and $\beta$. Using \eref{eqn:reg_1}, we have $\left \langle ( \bm{v}_\alpha \otimes \bm{v}_\beta ) \nabla_{\bm{x}_\beta} W \mid \bm{x}_\alpha = \bm{x} \right \rangle = \bm{0}$, for $\beta \ne \alpha$. From \eref{eqn:reg_2}, we have $\left \langle (\bm{v}_\alpha \otimes \nabla_{\bm{x}_\beta} \pot ) \nabla_{\bm{v}_\beta}W \mid \bm{x}_\alpha = \bm{x} \right \rangle = \bm{0}$, for $\beta \ne \alpha$, and for $\beta = \alpha$, we have
\[
\left \langle (\bm{v}_\alpha \otimes \nabla_{\bm{x}_\alpha} \pot ) \nabla_{\bm{v}_\alpha} W\mid \bm{x}_\alpha = \bm{x} \right \rangle = -\left \langle \nabla_{\bm{x}_\alpha} \pot W \mid \bm{x}_\alpha = \bm{x} \right \rangle,
\]
using the fact that $\divr_{\bm{u}}(\bm{u} \otimes \bm{w}) = \bm{w}$, for any vector $\bm{u}$ and for any vector $\bm{w}$ independent of $\bm{u}$. Therefore \eref{eqn:motion_alt} simplifies to
\begin{equation}
\frac{\partial \bm{p}}{\partial t} (\bm{x},t) = -\sum_{\alpha} m_\alpha \left \langle (\bm{v}_\alpha \otimes \bm{v}_\alpha) \nabla_{\bm{x}_\alpha} W \mid \bm{x}_\alpha = \bm{x} \right \rangle   -  \sum_{\alpha} \left \langle W \nabla_{\bm{x}_\alpha} \pot \mid \bm{x}_\alpha = \bm{x} \right \rangle.  \label{eqn:motion_1}
\end{equation}
Using the identity,
\begin{equation}
\divr_{\bm{x}}(a \bm{T})= \bm{T}\nabla_{\bm{x}}a,
\label{eqn:id2}
\end{equation}
where $a(\bm{x})$ is any $C^1$ scalar function of $\bm{x}$, and $\bm{\bm{T}}$ is any tensor independent of $\bm{x}$, we can rewrite \eref{eqn:motion_1} as
\begin{equation}
\frac{\partial \bm{p}}{\partial t} (\bm{x},t) = -\divr_{\bm{x}} \sum_{\alpha} m_\alpha \left \langle (\bm{v}_\alpha \otimes \bm{v}_\alpha) W \mid \bm{x}_\alpha = \bm{x} \right \rangle   -  \sum_{\alpha} \left \langle W \nabla_{\bm{x}_\alpha} \pot \mid \bm{x}_\alpha = \bm{x} \right \rangle. \label{eqn:motion_2}
\end{equation}
Now, note that the term $ \bm{v}_\alpha \otimes \bm{v}_\alpha $ can be written as \\
\begin{align}
\bm{v}_\alpha \otimes \bm{v}_\alpha &= (\bm{v}_\alpha - \bm{v}) \otimes (\bm{v}_\alpha - \bm{v}) + \bm{v} \otimes \bm{v}_\alpha + \bm{v}_\alpha \otimes \bm{v} - \bm{v} \otimes \bm{v} \notag \\
&= \bm{v}_\alpha^{\rm{rel}} \otimes \bm{v}_\alpha^{\rm{rel}} + \bm{v} \otimes \bm{v}_\alpha + \bm{v}_\alpha \otimes \bm{v} - \bm{v} \otimes \bm{v},  \label{eqn:id3}
\end{align}
where $\bm{v}_\alpha^{\rm{rel}}$ is the velocity of particle $\alpha$ relative to the pointwise velocity field.  Consider the first term on the right-hand side of \eref{eqn:motion_2}. Substituting \eref{eqn:id3} into this expression we have,
\begin{align}
&-\divr_{\bm{x}} \sum_{\alpha} m_\alpha \left \langle (\bm{v}_\alpha \otimes \bm{v}_\alpha) W \mid \bm{x}_\alpha = \bm{x} \right \rangle \notag\\
&= -\sum_{\alpha} m_\alpha \divr_{\bm{x}} \left \langle (\bm{v}_\alpha ^{\rm{rel}} \otimes \bm{v}_\alpha ^{\rm{rel}}) W \mid \bm{x}_\alpha = \bm{x} \right \rangle - \divr_{\bm{x}} \sum_\alpha \big [ \bm{v} \otimes m_\alpha \left \langle \bm{v}_\alpha W \mid \bm{x}_\alpha = \bm{x} \right \rangle  \notag\\
&\qquad + m_\alpha \left \langle \bm{v}_\alpha W \mid \bm{x}_\alpha = \bm{x} \right \rangle \otimes \bm{v} - m_\alpha \left \langle W \mid \bm{x}_\alpha = \bm{x} \right \rangle \bm{v} \otimes \bm{v} \big ] \notag \\ 
&=-\divr_{\bm{x}} \sum_{\alpha} m_\alpha \left \langle (\bm{v}_\alpha^{\rm{rel}} \otimes \bm{v}_\alpha^{\rm{rel}}) W \mid \bm{x}_\alpha = \bm{x} \right \rangle - \divr_{\bm{x}}(\rho \bm{v} \otimes \bm{v}), \label{eqn:motion_3}
\end{align}
where we have used \eref{eqn:define_density}, \eref{eqn:define_mom_density} and \eref{eqn:define_velocity} in the last step.
Substituting \eref{eqn:motion_3} into \eref{eqn:motion_2}, we obtain
\begin{align}
\frac{\partial \bm{p}}{\partial t} (\bm{x},t) + \divr_{\bm{x}}(\rho \bm{v} \otimes \bm{v}) = &-\sum_{\alpha} m_\alpha \divr_{\bm{x}} \left \langle (\bm{v}_\alpha^{\rm{rel}} \otimes \bm{v}_\alpha^{\rm{rel}}) W \mid \bm{x}_\alpha = \bm{x} \right \rangle \notag \\
&- \sum_{\alpha} \left \langle W \nabla_{\bm{x}_\alpha} \pot \mid \bm{x}_\alpha = \bm{x} \right \rangle. \label{eqn:motion_4}
\end{align}
The left-hand sides of \eref{eqn:motion_4} and \eref{eqn:motion} are identical. Therefore, the right-hand sides must also be equal. Hence
\begin{equation}
\divr_{\bm{x}}\bm{\sigma} + \bm{b} = -\sum_{\alpha} m_\alpha \divr_{\bm{x}} \left \langle (\bm{v}_\alpha^{\rm{rel}} \otimes \bm{v}_\alpha ^{\rm{rel}}) W \mid \bm{x}_\alpha = \bm{x} \right \rangle - \sum_{\alpha} \left \langle W \nabla_{\bm{x}_\alpha} \pot \mid \bm{x}_\alpha = \bm{x} \right \rangle.
\label{eqn:motion_5}
\end{equation}
To proceed, we divide the potential energy $\pot(\bm{x}_1,\bm{x}_2,\dots,\bm{x}_N)$ into two parts: 
\begin{enumerate}
\item An \emph{external} part, $\pot_{\rm{ext}}$, associated with long-range interactions such as gravity or electromagnetic fields.
\item An \emph{internal} part, $\pot_{\rm{int}}$, associated with short-range particle interactions. 
\end{enumerate}
It is natural to associate $\pot_{\rm{ext}}$ with the body force field $\bm{b}$ in \eref{eqn:motion_5}. We therefore define $\bm{b}(\bm{x},t)$ as
\begin{equation}
\label{eqn:body}
\bm{b}(\bm{x},t) := - \sum_{\alpha} \left \langle W \nabla_{\bm{x}_\alpha} \pot_{\rm{ext}} \mid \bm{x}_\alpha = \bm{x} \right \rangle.
\end{equation}
Substituting \eref{eqn:body} into \eref{eqn:motion_5}, we have
\begin{equation}
\divr_{\bm{x}} \bm{\sigma} = -\sum_{\alpha} m_\alpha \divr_{\bm{x}} \left \langle (\bm{v}_\alpha ^{\rm{rel}} \otimes \bm{v}_\alpha^{\rm{rel}}) W \mid \bm{x}_\alpha = \bm{x} \right \rangle - \sum_{\alpha} \left \langle W \nabla_{\bm{x}_\alpha} \pot_{\rm{int}}  \mid \bm{x}_\alpha = \bm{x} \right \rangle.
\label{eqn:motion_6}
\end{equation}
From \eref{eqn:motion_6}, we see that the pointwise stress tensor has two contributions:
\begin{equation}
\label{eqn:stress_split}
\bm{\sigma}(\bm{x},t) = \bm{\sigma}_{\rm{k}}(\bm{x},t) + \bm{\sigma}_{\rm{v}}(\bm{x},t),
\end{equation}
where $\bm{\sigma}_{\rm{k}}$ and $\bm{\sigma}_{\rm{v}}$ are, respectively, the \emph{kinetic} and \emph{potential} parts of the pointwise stress.  The kinetic part is given by
\begin{equation}
\bm{\sigma}_{\rm{k}}(\bm{x},t) = -\sum_{\alpha} m_\alpha \left \langle (\bm{v}_\alpha^{\rm{rel}} \otimes \bm{v}_\alpha ^{\rm{rel}}) W \mid \bm{x}_\alpha = \bm{x} \right \rangle.
\label{eqn:stress_kinetic}
\end{equation}
It is evident that the kinetic part of the stress tensor is symmetric. The presence of a kinetic contribution to the stress tensor appears at odds with the continuum definition of stress that is stated solely in terms of the forces acting between different parts of the body. This discrepancy has led to controversy in the past about whether the kinetic term belongs in the stress definition \cite{zhou2003}. The confusion is related to the difference between absolute velocity and relative velocity defined in \eref{eqn:id3} \cite{tsai1979}. The kinetic stress reflects the momentum flux associated with the vibrational kinetic energy portion of the internal energy.

Continuing with \eref{eqn:motion_6}, the potential part of the stress must satisfy the following differential equation:
\begin{equation}
\divr_{\bm{x}} \bm{\sigma}_{\rm{v}}(\bm{x},t) = \sum_{\alpha} \left \langle W \bm{f}^{\rm int}_\alpha \mid \bm{x}_\alpha = \bm{x} \right \rangle,
\label{eqn:stress_force_differential}
\end{equation}
where
\begin{equation}
\label{eqn:def_fi}
\bm{f}^{\rm int}_\alpha := -\nabla_{\bm{x}_\alpha} \pot_{\rm{int}},
\end{equation}
is the force on particle $\alpha$ due to internal interactions.  Equation~\eref{eqn:stress_force_differential} needs to be solved in order to obtain an explicit form for $\bm{\sigma}_{\rm{v}}$. In the original paper of Irving and Kirkwood \cite{ik1950}, this was done by applying a Taylor expansion to the Dirac delta distribution appearing in the right-hand side of the equation. In contrast, Noll showed that a closed-form solution for $\bm{\sigma}_{\rm{v}}$ can be obtained by recasting the right-hand side in a different form and applying a lemma proved in \cite{noll1955}.  We proceed with Noll's approach, except we place no restriction on the nature of the interatomic potential energy $\pot_{\rm{int}}$. The potential energy considered in \cite{ik1950} and \cite{noll1955} is limited to pair potentials. 

\subsubsection*{General interatomic potentials}
In general, the internal part of the potential energy, also called the \emph{interatomic potential energy}, depends on the positions of all particles in the system:
\begin{equation}
\pot_{\rm{int}} = \potfisher_{\rm int}(\bm{x}_1,\bm{x}_2,\dots,\bm{x}_N),
\end{equation}
where the ``hat'' indicates that the functional dependence is on absolute particle positions (as opposed to distances later on). We assume that $\potfisher_{\rm int}:\mathbb{R}^{3N} \to \mathbb{R}$ is a continuously differentiable function.\footnote{Note that this assumption may fail in systems undergoing first-order magnetic or electronic phase transformations.} This function must satisfy the following invariance principle:
\label{page:invariance}
\begin{quote}
The internal energy of a material system is invariant with respect to the Euclidean group $\mcal{G} := \{\bm{x} \mapsto \bm{Q}\bm{x}+\bm{c} \mid \bm{x} \in \real{3},\bm{Q} \in O(3),\bm{c} \in \real{3}\}$, where $O(3)$ denotes the full orthogonal group.
\end{quote}
To exploit this invariance, let us consider the action of $\mcal{G}$ on $\mathbb{R}^{3N}$, i.e., the action of any combination of translation and rotation (proper or improper), which is represented by an element $g:\bm{x} \mapsto \bm{Q}\bm{x}+\bm{c}$ in $\mcal{G}$, on any configuration of $N$ particles represented by a vector $(\bm{x}_1,\dots,\bm{x}_N) \in \mathbb{R}^{3N}$:
\begin{align}
\label{eqn:group_action}
g \cdot (\bm{x}_1,\dots,\bm{x}_N) = (\bm{Q}\bm{x}_1 + \bm{c},\dots,\bm{Q}\bm{x}_N+\bm{c}).
\end{align}
This action splits $\mathbb{R}^{3N}$ into disjoint sets of equivalence classes \cite{dummit}, which we now describe. For any $\bm{u}=(\bm{x}_1,\dots,\bm{x}_N) \in \mathbb{R}^{3N}$, let $\mcal{O}_{\bm{u}} \subset \mathbb{R}^{3N}$ denote an equivalence class which is defined as\footnote{The notation ``$\{g\cdot \bm{u} \mid g \in \mcal{G}\}$'' should be read as ``the set of all $g\cdot \bm{u}$, such that $g$ is in the Euclidean group $\mcal{G}$''.} 
\begin{align}
\mcal{O}_{\bm{u}} := \{g\cdot \bm{u} \mid  g \in \mcal{G}\},
\label{eqn:orbit}
\end{align}
where $g \cdot \bm{u}$ denotes the action of $g$ on $\bm{u}$ defined in \eref{eqn:group_action}. In other words, $\mcal{O}_{\bm{u}}$ represents the set of all configurations which are related to the configuration $\bm{u}$ by a rigid body motion and/or reflection. Due to the invariance of the potential energy, we can view the function $\pot_{\rm int}$ as a function on the set of equivalence classes, i.e., 
\begin{align}
\overline{\pot}_{\rm int}(\mcal{O}_{\bm{u}}) = \potfisher_{\rm int}(\bm{u}),
\label{eqn:pot_orbit}
\end{align}
because
\begin{align}
\potfisher_{\rm int}(\bm{v}) = \potfisher_{\rm int}(\bm{u}) \qquad \forall \bm{v} \in \mcal{O}_{\bm{u}}.
\end{align}
Now, consider a set $\mcal{S} \subset \mathbb{R}^{N(N-1)/2}$, defined as
\begin{align}
\mcal{S} := \{(r_{12},r_{13}, \dots,& r_{1N}, r_{23}, \dots, r_{(N-1)N}) \mid \notag \\
&r_{\alpha \beta} = \vnorm{\bm{x}_\alpha - \bm{x}_\beta}, (\bm{x}_1,\dots,\bm{x}_N) \in \mathbb{R}^{3N}\}.
\end{align}
In other words, the set $\mcal{S}$ consists of all possible $N(N-1)/2$-tuples of real numbers which correspond to the distances between $N$ particles in $\real{3}$.\footnote{\label{fn:phys_dist}The key here is that not all $N(N-1)/2$ combinations of real numbers constitute a valid set of physical distances. The distances must satisfy certain geometric constraints in order to be physically meaningful as explained below.} In technical terms, the coordinates of any point in $\mcal{S}$ are said to be \emph{embeddable} in $\real{3}$. Note that $\mcal{S}$ is a proper subset of $\mathbb{R}^{N(N-1)/2}$ as it consists of only those $N(N-1)/2$-tuple distances which satisfy certain geometric constraints. In fact, the set $\mcal{S}$ represents a $(3N-6)$-dimensional manifold in $\mathbb{R}^{N(N-1)/2}$, commonly referred to as the \emph{shape space}.

Let $\phi$ be the mapping taking a point in configuration space to the corresponding set of distances in $\mcal{S}$, i.e., $\phi:\mathbb{R}^{3N} \to \mcal{S}:(\bm{x}_1,\dots,\bm{x}_N) \mapsto (r_{12},\dots,r_{(N-1)N})$, where $r_{\alpha \beta}$ from here onwards is used to denote $\vnorm{\bm{x}_\alpha-\bm{x}_\beta}$. Since the Euclidean group preserves distances, it immediately follows that the map 
\begin{align}
\label{eqn:bijection}
\bar{\phi}:\{\text{Equivalence classes}\} \to \mcal{S},
\end{align}
defined as $\bar{\phi}(\mcal{O}_{\bm{u}}) = \phi(\bm{u})$, is a bijection (one-to-one and onto mapping) from the set of equivalence classes to the set $\mcal{S}$.\footnote{$\bar{\phi}$ is surjective (onto) by the definition of $\mcal{S}$. The proof that it is injective (one-to-one) is similar to the proof of the \emph{basic invariance theorem} for the simultaneous invariants of vectors due to Cauchy, which can be found in \cite[Section 11]{truesdell}.} This essentially means that for every set of equivalent configurations, i.e., configurations related to each other by a rigid body motion and/or reflection, there exists a unique $N(N-1)/2$-tuple of distances and vice versa. From \eref{eqn:pot_orbit} and \eref{eqn:bijection}, it immediately follows that the potential energy of the system can be completely described by a function $\breve{\pot}_{\rm int}:\mcal{S} \to \mathbb{R}$, defined as
\begin{align}
\breve{\pot}_{\rm int}(\bm{s}) := 
\overline{\pot}_{\rm int}(\bar{\phi}^{-1}(\bm s)) \qquad \forall \bm{s} \in \mcal{S}.
\label{eqn:pot_shape}
\end{align}

We now restrict our discussion to those systems for which there exists a continuously differentiable extension of $\breve{\pot}_{\rm int}$, defined on the shape space, to $\mathbb{R}^{N(N-1)/2}$.\footnote{The extension is necessary since $\breve{\pot}_{\rm int}$ is defined in \eref{eqn:pot_shape} only on the set $\mcal{S}$. We need to extend the definition to {\em all} points in $\mathbb{R}^{N(N-1)/2}$, whether they correspond to a set of physical distances or not, in order to be able to compute derivatives as explained later in the text. This issue has been overlooked in the past (see for example \cite{delph2005}), which leads to the conclusion that the stress tensor is always symmetric. It turns out that this conclusion is correct (at least for point masses without internal structure), but the reasoning is more involved as we show later.} This is justifiable because of the fact that all interatomic potentials used in practice, for a system of $N$ particles, are either continuously differentiable functions on $\mathbb{R}^{N(N-1)/2}$, or can easily be extended to one. For example, the pair potential and the embedded-atom method (EAM) potential \cite{eam} are continuously differentiable functions on $\mathbb{R}^{N(N-1)/2}$, while the Stillinger-Weber \cite{stillinger1985} and the Tersoff \cite{tersoff1988} potentials can be easily extended to $\mathbb{R}^{N(N-1)/2}$ by expressing the angles appearing in them as a function of distances between particles. Therefore, we assume that there exist a continuously differentiable function $\pot_{\rm int}:\mathbb{R}^{N(N-1)/2} \to \mathbb{R}$, such that the restriction of $\pot_{\rm int}$ to $\mcal{S}$ is equal to $\breve{\pot}_{\rm int}$:
\begin{align}
\pot_{\rm int} (\bm{s}) = \breve{\pot}_{\rm int}(\bm{s}) \quad
\forall \bm{s}=(r_{12},\dots,r_{(N-1)N}) \in \mcal{S}.
\label{eqn:restriction}
\end{align}                  

An immediate question that arises is whether this extension is unique in a neighborhood of $\bm{s} \in \mcal{S}$. Note that for $N \le 4$, $3N-6 = N(N-1)/2$. Therefore, for $N \le 4$, for every point $\bm{s} \in \mcal{S}$, there exists a neighborhood in $\mathbb{R}^{N(N-1)/2}$ which lies in $\mcal{S}$. However, for $N>4$, there may be multiple extensions of $\breve{\pot}_{\rm{int}}$.

As noted above, the reason we are considering an extension is to define the partial derivative of the potential energy with respect to each coordinate of a point in $\mathbb{R}^{N(N-1)/2}$. This will be used later to define the stress tensor. For example, the partial derivative of $\pot_{\rm int}(\zeta_{12},\dots,\zeta_{(N-1)N})$ with respect to $\zeta_{12}$ at any point $\bm{s}=(r_{12},\dots,r_{(N-1)N}) \in \mcal{S}$, defined as
\begin{align}
\label{eqn:diff_extension}
\frac{d \pot_{\rm int}}{d \zeta_{12}}(\bm{s}) = \lim_{\epsilon \to 0} 
\frac{\pot_{\rm int}(r_{12}+\epsilon,\dots,r_{N(N-1)/2}) -\pot_{\rm int}(r_{12},\dots,r_{N(N-1)/2})}{\epsilon},
\end{align}
requires us to evaluate the function at non-embeddable points. 

It will be shown later that the quantity evaluated in \eref{eqn:diff_extension} may differ for different extensions. On the other hand, $\nabla_{\bm{x}_\alpha} \pot_{\rm int}$ is uniquely defined for any extension. This is because
\begin{align}
\nabla_{\bm{x}_\alpha} \pot_{\rm int}(\bm{s}) &= \nabla_{\bm{x}_\alpha} \breve{\pot}_{\rm int}(\bm{s}) \notag \\
&= \nabla_{\bm{x}_\alpha}\overline{\pot}_{\rm int}(\bar{\phi}^{-1}(\bm{s})) \notag \\
&= \nabla_{\bm{x}_\alpha}\potfisher_{\rm int}(\bm{u}), \label{eqn:force_equality}
\end{align}
where $\bar{\phi}^{-1}(\bm{s})=\mcal{O}_{\bm{u}}$, which implies $\phi(\bm{u})=\bm{s}$,\footnote{Note that the vector $\bm{u}$ appearing in \eref{eqn:force_equality} can be replaced by any $\bm{v} \in \mcal{O}_{\bm{u}}$.} and we have used \eref{eqn:restriction}, \eref{eqn:pot_shape} and \eref{eqn:pot_orbit} in the first, second and the last equality respectively. 

We next address the possibility of having multiple extensions for the potential energy by studying the various constraints that the distances between particles have to satisfy in order to be embeddable in $\real{3}$.  We demonstrate, through a simple example, how multiple extensions for the potential energy can lead to a non-unique decomposition of the force on a particle, which in turn leads to a non-unique pointwise stress tensor.

\subsubsection*{Central-force decomposition and the possibility of alternate extensions}
\label{page:altext}
We will now show that the force on a particle can always be decomposed as a sum of central forces. The force on a particle due to internal interactions is defined in \eref{eqn:def_fi}. Therefore, for any configuration $\bm{u} \in \mathbb{R}^{3N}$, we have
\begin{align}
\label{eqn:force_alpha}
\bm{f}^{\rm int}_\alpha(\bm{u}) &= -\nabla_{\bm{x}_\alpha} \potfisher_{\rm{int}}(\bm{u}).
\end{align}
Using \eref{eqn:force_equality}, \eref{eqn:force_alpha} takes the form
\begin{align}
\bm{f}^{\rm int}_\alpha(\bm{u}) &= -\nabla_{\bm{x}_\alpha} \pot_{\rm{int}}(\bm{s}) \vert_{\bm{s}=\phi(\bm{u})} \notag \\
&= \sum_{\substack{\beta \\ \beta \ne \alpha}} \bm{f}_{\alpha\beta}(\bm{u}), \label{eqn:f_decomp_general}
\end{align}
where $\bm{s}=\phi(\bm{u})=(r_{12},\dots,r_{(N-1)N})$ and 
\begin{equation}
\label{eqn:define_fij}
\bm{f}_{\alpha \beta}(\bm{u}) := \left \{
\begin{array}{ll}
\frac{\partial\pot_{\rm int}}{\partial \zeta_{\alpha\beta}}(\phi(\bm{u})) \frac{\bm{x}_\beta - \bm{x}_\alpha}{r_{\alpha \beta}} & 
                \mbox{if $\alpha<\beta$}, \\
\frac{\partial\pot_{\rm int}}{\partial \zeta_{\beta\alpha}}(\phi(\bm{u})) \frac{\bm{x}_\beta - \bm{x}_\alpha}{r_{\alpha \beta}} & 
                \mbox{if $\alpha>\beta$},
\end{array}
\right.
\end{equation}
is the contribution to the force on particle $\alpha$ due to the presence of particle $\beta$.

Note that $\bm{f}_{\alpha\beta}$ is parallel to the direction $\bm{x}_\beta - \bm{x}_\alpha$ and satisfies $\bm{f}_{\alpha \beta}=-\bm{f}_{\beta \alpha}$.  We therefore note the important result that the \emph{internal force on a particle, for any interatomic potential that has a continuously differentiable extension, can always be decomposed as a sum of central forces, i.e., forces parallel to directions connecting the particle to its neighbors}.\footnote{
The result that the force on a particle, modeled using any interatomic potential with a continuously differentiable extension, can be decomposed as sum of central forces may seem strange to some readers. This may be due to the common confusion in the literature of using the term ``central-force models'' to refer to simple pair potentials. In fact, we see that due to the invariance requirement stated on Page~\pageref{page:invariance}, {\em all} interatomic potentials (including those with explicit bond angle dependence) that can be expressed as a continuously differentiable function as described in the text, are central-force models. By this we mean that the force on any particle (say $\alpha$) can be decomposed as a sum of terms, $\bm{f}_{\alpha\beta}$, aligned with the vectors joining particle $\alpha$ with its neighbors and satisfying action and reaction. 

The difference between the general case and that of a pair potential is that for a pair potential, $\vnorm{\bm{f}_{\alpha\beta}}$ depends {\em only} on the distance $r_{\alpha\beta}$ between the particles, whereas for a general potential, the dependence is on a larger set of distances, $\vnorm{\bm{f}_{\alpha\beta}}=\frac{\partial \pot_{\rm int}}{\partial \zeta_{\alpha\beta}}(r_{12},r_{13},\dots,r_{(N-1),N})$, i.e., $\vnorm{\bm{f}_{\alpha\beta}}$ depends on the {\em environment} of the ``bond'' between $\alpha$ and $\beta$. For this reason, $\bm{f}_{\alpha\beta}$ for a pair potential is a property of particles $\alpha$ and $\beta$ alone and can be physically interpreted as the ``force exerted on particle $\alpha$ by particle $\beta$''. Whereas, in the more general case of arbitrary interatomic potentials, the physical significance of the interatomic force is less clear and at best we can say that $\bm{f}_{\alpha\beta}$ is the ``contribution to the force on particle $\alpha$ due to the presence of particle $\beta$''.
}
We will see later in \sref{sec:stronglaw} that the central-force decomposition is the only physically-meaningful partitioning of the force.

The remaining question is how different potential energy extensions affect the force decomposition in \eref{eqn:f_decomp_general}. We have already established in \eref{eqn:force_equality} and \eref{eqn:force_alpha} that the force $\bm{f}_\alpha^{\rm int}$ is independent of the particular extension used. However, we show below that the individual terms in the decomposition, $\bm{f}_{\alpha\beta}$, are {\em not} unique. These terms depend on the manner in which the potential energy, defined on the shape space, is extended to its neighborhood in the higher-dimensional Euclidean space.

In order to construct different extensions, we use the geometric constraints that the distances have to satisfy in order for them to be embeddable in $\real{3}$.\footnote{We thank Ryan Elliott for suggesting this line of thinking.} The nature of these constraints is studied in the field of \emph{distance geometry}, which describes the geometry of sets of points in terms of the distances between them (see Appendix \ref{sec:geometry}). One of the main results of this theory, is that the constraints are given by \emph{Cayley-Menger determinants}, which are related to the volume of a simplex formed by $N$ points in an $N-1$ dimensional space.

For simplicity let us restrict our discussion to one dimension. It is easy to see that in one dimension the number of independent coordinates are $N-1$ and for $N>2$ the number of interatomic distances exceeds the number of independent coordinates. Therefore, let the material system $\mathcal{M}$ consist of three point masses interacting in one dimension. The standard pair potential representation for this system, which is an extension of the potential energy to the higher-dimensional Euclidean space, is given by
\begin{equation}
\pot_{\rm{int}}(\zeta_{12},\zeta_{13},\zeta_{23}) = 
\pot_{12}(\zeta_{12}) + \pot_{13}(\zeta_{13}) + \pot_{23}(\zeta_{23}).
\label{eqn:pair_pot}
\end{equation}                                                           
Since the calculation gets unwieldy, let us consider the special case when the particles are arranged to satisfy $x_1<x_2<x_3$, for which $r_{13}=r_{12}+r_{23}$. Using \eref{eqn:f_decomp_general}, the internal force, $f_1^{\rm{int}}$, evaluated at this configuration, is decomposed as
\begin{align}
f_1^{\rm{int}}(r_{12},r_{13},r_{23}) = -\frac{d\pot_{\rm{int}}}{dx_1} &= -\frac{d\pot_{12}}{dx_1}-\frac{d\pot_{13}}{dx_1}  \notag \\
&= \pot'_{12}(r_{12})  + \pot'_{13}(r_{13}) \notag \\
& =: f_{12} + f_{13}. \label{eqn:define_fij_pair}
\end{align}                         
We now provide an alternate extension to the standard pair potential representation given in \eref{eqn:pair_pot}. The Cayley-Menger determinant corresponding to a cluster of three points (see \eref{eq:CMD4}) is identically equal to zero at every point on the shape space. This is because the shape space corresponds to a configuration of three collinear points, and the area of the triangle formed by three collinear points is zero. Thus, we have
\begin{align}
\chi(r_{12},r_{13},r_{23}) &= (r_{12}-r_{13}-r_{23})(r_{23}-r_{12}-r_{13}) \notag \\
& \quad \times (r_{13}-r_{23}-r_{12})(r_{12}+r_{13}+r_{23}) \notag \\
&= 0. \label{eqn:cayley_1d}
\end{align}
Using the identity in \eref{eqn:cayley_1d}, an alternate extension $\pot^{\mcal{A}}_{\rm{int}}$ is constructed:
\begin{align}
\pot^{\mcal{A}}_{\rm{int}}(\zeta_{12},\zeta_{13},\zeta_{23}) &= \pot_{\rm int}(\zeta_{12},\zeta_{13},\zeta_{23}) + \chi(\zeta_{12},\zeta_{13},\zeta_{23}).
\label{eqn:pot_rep_2}
\end{align}
Note that $\pot^{\mcal{A}}_{\rm{int}}$ is indeed an extension because from \eref{eqn:cayley_1d} it is clear that $\pot^{\mcal{A}}_{\rm{int}}$ is equal to $\pot_{\rm{int}}$ at every point on the shape space of the system and it is continuously differentiable because $\chi(\zeta_{12},\zeta_{13},\zeta_{23})$, being a polynomial, is infinitely differentiable. Let us now see how the internal force, $f_1^{\rm int}$, for the special configuration considered in this example, is decomposed using the new extension:
\begin{align}
f_1^{\rm{int}} = -\frac{d\pot_{\rm{int}}^{\mcal{A}}}{dx_1} &= 
-\frac{d\pot_{\rm{int}}}{dx_1}  - \frac{d\chi}{dx_1} \notag \\
&= \left (\pot'_{12} - \frac{\partial \chi}{\partial \zeta_{12}}(\bm{s}) \frac{\partial \zeta_{12}}{\partial x_1}(\bm{s}) \right ) +
\left( \pot'_{13} - \frac{\partial \chi}{\partial \zeta_{13}}(\bm{s}) \frac{\partial \zeta_{13}}{\partial x_1}(\bm{s}) \right ) \notag \\
&= \left (f_{12}-8r_{12}r_{23}(r_{12} + r_{23}) \right ) + \left ( f_{13}+8r_{12}r_{23}(r_{12} + r_{23}) \right ) \notag \\
&=: \tilde{f}_{12} + \tilde{f}_{13}, \label{eqn:f1_decomp}
\end{align}
It is clear from \eref{eqn:define_fij_pair} and \eref{eqn:f1_decomp} that the central-force decomposition is not the same for the two representations, i.e., $f_{12} \ne \tilde{f}_{12}$ and $f_{13} \ne \tilde{f}_{13}$, however the force on particle 1, $f_1^{\rm int}$, is the same in both cases as expected.

It is very interesting to note that $\pot^{\mcal{A}}_{\rm int}$ is \emph{not} a pair potential (based on the definition of a pair potential), but it is equivalent to a pair potential, i.e., it agrees with a pair potential on the shape space. Thus, the set of continuously differentiable extensions of a given interatomic potential function form an equivalence class.  It is not clear at this stage if these equivalence classes can be fully expressed in terms of the Cayley-Menger determinant constraints. 

Although the above example is quite elementary, this process can be extended to any arbitrary number of particles in three dimensions. Any given potential can be altered to an equivalent potential by adding a function of the Cayley-Menger determinants corresponding to any cluster of 5 or 6 particles (see Appendix \ref{sec:geometry}). This function must be continuously differentiable and equal to zero when all of its arguments are zero. For example, a new representation in three dimensions can be constructed by adding a linear combination of the Cayley-Menger determinants:
\begin{equation}
\pot_{\rm int}^* = \pot_{\rm int}(\zeta_{12}, \dots, \zeta_{(N-1)N}) + \sum_{k=1}^m \lambda_k \chi_k,
\label{eqn:3drep}
\end{equation}
where there are $m$ constraints defined by the Cayley-Menger determinants $\chi_k$, and $\lambda_k$ are constants.\footnote{Note that \eref{eqn:3drep} has the same form as a Lagrangian with the $\lambda$ terms playing the role of Lagrange multipliers. For a static minimization problem, we seek to minimize $\pot_{\rm int}^*$, without violating the physical constraints relating the distances to each other. (This is equivalent to minimizing $\potfisher_{\rm int}$ with respect to the positions of particles.) Thus, the original constrained minimization of $\pot_{\rm int}$ is replaced by the problem of finding the saddle points of $\pot_{\rm int}^*$.} 

From this point on, we abuse our notation slightly, and write for any $\bm{s}=(r_{12}, \dots,\\ r_{(N-1)N})\in\mcal{S}$:
\begin{align}
\frac{\partial \pot_{\rm int}}{\partial r_{\alpha \beta}} \quad \text{for} \quad \frac{\partial \pot_{\rm int}}{\partial \zeta_{\alpha \beta}}(\bm{s}).
\end{align}
Also, we assume that there exists a continuously differentiable extension whenever we write $\bm{f}_{\alpha\beta}$ and sometimes refer to a continuously differentiable extension as an extension.
\subsubsection*{Derivation of the pointwise stress tensor}
We now return to the differential equation in \eref{eqn:stress_force_differential} for the potential part of the pointwise stress tensor. Substituting the force decomposition given in \eref{eqn:f_decomp_general} corresponding to a continuously differentiable extension, into \eref{eqn:stress_force_differential}, we obtain
\begin{equation}
\label{eqn:stress_force_differential_*}
\divr_{\bm{x}} \stress_{\rm{v}}(\bm{x},t) = \sum_{\substack{\alpha,\beta \\ \alpha \ne \beta}} \langle W \bm{f}_{\alpha \beta} \mid \bm{x}_\alpha = \bm{x} \rangle.
\end{equation}
On using the identity 
\begin{equation}
\left \langle \bm{f}_{\alpha\beta} W \mid \bm{x}_\alpha = \bm{x} \right \rangle = \int_{\real{3}} \left \langle \bm{f}_{\alpha\beta} W \mid \bm{x}_\alpha = \bm{x}, \bm{x}_\beta = \bm{y} \right \rangle \, d \bm{y},
\end{equation}
equation \eref{eqn:stress_force_differential_*} takes the form
\begin{equation}
\label{eqn:stress_force_differential_**}
\divr_{\bm{x}} \bm{\sigma}_{\rm{v}}(\bm{x},t) = \sum_{\substack{\alpha,\beta \\ \alpha \ne \beta}} \int_{\real{3}} \left \langle W \bm{f}_{\alpha\beta} \mid \bm{x}_\alpha = \bm{x}, \bm{x}_\beta = \bm{y} \right \rangle \, d\bm{y}.
\end{equation}
We now note that, being anti-symmetric, the integrand in the right-hand side of the above equation satisfies all the necessary conditions for the application of Lemma \ref{lem1} given in Appendix A. Conditions (1) and (2) in Appendix A are satisfied through the regularity conditions on $W$. Therefore, using Lemma \ref{lem1}, which was proved by Noll in \cite{noll1955}, we have
\begin{align}
\label{eqn:stress_force_general}
&\bm{\sigma}_{\rm{v}}(\bm{x},t) \\ 
&= \frac{1}{2} \sum_{\substack{\alpha,\beta \\  \alpha \neq \beta}} \int_{\real{3}} \int_{s=0}^{1} \left \langle -\bm{f}_{\alpha\beta} W \mid \bm{x}_\alpha=\bm{x}+s\bm{z},\bm{x}_\beta = \bm{x} - (1-s)\bm{z} \right \rangle \, ds \otimes \bm{z} \, d\bm{z} \notag \\
&= \frac{1}{2} \sum_{\substack{\alpha,\beta \\  \alpha \neq \beta}} \int_{\real{3}} \frac{\bm{z} \otimes \bm{z}}{\vnorm{\bm{z}}} \int_{s=0}^{1} \left \langle \frac{\partial\pot_{\rm int}}{\partial r_{\alpha\beta}} W \mid \bm{x}_\alpha=\bm{x}+ s\bm{z},\bm{x}_\beta = \bm{x} - (1-s)\bm{z} \right \rangle \, ds \, d\bm{z},
\end{align}
where in passing to the second line, we have used \eref{eqn:define_fij} and the identity $\bm{x}_\alpha-\bm{x}_\beta = \bm{x}+s\bm{z} - [\bm{x}-(1-s)\bm{z}] = \bm{z}$.  For the special case of a pair potential, $\partial\pot_{\rm int}/\partial r_{\alpha\beta}=\pot'_{\alpha\beta}(r_{\alpha\beta})$, and \eref{eqn:stress_force_general} reduces to the expression originally given in \cite{noll1955}.

\begin{figure}
\centering
\includegraphics[scale=0.4]{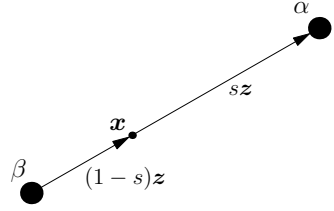}
\caption{A schematic diagram helping to explain the vectors appearing in the pointwise potential stress expression in \eref{eqn:stress_force_general}. The bond $\alpha$--$\beta$ is defined by the vector $\bm{z}$. When $s=0$, atom $\alpha$ is located at point $\bm{x}$, and when $s=1$, atom $\beta$ is located at $\bm{x}$.}
\label{fig:stressbond}
\end{figure} 

\medskip
The expression for the potential part of the pointwise stress tensor in \eref{eqn:stress_force_general} is a general result applicable to all interatomic potentials. We make some important observations regarding this expressions below:
\begin{enumerate}

\item Although the expression for $\bm{\sigma}_{\rm{v}}$ appears complex, it is actually conceptually quite simple. $\bm{\sigma}_{\rm{v}}$ at a point $\bm{x}$ is the superposition of the expectation values of the forces in all possible bonds passing through $\bm{x}$. The variable $\bm{z}$ selects a bond length and direction and the variable $s$ slides the bond through $\bm{x}$ from end to end (see \fref{fig:stressbond}).

\item $\bm{\sigma}_{\rm{v}}$ is symmetric. This is clear because the term $\bm{z}\otimes\bm{z}$ is symmetric. Since the kinetic part of the stress in \eref{eqn:stress_kinetic} is also symmetric, the conclusion is that the \emph{pointwise stress tensor is symmetric for all interatomic potentials}.

\item Since $\bm{\sigma}_{\rm{v}}$ depends on the nature of the force decomposition and different extensions of a given potential energy can result in different force decompositions, we conclude that the pointwise stress tensor is \emph{non-unique} for all interatomic potentials (including the pair potential). We show in \sref{sec:unique_macro_stress} that the difference due to any two pointwise stress tensors, resulting from different extensions for the interatomic potential energy, tends to zero as the volume of the domain over which these pointwise quantities are spatially averaged tends to infinity. Therefore, as expected, the macroscopic stress tensor, which is defined in the thermodynamic limit (see footnote~\ref{foot:tdlimit} on page~\pageref{foot:tdlimit}), is always unique and is independent of the potential energy extension. 

\item Another source of non-uniqueness is that any expression of the form, $\bm{\sigma}_{\rm{v}}+\tilde{\bm{\sigma}}$, where $\divr_{\bm{x}} \tilde{\bm{\sigma}} = {\bf 0}$, also satisfies the balance of linear momentum and is therefore also a solution.  We address this issue in \sref{sec:hardy}, where we show that in the thermodynamic limit under equilibrium conditions, the spatially averaged counterpart to $\bm{\sigma}_{\rm{v}}$ converges to the virial stress derived in \sref{ch:canonical}.
\end{enumerate}

The above results hinge on the use of the central-force decomposition in \eref{eqn:f_decomp_general}. One may wonder whether other \emph{non-central} decompositions exist, and if yes, why are these discarded. This is discussed in the next section.

\subsection{Non-central-force decompositions and the strong law of action and reaction}
\label{sec:stronglaw}

In the previous section, we showed that as a consequence of the invariance of the potential energy with respect to the Euclidean group, for any interatomic potential with a continuously differentiable extension, the force on a particle can always be represented as a sum of central forces.  In this section, we show that other \emph{non-central-force decompositions} are possible, however that these violate the \emph{strong law of action and reaction}, which we prove below, and therefore do not constitute physically-meaningful force decompositions.

\subsubsection*{A proposal for a non-symmetric stress tensor for a three-body potential}
As an example, let us now consider the case of a three-body potential. For simplicity, we assume that the potential only has three-body terms and all particles are identical. Under these conditions, the internal potential energy is
\begin{equation}
\label{eqn:pot_rep_3body}
\pot_{\rm{int}} = \sum_{\substack{\alpha,\beta,\gamma \\ \alpha < \beta < \gamma}} \potfisher(\bm{x}_{\alpha},\bm{x}_{\beta},\bm{x}_{\gamma}),
\end{equation}
where $\potfisher(\bm{x}_\alpha,\bm{x}_\beta,\bm{x}_\gamma)$ is the potential energy of an isolated cluster, $\{\alpha,\beta,\gamma\}$, and $\sum_{\substack{\alpha,\beta,\gamma \\ \alpha < \beta <\gamma}}$ represents a triple sum. We know that a central-force decomposition can be obtained by following the procedure outlined in the previous section and that this leads to a symmetric pointwise stress tensor in \eref{eqn:stress_force_general}.  Alternatively, a \emph{non-symmetric} three-body stress tensor is derived as follows. To keep things simple, we derive the stress tensor for a system containing only three particles.  Rewriting \eref{eqn:pot_rep_3body} for this case, we have
\begin{equation}
\label{eqn:summation_split}
\pot_{\rm{int}} = \potfisher(\bm{x}_1,\bm{x}_2,\bm{x}_3) = \sum_{\alpha=1}^3 \phi_\alpha,
\end{equation}
where
\begin{equation}
\label{eqn:def_phi_i}
\phi_\alpha= \frac{1}{3} \potfisher(\bm{x}_1,\bm{x}_2,\bm{x}_{3})
\end{equation}
is the potential energy assigned to particle $\alpha$, equal to one-third of the total potential energy. Substituting \eref{eqn:summation_split} into \eref{eqn:stress_force_differential}, we obtain
\begin{align}
\divr_{\bm{x}} \bm{\sigma}_{\rm{v}}(\bm{x},t) &= - \sum_{\alpha,\beta} \left \langle W \nabla_{\bm{x}_\alpha} \phi_\beta \mid \bm{x}_\alpha = \bm{x} \right \rangle \notag \\
&=- \sum_{\substack{\alpha,\beta \\ \alpha \ne \beta}} \left \langle W \nabla_{\bm{x}_\alpha} \phi_\beta \mid \bm{x}_\alpha = \bm{x} \right \rangle -  \sum_{\alpha} \left \langle W \nabla_{\bm{x}_\alpha} \phi_\alpha \mid \bm{x}_\alpha = \bm{x} \right \rangle. \label{eqn:stress_force_differential_3body_*}
\end{align}
Since the cluster of three particles is isolated, the net force on the cluster due to internal interactions is zero. Therefore, from \eref{eqn:summation_split}, we have
\begin{equation}
\label{eqn:postulate}
\nabla_{\bm{x}_\alpha}\phi_\alpha = -\sum_{\beta \ne \alpha} \nabla_{\bm{x}_\beta}\phi_\alpha.
\end{equation}
Using this relation, equation \eref{eqn:stress_force_differential_3body_*} simplifies to
\begin{equation}
\label{eqn:stress_force_differential_3body_**}
\divr_{\bm{x}} \bm{\sigma}_{\rm{v}}(\bm{x},t) = - \sum_{\substack{\alpha,\beta \\ \alpha \ne \beta}} \left \langle W (\nabla_{\bm{x}_\alpha} \phi_\beta - \nabla_{\bm{x}_\beta} \phi_\alpha) \mid \bm{x}_\alpha = \bm{x} \right \rangle.
\end{equation}
Let
\begin{equation}
\label{eqn:f_ij_3body}
\bar{\bm{f}}_{\alpha\beta} := \nabla_{\bm{x}_\beta} \phi_\alpha  - \nabla_{\bm{x}_\alpha} \phi_\beta.
\end{equation}
Now, using the identity
\[
\left \langle \bar{\bm{f}}_{\alpha\beta} W \mid \bm{x}_\alpha = \bm{x} \right \rangle = \int_{\real{3}} \left \langle \bar{\bm{f}}_{\alpha\beta} W \mid \bm{x}_\alpha = \bm{x}, \bm{x}_\beta = \bm{y} \right \rangle \, d \bm{y},
\]
and the definition given in \eref{eqn:f_ij_3body}, equation \eref{eqn:stress_force_differential_3body_**} takes the form
\begin{equation}
\label{eqn:stress_force_differential_3body_***}
\divr_{\bm{x}} \bm{\sigma}_{\rm{v}}(\bm{x},t) = \sum_{\substack{\alpha,\beta \\ \alpha \ne \beta}} \int_{\real{3}} \left \langle W \bar{\bm{f}}_{\alpha\beta} \mid \bm{x}_\alpha = \bm{x}, \bm{x}_\beta = \bm{y} \right \rangle \, d\bm{y}.
\end{equation}
Let us now study the definition of $\bar{\bm{f}}_{\alpha\beta}$ given in \eref{eqn:f_ij_3body}. From \eref{eqn:def_phi_i} we have
\begin{equation}
\bar{\bm{f}}_{\alpha\beta} = -\nabla_{\bm{x}_\alpha} \phi_\beta + \nabla_{\bm{x}_\beta} \phi_\alpha = \frac{1}{3} \left [ - \frac{\partial \potfisher}{\partial \bm{x}_\alpha} + \frac{\partial \potfisher}{\partial \bm{x}_\beta} \right ]
= \frac{1}{3} (\bm{f}^{\rm int}_\alpha  - \bm{f}^{\rm int}_\beta). \label{eqn:phi_ij_antisym}
\end{equation}
The above equation suggests how the force $\bm{f}^{\rm int}_\alpha$ is decomposed. For example, $\bm{f}^{\rm int}_1$ is decomposed as 
\begin{equation}
\bm{f}^{\rm int}_1 = \bar{\bm{f}}_{12} + \bar{\bm{f}}_{13} = \frac{1}{3}(\bm{f}^{\rm int}_1 - \bm{f}^{\rm int}_2) + \frac{1}{3}(\bm{f}^{\rm int}_1 - \bm{f}^{\rm int}_3).
\end{equation}
Rearranging this relation gives
\begin{equation}
\bm{f}^{\rm int}_1 + \bm{f}^{\rm int}_2 + \bm{f}^{\rm int}_3 = \bm{0}, 
\end{equation}
which is true since the cluster $\{1,2,3\}$ is isolated. 

\begin{figure}
\centering
\subfigure[]{\includegraphics[scale=0.6]{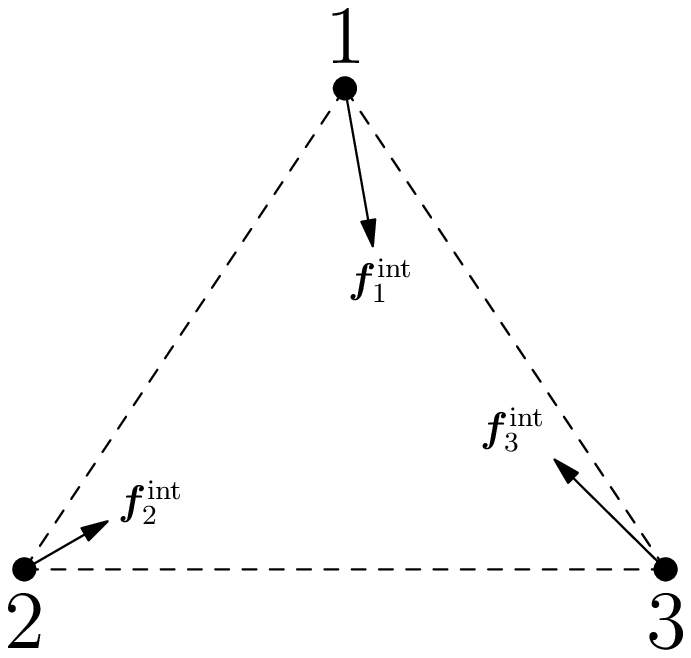} \label{fig:force}}
\subfigure[]{\includegraphics[scale=0.6]{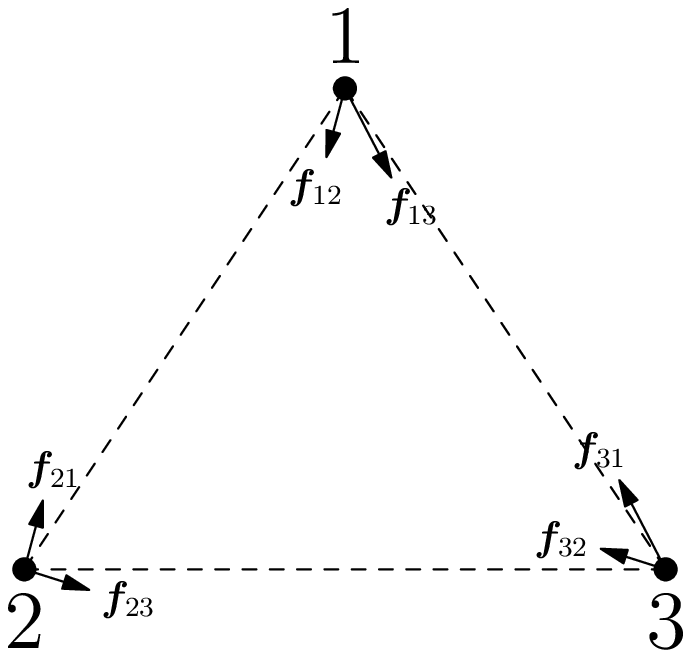} \label{fig:force_decomp}}
\caption{(a) shows the force on each particle in a system consisting of 3 particles which interact through a 3-body potential given in \eref{eqn:summation_split}. Since the potential is derived from an energy decomposition, we have $\bm{f}_1^{\rm{int}} + \bm{f}_2^{\rm{int}} +\bm{f}_3^{\rm{int}}=\bm{0}$. (b) shows the force decomposition of each $\bm{f}_{\alpha}^{\rm{int}}$ such that $\bm{f}_{\alpha\beta}=-\bm{f}_{\beta\alpha}$, but not necessarily parallel to the line joining particles $\alpha$ and $\beta$.}
\label{fig:force_3body}
\end{figure}

From \eref{eqn:phi_ij_antisym} it is clear that $\bar{\bm{f}}_{\alpha\beta}$ is anti-symmetric with respect to its arguments. Therefore, the integrand on the right-hand side of \eref{eqn:stress_force_differential_3body_***} satisfies all the necessary conditions for the application of Lemma \ref{lem1} given in Appendix \ref{ch:noll}. Conditions (1) and (2) in Appendix \ref{ch:noll} are satisfied through the regularity conditions on W. Therefore, using Lemma \ref{lem1}, we have
\begin{align}
\label{eqn:stress_general}
\bar{\stress}_{\rm{v}}(\bm{x},t) &= 
\frac{1}{2} \sum_{\substack{\alpha,\beta \\ \alpha \ne \beta}} \int_{\real{3}} \int_{s=0}^{1} \notag \\
& \left \langle -( \nabla_{\bm{x}_\beta} \phi_\alpha - \nabla_{\bm{x}_\alpha} \phi_\beta ) W \mid \bm{x}_\alpha = \bm{x} + s\bm{z}, \bm{x}_\beta = \bm{x} - (1-s) \bm{z} \right \rangle \, ds\otimes \bm{z} \, d\bm{z}.
\end{align}
The stress $\bar{\stress}_{\rm{v}}$ is non-symmetric in general because $\bar{\bm{f}}_{\alpha\beta}$, defined in \eref{eqn:phi_ij_antisym}, need not be parallel to the line joining particles $\alpha$ and $\beta$ as shown in \fref{fig:force_3body}. We therefore have two expressions for the stress for the same three-body potential. The symmetric expression in \eref{eqn:stress_force_general} and the non-symmetric expression in \eref{eqn:stress_general}. We show next that the non-central-force decomposition that led to the non-symmetric stress tensor is not physically meaningful since it violates the strong law of action and reaction.

\subsubsection*{Weak and strong laws of action and reaction\footnote{This derivation is due to Roger~Fosdick \cite{fosdick}.}}
The following derivation hinges on the fact that in a material system the balance laws of linear and angular momentum must be satisfied for any part of the body.

Consider a system of $N$ particles with masses $m_\alpha$ $(\alpha=1,\dots,N)$. The total force on particle $\alpha$ is 
\begin{equation}
\bm{f}_{\alpha} = \bm{f}^{\rm{ext}}_{\alpha} + \sum_{\substack{\beta \\ \beta \ne \alpha}} \bm{f}_{\alpha\beta},
\end{equation}
where $\bm{f}^{\rm{ext}}_{\alpha}$ is the external force on particle $\alpha$, and, as above, $\bm{f}_{\alpha\beta}$ is the contribution to the force on particle $\alpha$ due to the presence of particle $\beta$. No assumptions are made regarding the terms $\bm{f}_{\alpha\beta}$ or the interatomic potential from which they are derived.

A ``part'' $\wp_t$ of the system consists of $K \le N$ particles. We suppose $\bm{x}_0$ is a fixed point in space. Let $\bm{F}^{\rm{ext}}(\wp_t)$ denote the total force on the part $\wp_t$ external to the part. Let $\bm{M}^{\rm{ext}}(\wp_t; \bm{x}_0)$ denote the total external moment on $\wp_t$ about $\bm{x}_0$. Let $\bm{L}(\wp_t)$ be the linear momentum of the part $\wp_t$ and $\bm{H}(\wp_t; \bm{x}_0)$ be the angular momentum of $\wp_t$ about $\bm{x}_0$.

We adopt the following balance laws, valid for all parts of the system:\footnote{The view that the balance of linear momentum and the balance of angular momentum are fundamental laws of mechanics lies at the basis of continuum mechanics. See, for example, Truesdell's article ``Whence the Law of Moment and Momentum?'' in \cite{truesdell_essays}.}
\begin{align}
\bm{F}^{\rm{ext}}(\wp_t) &= \frac{d\bm{L}}{dt}(\wp_t), \label{eqn:force_balance}\\
\bm{M}^{\rm{ext}}(\wp_t;\bm{x}_0) &= \frac{d\bm{H}}{dt}(\wp_t;\bm{x}_0).\label{eqn:mom_balance}
\end{align}
We now show that by applying these balance laws to particular parts of the system, that the strong law of action and reaction can be established. As a first observation, let $\wp_t$ consist of the single particle $\alpha$. The external force and linear momentum for $\wp_t = \{\alpha\}$ is 
\begin{align}
\bm{F}^{\rm{ext}}(\{\alpha\}) &= \bm{f}^{\rm{ext}}_{\alpha}(t) + \sum_{\substack{\gamma \\ \gamma \ne \alpha}}\bm{f}_{\alpha\gamma}(t), \\
\bm{L}(\{\alpha\}) &= m_{\alpha}\dot{\bm{x}}_{\alpha}(t) \qquad \text{(no sum)}.
\end{align}
The balance of linear momentum in \eref{eqn:force_balance} requires
\begin{equation}
\label{eqn:newton_law_1}
\bm{f}^{\rm{ext}}_{\alpha} + \sum_{\substack{\gamma \\ \gamma \ne \alpha}} \bm{f}_{\alpha\gamma} = m_{\alpha} \ddot{\bm{x}}_{\alpha}.
\end{equation}
The external moment and angular momentum of $\wp_t$ is 
\begin{equation}
\bm{M}^{\rm{ext}}(\{\alpha\};\bm{x}_0) = (\bm{x}_{\alpha}(t) - \bm{x}_0) \times (\bm{f}^{\rm{ext}}_{\alpha} + \sum_{\substack{\gamma \\ \gamma \ne \alpha}} \bm{f}_{\alpha\gamma}) = m_{\alpha}(\bm{x}_{\alpha}(t) - \bm{x}_0) \times \ddot{\bm{x}}_{\alpha}(t),
\end{equation}
where we have used \eref{eqn:newton_law_1}, and 
\begin{equation}
\bm{H}(\{\alpha\};\bm{x}_0) = (\bm{x}_\alpha-\bm{x}_0) \times m_{\alpha} \dot{\bm{x}}_{\alpha}(t).
\end{equation}
The balance of angular momentum in \eref{eqn:mom_balance} is satisfied identically, since
\begin{align*}
m_{\alpha}(\bm{x}_{\alpha}(t) - \bm{x}_0) \times \ddot{\bm{x}}_\alpha(t) &= \frac{d}{dt} \left [ (\bm{x}_{\alpha}(t) - \bm{x}_0) \times m_{\alpha} \dot{\bm{x}}_{\alpha}(t) \right ] \\ 
&= \dot{\bm{x}}_{\alpha}(t) \times m_{\alpha} \dot{\bm{x}}_{\alpha}(t) + (\bm{x}_{\alpha}(t) - \bm{x}_0) \times m_{\alpha} \ddot{\bm{x}}_{\alpha}(t) \\
&= m_{\alpha} (\bm{x}_{\alpha}(t) - \bm{x}_0) \times \ddot{\bm{x}}_{\alpha}(t).
\end{align*}

\smallskip
As a second observation, let $\wp_t$ consist of the union of the two particles $\alpha$ and $\beta$. The external force and linear momentum are
\begin{align}
\bm{F}^{\rm{ext}}(\{\alpha,\beta\}) &= \bm{f}^{\rm{ext}}_{\alpha} + \bm{f}^{\rm{ext}}_{\beta} + \sum_{\substack{\gamma \\ \gamma \ne \alpha \ne \beta}} ( \bm{f}_{\alpha\gamma} + \bm{f}_{\beta\gamma}), \\
\bm{L}(\{\alpha,\beta\}) &= m_{\alpha} \dot{\bm{x}}_{\alpha} + m_{\beta} \dot{\bm{x}}_{\beta}.
\end{align}
The balance of linear momentum in \eref{eqn:force_balance} requires
\begin{equation}
\bm{f}^{\rm{ext}}_{\alpha} + \bm{f}^{\rm{ext}}_{\beta} + \sum_{\substack{\gamma \\ \gamma \ne \alpha \ne \beta}}( \bm{f}_{\alpha\gamma} + \bm{f}_{\beta\gamma}) = m_{\alpha} \ddot{\bm{x}}_{\alpha} + m_{\beta} \ddot{\bm{x}}_{\beta}.
\end{equation}
Subtracting \eref{eqn:newton_law_1} for particles $\alpha$ and $\beta$ gives
\begin{equation}
\sum_{\substack{\gamma \\ \gamma \ne \alpha \ne \beta}}(\bm{f}_{\alpha\gamma} + \bm{f}_{\beta\gamma}) - \sum_{\substack{\gamma \\ \gamma \ne \alpha}} \bm{f}_{\alpha\gamma} - \sum_{\substack{\gamma \\ \gamma \ne \beta}} \bm{f}_{\beta\gamma} = \bm{0},
\end{equation}
from which
\begin{equation}
\bm{f}_{\alpha\beta} + \bm{f}_{\beta\alpha} = \bm{0}.
\label{eqn:f_antisym}
\end{equation}
This relation is called the \emph{weak law of action and reaction} \cite{goldstein}. It shows that $\bm{f}_{\alpha\beta} = -\bm{f}_{\beta\alpha}$, but does not guarantee that $\bm{f}_{\alpha\beta}$ lies along the line connecting particles $\alpha$ and $\beta$.

\smallskip
Next, the external moment and angular momentum of $\wp_t$ is
\begin{align}
\bm{M}^{\rm{ext}}&(\{\alpha,\beta\};\bm{x}_0) \notag \\
&= (\bm{x}_{\alpha} - \bm{x}_0) \times ( \bm{f}^{\rm{ext}}_{\alpha} + \sum_{\substack{\gamma \\ \gamma \ne \alpha \ne \beta}} \bm{f}_{\alpha\gamma} ) + (\bm{x}_{\beta} -\bm{x}_0) \times ( \bm{f}^{\rm{ext}}_{\beta} + \sum_{\substack{\gamma \\ \gamma \ne \beta \ne \alpha}} \bm{f}_{\beta\gamma} ) \notag \\
&= (\bm{x}_{\alpha} - \bm{x}_0) \times (m_\alpha \ddot{\bm{x}}_\alpha - \bm{f}_{\alpha\beta}) + (\bm{x}_{\beta} - \bm{x}_0) \times (m_\beta \ddot{\bm{x}}_\beta- \bm{f}_{\beta\alpha}),
\end{align}
where we have used \eref{eqn:newton_law_1}, and
\begin{equation}
\bm{H}(\{\alpha,\beta\};\bm{x}_0) = (\bm{x}_\alpha-\bm{x}_0) \times m_{\alpha} \dot{\bm{x}}_{\alpha} + (\bm{x}_\beta-\bm{x}_0) \times m_{\beta} \dot{\bm{x}}_{\beta}.
\end{equation}
The balance of angular momentum in \eref{eqn:mom_balance} requires
\begin{align}
(\bm{x}_\alpha &- \bm{x}_0) \times (m_\alpha \ddot{\bm{x}}_\alpha - \bm{f}_{\alpha\beta}) + (\bm{x}_\beta- \bm{x}_0) \times (m_\beta \ddot{\bm{x}}_\beta- \bm{f}_{\beta\alpha}) \notag \\
&= \frac{d}{dt} \left [ (\bm{x}_{\alpha} - \bm{x}_0) \times m_{\alpha} \dot{\bm{x}}_{\alpha} + (\bm{x}_{\beta} - \bm{x}_0) \times m_{\beta} \dot{\bm{x}}_{\beta} \right ] \notag \\ 
&= \dot{\bm{x}}_{\alpha} \times m_{\alpha} \dot{\bm{x}}_{\alpha} + (\bm{x}_{\alpha} - \bm{x}_0) \times m_{\alpha} \ddot{\bm{x}}_{\alpha} + \dot{\bm{x}}_{\beta} \times m_{\beta} \dot{\bm{x}}_{\beta} + (\bm{x}_{\beta} - \bm{x}_0) \times m_{\beta} \ddot{\bm{x}}_{\beta} \notag \\
&= (\bm{x}_{\alpha} - \bm{x}_0) \times m_\alpha \ddot{\bm{x}}_{\alpha} + (\bm{x}_{\beta} - \bm{x}_0) \times m_\beta \ddot{\bm{x}}_{\beta},
\end{align}
which simplifies to
\[
(\bm{x}_\alpha - \bm{x}_0) \times \bm{f}_{\alpha\beta} + (\bm{x}_\beta- \bm{x}_0) \times \bm{f}_{\beta\alpha} = \bm{0},
\]
and, after using \eref{eqn:f_antisym}, we obtain
\begin{equation}
(\bm{x}_\alpha - \bm{x}_\beta) \times \bm{f}_{\alpha\beta} = \bm{0}.
\end{equation}
This shows that $\bm{f}_{\alpha\beta}$ must be {\em parallel} to the line joining particles $\alpha$ and $\beta$. This is the \emph{strong law of action and reaction}. We have shown that this law must hold for any force decomposition, in order for the balance of linear and angular momentum to hold for any subset of a system of particles.  

\subsubsection*{The possibility of non-symmetric stress}
Based on the proof given above for the strong law of action and reaction, we argue that only force decompositions that satisfy the strong law of action and reaction provide a physically-meaningful definition for $\bm{f}_{\alpha\beta}$. For example, the definition in \eref{eqn:f_ij_3body} is not physical because if it were used to compute the external moment acting on a sub-system of particles, as is done above, the balance of angular momentum would be violated. For this reason, this decomposition and the corresponding non-symmetric stress in \eref{eqn:stress_general} are discarded. The conclusion is that \emph{a pointwise stress tensor for a discrete system of point masses without internal structure has to be symmetric}.

In the next section, we discuss the possibility of expanding the class of solutions resulting from Irving--Kirkwood--Noll procedure in a way that makes it possible to obtain non-symmetric stress tensors for systems where the point particles have internal structure. This involves a relaxation of the assumption that the ``bonds'' connecting two particles are necessarily straight.

\subsection{Generalized non-symmetric pointwise stress for particles with internal structure}
\label{sec:gen_stress}
In \sref{sec:s_motion}, we saw that the Irving--Kirkwood--Noll procedure, when applied to multi-body potentials, results in a symmetric non-unique pointwise stress tensor.  We now seek to find additional solutions to \eref{eqn:stress_force_differential_**}, which are not obtained using the standard Irving--Kirkwood--Noll procedure. In arriving at \eref{eqn:stress_force_general} using Lemma \ref{lem1}, we can see that the contribution to the potential part of the stress at position $\bm{x}$ is due to all possible bonds, \emph{assumed to be straight lines}, that pass through $\bm{x}$. The question that naturally arises is to what extent can this assumption be weakened. In other words, can Lemma \ref{lem1} be generalized in a suitable manner so that non-straight bonds can be accommodated? Such a possibility was first discussed by Schofield and Henderson \cite{schofield1982}, who used the Irving and Kirkwood approach with a series expansion of the Dirac-delta distribution. It will be shown in this section, using Noll's more rigorous approach, that solutions giving rise to non-straight bonds are possible.

From a physical standpoint, non-straight bonds are possible for systems with internal degrees of freedom. An example would be the dipole-dipole interactions between water molecules resulting from the electrical dipole of each molecule.  The possibility of internal degrees of freedom was already raised by Kirkwood in his 1946 paper \cite{kirkwood1946}. The idea is to relate the shape of the non-straight bonds to the additional physics associated with the internal degrees of freedom.  This issue will be further explored in future work. For now, we only investigate the possible existence of additional solutions other than that given by \eref{eqn:stress_force_general}. We begin by describing the shape of a bond in a more precise way through the following definition.

\begin{definitions} The ``path of interaction'' between any two interacting particles $\alpha$ and $\beta$ is the unique contour that connects $\alpha$ and $\beta$, such that there is a non-zero contribution to the potential part of the pointwise stress, $\bm{\sigma}_{\rm{v}}$, at any point on this contour.
\end{definitions}

In this section, the terms \emph{bond} and \emph{path of interaction} are used synonymously. Therefore, for the case of the pointwise stress tensor in \eref{eqn:stress_force_general}, this path of interaction is given by the straight line joining $\alpha$ and $\beta$. It is shown in Appendix \ref{ch:noll}, that under certain restrictions on the path of interaction, Lemma \ref{lem1} can be generalized to Lemma \ref{lem2} given in Appendix \ref{ch:noll}. Roughly speaking these restrictions are given by the following conditions:
\begin{enumerate}
\item \label{item:bond_shape_1} The shape of the bond connecting particles $\alpha$ and $\beta$ only depends on the distance between $\alpha$ and $\beta$. \footnote{\label{foot:curveshape} This is in essence a constitutive postulate similar to the assumption in pair potentials that the energy depends on only the distance between particles. A more general dependence of the shape of the path of interaction on the environment of a pair of particles might be possible, but is not pursued here. See also Definition~\ref{def:poi} in Appendix~\ref{ch:noll} and following discussion.}
\item For any two pairs of particles $(\alpha,\beta)$ and $(\gamma,\delta)$ separated by the same distance, the bonds $\alpha-\beta$ and $\gamma-\delta$ are related by a rigid body motion. In addition, if $\bm{x}_\alpha - \bm{x}_\beta = \bm{x}_\gamma - \bm{x}_\delta$ then this rigid body motion involves only translation.\label{item:bond_shape_2}
\end{enumerate}
From condition \ref{item:bond_shape_1}, it is clear that the shape of the bonds can be described by contours $\bm{\Upsilon}_{l}:[0,1] \to \real{3}$, for $l>0$, with $\bm{\Upsilon}_l(0)=(0,0,0)$, $\bm{\Upsilon}_l(1)=(l,0,0)$ along with some mild restrictions. Hence, defining the contours $\bm{\Upsilon}_l$ for $l>0$ is equivalent to defining the paths of interaction between any two points in $\real{3}$.  
For a precise definition of $\bm{\Upsilon}_l$ and the paths of interaction see Appendix \ref{ch:noll}. Since all the necessary conditions for the application of Lemma \ref{lem2} are same as those for Lemma \ref{lem1}, we can use Lemma \ref{lem2} in the Irving--Kirkwood--Noll procedure instead of Lemma \ref{lem1}. In doing so, we arrive at a definition for the generalized pointwise stress tensor $\bm{\sigma}_{\rm{v}}^{\rm{G}}$, for given paths of interaction with the above mentioned properties. This is given by
\begin{align}
\label{eqn:stress_force_schofield}
&\bm{\sigma}_{\rm{v}}^{\rm{G}}(\bm{x},t) = \notag \\
& \frac{1}{2} \sum_{\substack{\alpha,\beta \\  \alpha \neq \beta}} \int_{\real{3}} \int_{s=0}^{1} \left \langle \bm{f}_{\alpha\beta} W \mid \bm{x}_\alpha=\bm{x}_{\perp}+s\bm{z},\bm{x}_\beta = \bm{x}_{\perp} - (1-s)\bm{z} \right \rangle \otimes \Qz \bm{\Upsilon}'_{\vnorm{\bm{z}}}(s) \, ds \, d\bm{z},
\end{align}
where $\bm{f}_{\alpha\beta}$ is defined in \eref{eqn:define_fij}, and $\bm{x}_{\perp}(s,\bm{x},\bm{z}) = \bm{x} - s\bm{z}-\Qz\bm{\Upsilon}_{\vnorm{\bm{z}}}(s), \Qz \in \bm{SO}(3)$. Here, $\bm{x}_\perp$ represents the projection of $\bm{x}$ onto the line joining the endpoints, $\bm{x}_\alpha = \bm{x}_\perp + s\bm{z}$ and $\bm{x}_\beta = \bm{x}_\perp - (1-s)\bm{z}$, of the path of interaction being considered. $\Qz$ represents the rotation part of the rigid body motion described in condition \ref{item:bond_shape_2} that maps the contour $\bm{\Upsilon}_{\vnorm{\bm{x}_\alpha - \bm{x}_\beta}}$ to the path of interaction that connects $\bm{x}_\alpha$ and $\bm{x}_\beta$.

\medskip
Equation \eref{eqn:stress_force_schofield} is a general expression for the potential part of the pointwise stress tensor, of which \eref{eqn:stress_force_general} is a special case. We discuss several key features of this expression below:
\begin{enumerate}
\item Equation \eref{eqn:stress_force_schofield} is unique only up to given paths of interaction for a given potential energy extension. It is a more general result than \eref{eqn:stress_force_general}, since \eref{eqn:stress_force_general} can be obtained from \eref{eqn:stress_force_schofield} by assuming that the path of interaction between any two points is the straight line connecting them. For this special case it is easy to see that $\bm{x}_{\perp} = \bm{x}$ and $\Qz \bm{\Upsilon}'_{\vnorm{\bm{z}}}(s) = -\bm{z}$.
\item $\bm{\sigma}_{\rm{v}}^{\rm{G}}$ is in general non-symmetric, whereas the stress obtained through the standard Irving--Kirkwood--Noll procedure is always symmetric for any multi-body potential with an extension. Since the kinetic part of the stress tensor $\bm{\sigma}_{\rm{k}}$ (see \eref{eqn:stress_kinetic}) is symmetric, it follows that the total pointwise stress tensor obtained from the generalized stress tensor is usually non-symmetric. Therefore, under the present setting, the balance of angular momentum is satisfied only through the presence of couple stresses. This suggests that non-straight bonds might correspond to systems with particles having internal degrees of freedom.

\begin{figure}
\centering
\subfigure[]{\includegraphics[totalheight=0.2\textheight]{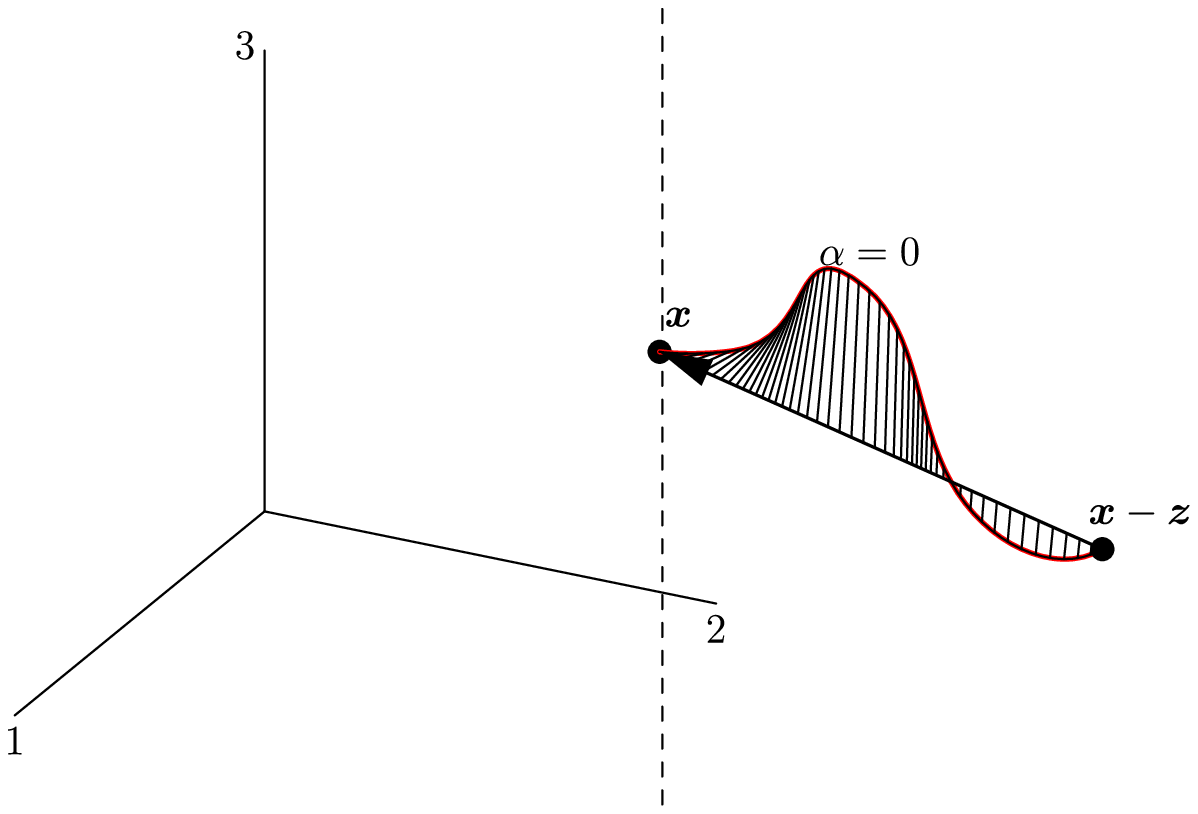}} \\
\subfigure[]{\includegraphics[totalheight=0.2\textheight]{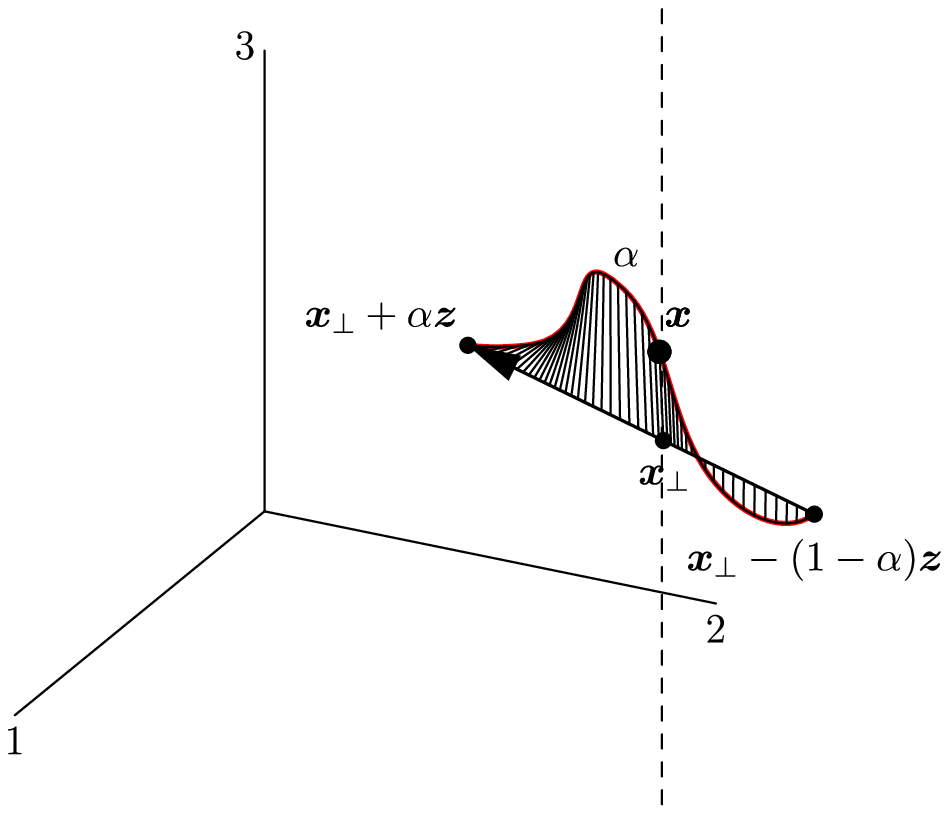}} \\
\subfigure[]{\includegraphics[totalheight=0.2\textheight]{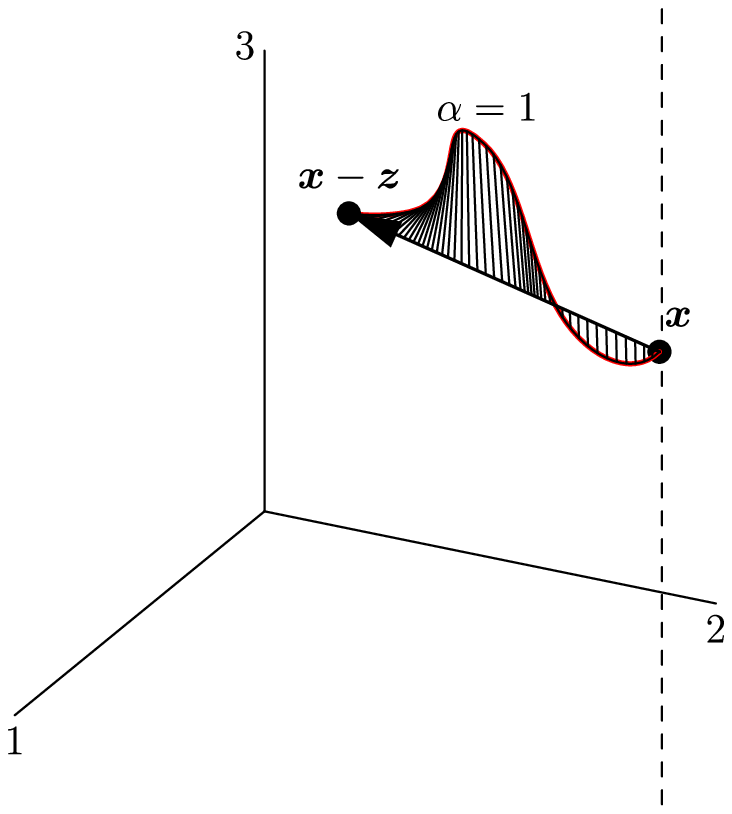}}
\caption{A schematic diagram helping to explain the vectors appearing in the inner integral of \eref{eqn:stress_force_schofield} for a given point $\bm{x}$. The integral in \eref{eqn:stress_force_schofield} in an integral over all possible paths of interaction that pass through the point $\bm{x}$. The inner integral with respect to $s$, with $\bm{z}$ fixed, is an integral over those paths, where $\bm{z}$ is the vector joining its endpoints. Frame (a) shows a path of interaction when $s=0$. As $s$ is increased the path ``slides'' through $\bm{x}$. Frame (b) shows the path for an arbitrary $s$ in the interval $[0,1]$. The end points are represented by $\bm{x}_{\perp}+s\bm{z}$ and $\bm{x}_{\perp}-(1-s)\bm{z}$. Frame (c) shows the position of the path for $s=1$.}
\label{fig:bond}
\end{figure}

\item  Since both \eref{eqn:stress_force_general} and \eref{eqn:stress_force_schofield} are valid definitions for the potential part of the pointwise stress, the question of which one to choose depends on the presence of internal degrees of freedom in each particle. In the absence of internal degrees of freedom only straight bonds are possible due to symmetry. The issue of the pointwise stress being unique only up to a divergence-free tensor-valued function is  partially addressed here, since the expression given by the difference between the two definitions is divergence-free.
\item The expression in \eref{eqn:stress_force_schofield} is very similar to \eref{eqn:stress_force_general}. The pointwise stress at $\bm{x}$ is a superposition of the expectations of the force of all possible \emph{bonds}/\emph{paths of interaction} passing through $\bm{x}$. The vector $\bm{z}$ selects an orientation and the size of the vector connecting the two ends of the bond, and $s$ slides it through $\bm{x}$ from end to end as shown in \fref{fig:bond}. 
\end{enumerate}

\subsection{Definition of the pointwise traction vector}
\label{sec:traction_vector}
In this section, we derive the formula for the pointwise traction vector $\bm{t}(\bm{x},\bm{n};t)$ defined on the surface passing through $\bm{x}$, with normal $\bm{n}$ at time $t$. The following derivation is based on \cite{noll1955} and it can be easily extended to curved paths of interaction. As usual, let $\mathcal{M}$ denote our material system. Let $\Omega \subset \real{3}$ be a domain in three-dimensional space with continuously differentiable surface $\mathcal{S}$, representing a part of the body. By this definition, each of the $N$ point masses described by $\mathcal{M}$ either belong to $\Omega$ or in the space surrounding $\Omega$, denoted by $\Omega^{\rm{c}}$. Let $\bm{f}$ denote the force exerted by the particles in $\Omega^{\rm{c}}$ on particles in $\Omega$. We note that in continuum mechanics $\bm{f}$ is related to $\bm{t}$ by
\begin{equation}
\bm{f}(t) = \int_{\mathcal{S}} \bm{t}(\bm{x},\bm{n},t) \, d\mathcal{S}(\bm{x}),
\label{eqn:continuum_force_1}
\end{equation}
where $\bm{n}(\bm{x})$ is the outer normal at $\bm{x} \in \mathcal{S}$. Using the Cauchy relation, $\bm{t}(\bm{x},\bm{n},t) = \bm{\sigma}(\bm{x},t) \bm{n}$, we obtain
\begin{equation}
\label{eqn:continuum_force_2}
\bm{f}(t) = \int_{\mathcal{S}} \bm{\sigma} \bm{n} \, d\mathcal{S}(\bm{x}).
\end{equation}
Now, note that the net force exerted by $\Omega^{\rm{c}}$ on $\Omega$ due to particle interaction, denoted by $\bm{f}_{\rm{v}}(t)$, is given by
\begin{equation}
\label{eqn:discrete_force_1}
\bm{f}_{\rm{v}}(t) = \sum_{\substack{\alpha,\beta \\  \alpha \neq \beta}} \int_{\bm{u} \in \Omega} \int_{\bm{v} \in \Omega^{\rm{c}}} \left \langle \bm{f}_{\alpha\beta} W \mid \bm{x}_\alpha = \bm{u}, \bm{x}_\beta = \bm{v} \right \rangle \, d\bm{u} \, d\bm{v},
\end{equation}
where $\bm{f}_{\alpha\beta}$ is defined in \eref{eqn:define_fij}.
Since the integrand in \eref{eqn:discrete_force_1} satisfies all the conditions for the application of the lemmas in Appendix \ref{ch:noll}, we can now use a special case of Lemma \ref{lem3} by restricting to straight bonds.\footnote{Specifically, for straight bonds, we set: $\bm{x}_\perp = \bm{x}$ and $\bm{Q}_{\bm{z}}\Upsilon'_{\vnorm{\bm{z}}}(s) = -\bm{z}$.} We therefore have,
\begin{align}
\label{eqn:discrete_force_2}
&\bm{f}_{\rm{v}}(t) = \frac{1}{2} \sum_{\substack{\alpha,\beta \\  \alpha \neq \beta}} \notag \\
& \int_{\mathcal{S}} \int_{\real{3}} \int_{s= 0}^{1} \left \langle -\bm{f}_{\alpha\beta} W \mid \bm{x}_\alpha = \bm{x} + s\bm{z}, \bm{x}_\beta = \bm{x} - (1-s) \bm{z} \right \rangle (\bm{z} \cdot \bm{n}) \, ds\, d\bm{z} \, d\mathcal{S}(\bm{x}).
\end{align}
We now note that $\bm{f}_{\rm{v}}$ in \eref{eqn:discrete_force_2} exactly satisfies
\begin{equation}
\label{eqn:discrete_force_3}
\bm{f}_{\rm{v}}(t) = \int_{\mathcal{S}} \bm{\sigma}_{\rm{v}} \bm{n} \, d\mathcal{S}(\bm{x}),
\end{equation}
where $\bm{\sigma}_{\rm{v}}$ is given by \eref{eqn:stress_force_general}.
It is therefore clear that $\bm{f}_{\rm{v}}$ describes the potential part of the interaction force $\bm{f}$. Hence, it is natural to assign a potential part of the pointwise traction vector, $\bm{t}_{\rm{v}}$ to $\bm{f}_{\rm{v}}$, given by
\begin{align}
\bm{t}_{\rm{v}}&(\bm{x},\bm{n};t) := \bm{\sigma}_{\rm{v}} \bm{n} \notag \\
&= \frac{1}{2} \sum_{\substack{\alpha,\beta \\  \alpha \neq \beta}} \int_{\real{3}} \int_{s= 0}^{1} \left \langle -\bm{f}_{\alpha\beta} W \mid \bm{x}_\alpha = \bm{x} + s\bm{z}, \bm{x}_\beta = \bm{x} - (1-s) \bm{z} \right \rangle (\bm{z} \cdot \bm{n}) \, ds\, d\bm{z}.
\label{eqn:discrete_traction_force}
\end{align}
The above formula is conceptually quite simple. 
\emph{It gives the measure of the force per unit area of all the bonds that cross the surface, where the force is calculated with respect to a surface measure (see {\rm \eref{eqn:discrete_force_2}})}.
Using this viewpoint, we motivate the definitions for the macroscopic traction vector and the stress tensor, when we incorporate spatial averaging in the next section. 

It is now natural to assign the kinetic contribution to the force across the surface to the kinetic part of the pointwise stress tensor. Subtracting \eref{eqn:discrete_force_3} from \eref{eqn:continuum_force_2}, we obtain the kinetic contribution to the force across a surface,
\[
\bm{f}_{\rm{k}}(t) := \int_{\mathcal{S}} (\bm{\sigma} - \bm{\sigma}_{\rm{v}}) \bm{n} \, d\mathcal{S}(\bm{x})
= \int_{\mathcal{S}} \bm{\sigma}_{\rm{k}} \bm{n} \, d\mathcal{S}(\bm{x}).
\]
Therefore the kinetic contribution to the pointwise traction vector $\bm{t}_{\rm{k}}$ is given by
\begin{align}
\bm{t}_{\rm{k}}(\bm{x},\bm{n};t) &:= \bm{\sigma}_{\rm{k}} \bm{n} \notag \\
&= -\sum_\alpha m_\alpha \left \langle \bm{v}_\alpha^{\rm{rel}} (\bm{v}_{\alpha}^{\rm{rel}} \cdot \bm{n}) W \mid \bm{x}_\alpha = \bm{x} \right \rangle. \label{eqn:discrete_traction_kinetic}
\end{align}

Finally, we note that the definitions of $\bm{t}_{\rm{v}}$ and $\bm{t}_{\rm{k}}$ are functions of $\bm{x}$ and $\bm{n}$ alone. Hence, this result is related to the work of Fosdick and Virga \cite{fosdick1989}, who give a variational proof for the stress theorem of Cauchy in the continuum version. In that work the traction vector is allowed to depend on the unit normal and the surface gradient and is shown to be independent of the surface gradient. 

\medskip
The fields defined and derived in this section are pointwise quantities. In the next section, expressions for macroscopic fields are obtained by spatially averaging the pointwise fields over an appropriate macroscopic domain.

\section{Spatial averaging}
\label{ch:spatial}
In the previous section, the Irving--Kirkwood--Noll procedure was used to construct pointwise fields from the underlying discrete microscopic system using the principles of classical statistical mechanics. Although the resulting fields resemble the continuum mechanics fields and satisfy the continuum conservation equations, they are not macroscopic continuum fields. For example, the pointwise stress field in \eref{eqn:stress_force_general}, at sufficiently low temperature, will be highly non-uniform, exhibiting a criss-cross pattern with higher stresses along bond directions, even when macroscopically the material is nominally under uniform or even zero stress.

To measure the fields derived in the previous section in an experiment, one needs a probe which can extract data only from a single point of interest in space. Since this is not possible practically, there is no way we can correlate the experimental data with theoretical predictions. Therefore a true macroscopic quantity is by necessity an average over some spatial region surrounding the continuum point where it is nominally defined.\footnote{We do not include time averaging, because this is indirectly performed due to the presence of $W$. The reasoning for this comes from the \emph{frequentist's} interpretation of probability, wherein the probability of a state is equal to the fraction of the total time spent by the system in that state.} Thus, if $f(\bm{x},t)$ is an Irving--Kirkwood--Noll pointwise field, such as density or stress, the corresponding macroscopic field ${f}_w(\bm{x},t)$ is given by
\begin{equation}
\label{eqn:define_se_0}
f_w(\bm{x},t) = \int_{\real{3}} w(\bm{y} - \bm{x}) f(\bm{y},t) \, d\bm{y},
\end{equation}
where $w(\bm{r})$ is a weighting function representing the properties of the probe and its lengthscale.

The important thing to note is that due to the linearity of the phase averaging in the Irving--Kirkwood--Noll procedure, the averaged macroscopic function ${f}_w(\bm{x},t)$ satisfies the same balance equations as does the pointwise measure $f(\bm{x},t)$.

\subsubsection*{Weighting function} The weighting function $w(\bm{r})$ is an $\mathbb{R}^+$-valued function with compact support so that $w(\bm{r})=0$ for $\vnorm{\bm{r}} > \lambda$, where $\lambda$ is a microstructural lengthscale. The weighting function has units of $\rm{volume}^{-1}$ and must satisfy the normalization condition
\begin{equation}
\label{eqn:normal}
\int_{\real{3}} w(\bm{r}) d\bm{r} = 1.
\end{equation}
This condition ensures that the correct macroscopic stress is obtained when the pointwise stress is uniform. For a spherically-symmetric distribution, $w(\bm{r}) = \hat{w}(r)$, where $r=\vnorm{\bm{r}}$. The normalization condition in this case is
\[
\int_{0}^{\infty} \hat{w}(r)4 \pi r^2 dr = 1.
\]
The simplest choice for $\hat{w}(r)$ is a spherically-symmetric uniform distribution over a specified radius $r_w$, given by
\begin{equation}
\label{eqn:constant_w}
\hat{w}(r) = \left \{ \begin{array}{ll}
1/V_w & \mbox{if $r \leq r_w$},\\
0 & \mbox{otherwise},\end{array} \right. 
\end{equation}
where $V_w = \frac{4}{3}\pi r_{w}^{3}$ is the volume of the sphere. This function is discontinuous at $r=r_w$. If this is a concern, a "mollifying" function that smoothly takes $w(r)$ to zero at $r_w$ over some desired range can be added \cite{murdoch2007}.\footnote{An example of a mollifying function is given later in equation \eref{eqn:weight_molly}.}  Another possible choice for $\hat{w}(r)$ is a Gaussian function \cite{hardy1982}
\begin{equation}
\label{eqn:gaussian_w}
\hat{w}(r) = \pi^{-\frac{3}{2}} r_{w}^{-3} \exp \left [ -r^2/r_{w}^{2} \right ].
\end{equation}
This function does not have compact support. However it decays rapidly with distance so that a numerical cutoff can be imposed where its value drops below a specified tolerance. Another possibility is a quartic spline used in meshless method applications (where it is called a \emph{kernel function} \cite{belytchko1996}),
\begin{equation}
\hat{w}(r) = \left \{ \begin{array}{ll}
\frac{105}{16 \pi r_{w}^{3}}(1+3\frac{r}{r_w})(1- \frac{r}{r_w})^3 & \mbox{if $r \leq r_w$},\\
0 & \mbox{otherwise}.\end{array} \right.
\label{eqn:qspline_w}
\end{equation}
This spline has the advantage that it goes smoothly to zero at $r=r_w$, i.e., $\hat{w}(r_w)=0$, $\hat{w}'(r_w)=0$, and $\hat{w}''(r_w)=0$. \fref{fig:weight} shows the plots of the three weighting functions given above.

\begin{figure}
\centering
\includegraphics[totalheight=0.2\textheight]{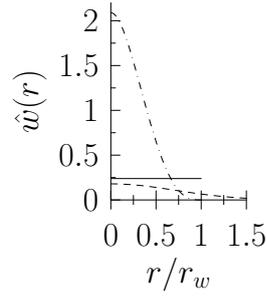}
\caption{Three weighting functions for spatial averaging: uniform weighting (solid line) in \eref{eqn:constant_w}; Gaussian weighting (dashed line) in \eref{eqn:gaussian_w}; Quartic spline weighting (dash-dot line) in \eref{eqn:qspline_w}. Note that the areas under the curves are not equal because the normalization in \eref{eqn:normal} is according to volume.}
\label{fig:weight}
\end{figure}

\subsection{Spatial averaging and macroscopic fields}
\label{sec:spatial_average}
Continuum fields such as density and momentum density fields are defined using \eref{eqn:define_se_0} as the ensemble average via the probability density function $W$, followed by a spatial average via the weight function $w$ as follows:
\begin{align} 
\rho_w(\bm{x},t) &:= \sum_{\alpha} m_\alpha \int_{\real{3}}w(\bm{y}-\bm{x})\langle W\mid \bm{x}_\alpha=\bm{y} \rangle \, d\bm{y},  \label{eqn:define_se_1} \\
\bm{p}_w(\bm{x},t) &:= \sum_{\alpha} m_\alpha \int_{\real{3}}w(\bm{y}-\bm{x}) \langle W \bm{v}_\alpha \mid \bm{x}_\alpha = \bm{y} \rangle \, d\bm{y}.  \label{eqn:define_se_2}
\end{align}
It is straightforward to show that using definitions, \eref{eqn:define_se_1} and \eref{eqn:define_se_2}, the macroscopic version of the generalized pointwise stress tensor given by \eref{eqn:stress_force_schofield} divides into potential and kinetic parts as,
\begin{align}
\stress_{w,\rm{v}}(\bm{x},t) &= \frac{1}{2} \sum_{\substack{\alpha,\beta \\  \alpha \neq \beta}} \int_{\real{3}}{w(\bm{y} - \bm{x})} \int_{\real{3}}  \int_{s=0}^{1} \notag \\
&\langle \bm{f}_{\alpha\beta} W \mid \bm{x}_\alpha=\bm{y}_{\perp} + s \bm{z},\bm{x}_\beta = \bm{y}_{\perp} - (1-s)\bm{z} \rangle \otimes \Qz \bm{\Upsilon}'_{\vnorm{\bm{z}}}(s) \, ds \, d\bm{z} \, d\bm{y} \label{eqn:stress_force_w},
\end{align}
where $\bm{f}_{\alpha\beta}$ is defined in \eref{eqn:define_fij}, $\bm{y}_{\perp} = \bm{y}-s\bm{z} - \Qz \bm{\Upsilon}_{\vnorm{\bm{z}}}(s)$, $\Qz \in \bm{SO}(3)$, and
\begin{equation}
\stress_{w,\rm{k}}(\bm{x},t) = -\sum_{\alpha} \int_{\real{3}} w(\bm{y} - \bm{x}) m_\alpha \langle (\bm{v}_{\alpha}^{\rm{rel}} \otimes \bm{v}_{\alpha}^{\rm{rel}}) W \mid \bm{x}_\alpha = \bm{y} \rangle \, d\bm{y} 
\label{eqn:stress_kinetic_w}.
\end{equation}
We now intend to express the potential part of stress in a more convenient form. This is done by two consecutive changes of variables. Under the assumption that $\Qz$ and $\bm{\Upsilon}_{\vnorm{\bm{z}}}$ are differentiable with respect to $\bm{z}$ and $\vnorm{\bm{z}}$, respectively, the Jacobian of the transformation $(s,\bm{y},\bm{z}) \mapsto (s,\bm{y}_{\perp},\bm{z})$ is unity. Therefore,
\begin{align}
\stress_{w,\rm{v}}(\bm{x},t) &= \frac{1}{2} \sum_{\substack{\alpha,\beta \\  \alpha \neq \beta}} \int_{\real{3}}{w(\bm{y} - \bm{x})} \int_{\real{3}} \int_{s=0}^{1} \notag \\
&\langle \bm{f}_{\alpha\beta} W \mid \bm{x}_\alpha=\bm{y}_{\perp}+s\bm{z},\bm{x}_\beta = \bm{y}_{\perp} - (1-s)\bm{z} \rangle \otimes \Qz \bm{\Upsilon}'_{\vnorm{\bm{z}}} \, ds\, d\bm{z} \, d\bm{y}_{\perp}, \label{eqn:stress_force_w_1}  
\end{align}
where $\bm{y} = \bm{y}(s,\bm{y}_{\perp},\bm{z})$. A second change of variables is introduced as follows
\begin{equation}
\bm{y}_{\perp} + s\bm{z} = \bm{u},\qquad \bm{y}_{\perp} - (1-s)\bm{z} = \bm{v}, \label{eqn:var_change_1}
\end{equation}
which implies,
\begin{equation}
\bm{z}=\bm{u}-\bm{v},\qquad \bm{y}_{\perp}=(1-s)\bm{u} + s\bm{v}. \label{eqn:var_change_2}
\end{equation}
The Jacobian of the transformation is
\begin{equation}
J = \det \left [ \begin{array} {cc}
\nabla_{\bm{u}} \bm{z}  & \nabla_{\bm{v}} \bm{z} \\   
\nabla_{\bm{u}} \bm{y}_{\perp}  & \nabla_{\bm{v}} \bm{y}_{\perp}
\end{array} \right ] = \det \left [ \begin{array}{cc}
  \bm{I}           & -\bm{I}       \\
(1-s) \bm{I}  & s\bm{I} 
\end{array} \right] = 1.
\label{eqn:jacobian} 
\end{equation}
Using \eref{eqn:var_change_1}, \eref{eqn:var_change_2} and \eref{eqn:jacobian} to rewrite \eref{eqn:stress_force_w_1}, we obtain
\begin{equation}
\stress_{w,\rm{v}}(\bm{x},t) = \frac{1}{2} \sum_{\substack{\alpha,\beta \\  \alpha \neq \beta}} \int_{\real{3} \times \real{3}} \langle -\bm{f}_{\alpha\beta} W \mid \bm{x}_\alpha=\bm{u},\bm{x}_\beta=\bm{v} \rangle \otimes \bm{\mathfrak{b}}(\bm{x};\bm{u},\bm{v}) \, d\bm{u} \, d\bm{v},
\label{eqn:stress_force_w_hardy}
\end{equation}
where
\begin{equation}
\label{eqn:bond_vector}
\bm{\mathfrak{b}}(\bm{x};\bm{u},\bm{v}) := -\int_{s=0}^{1} w(\hat{\bm{y}}-\bm{x}) \bm{Q}_{\bm{u}-\bm{v}} \bm{\Upsilon}'_{\vnorm{\bm{u} - \bm{v}}} \, ds
\end{equation}
is called the \emph{bond vector}, with
\begin{equation}
\notag
\hat{\bm{y}}(s,\bm{u},\bm{v}) = \bm{y}(s, \bm{y}_{\perp}(s,\bm{u},\bm{v}),\bm{z}(\bm{u},\bm{v})).
\end{equation}
For the special case of straight bonds, we have 
\[
\hat{\bm{y}} = (1-s) \bm{u} + s \bm{v}
\quad\text{and}\quad
\bm{Q}_{\bm{u}-\bm{v}} \bm{\Upsilon}'_{\vnorm{\bm{u}-\bm{v}}}(s) = -(\bm{u}-\bm{v}). 
\]
Therefore the bond vector simplifies to
\begin{align}
\bm{\mathfrak{b}}(\bm{x};\bm{u},\bm{v}) &=(\bm{u} - \bm{v}) \int_{s=0}^{1} w((1-s) \bm{u} + s\bm{v} - \bm{x}) \, ds \notag \\
&= (\bm{u} - \bm{v}) b(\bm{x};\bm{u},\bm{v}), \label{eqn:bond_function}
\end{align}
where $b(\bm{x};\bm{u},\bm{v})$ is commonly referred to as the \emph{bond function}. The geometrical significance of the bond function is explained in \fref{fig:bond_function}. 

\begin{figure}
\centering
\includegraphics[totalheight=0.2\textheight]{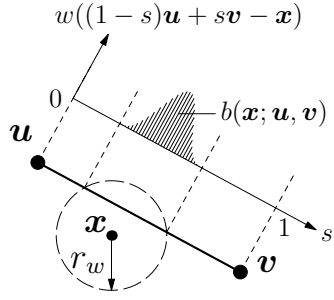}
\caption{The bond function $b(\bm{x};\bm{u},\bm{v})$ is the integral of the weighting function centered at $\bm{x}$ along the line connecting points $\bm{u}$ and $\bm{v}$. The graph shows the result for a quartic spline weighting function. The bond function is the area under the curve.}
\label{fig:bond_function}
\end{figure}

For the special case of straight bonds, equation \eref{eqn:stress_force_w_hardy} simplifies to
\begin{equation}
\label{eqn:stress_force_w_hardy_straight}
\stress_{w,\rm{v}}(\bm{x},t) = \frac{1}{2} \sum_{\substack{\alpha,\beta \\  \alpha \neq \beta}} \int_{\real{3} \times \real{3}} \langle -\bm{f}_{\alpha\beta} W \mid \bm{x}_\alpha=\bm{u},\bm{x}_\beta=\bm{v} \rangle \otimes (\bm{u}-\bm{v}) b(\bm{x};\bm{u},\bm{v}) \, d\bm{u} \, d\bm{v}.
\end{equation}
The expressions for the potential and kinetic parts of the spatially-averaged stress tensor in equations \eref{eqn:stress_kinetic_w} and \eref{eqn:stress_force_w_hardy_straight} are our main result and constitute the general definitions for the macroscopic stress computed for a discrete system.  It will be shown in \sref{ch:compare} that these relations reduce to the Hardy stress tensor \cite{hardy1982} under suitable approximations. The issue of the uniqueness of the stress tensor (in the sense that any divergence-free field can be added to it) is deferred to \sref{sec:unique_macro_stress}.

\subsection{Comparison with the Murdoch--Hardy procedure}
\label{sec:murdoch_proc}
An alternative procedure of defining continuum fields to the one described above, due to Murdoch \cite{murdoch1982,murdoch1993,murdoch1994} and Hardy \cite{hardy1982}, only involves spatial averaging. We refer to this approach as the \emph{Murdoch--Hardy procedure}. Under the Murdoch--Hardy procedure, continuum fields are defined as direct spatial averages of microscopic variables without incorporating statistical mechanics ideas. Therefore, the Murdoch--Hardy procedure is purely deterministic in nature. For example, the density and momentum density fields at a particular instant of time, corresponding to a given weighting function $w$, are defined as
\begin{subequations}
\label{eqn:define_s}
\begin{align} 
\rhomh_w(\bm{x},t) &:= \sum_{\alpha} m_\alpha w( \bm{x}_\alpha(t) - \bm{x}),  \label{eqn:define_s_1} \\
\pmh_w(\bm{x},t) &:= \sum_{\alpha} m_\alpha \bm{v}_\alpha(t) w( \bm{x}_\alpha(t) - \bm{x} ),  \label{eqn:define_s_2}
\end{align}
\end{subequations}
respectively, where $\bm{x}_\alpha$ and $\bm{v}_\alpha$ are deterministic quantities. We denote spatially-averaged variables obtained from the Murdoch--Hardy procedure with a superposed tilde to distinguish them from quantities obtained in \sref{sec:spatial_average}. Equation \eref{eqn:define_s} is used to ``smear'' a discrete system to form a continuum. The reasoning for abandoning statistical mechanics is the lack of knowledge of the ensemble of the system as explained by Murdoch and Bedeaux in \cite{murdoch1993}:
\begin{quotation}
Physical interpretations of any given ensemble average clearly depends on the definition of the ensemble $\dots$ for example, if a container is filled to a given level with water and then poured onto a surface, the lack of precision with which the pouring is effected may result in many different macroscopic flows. Here no single description is available within deterministic continuum mechanics: in this case the ensemble (defined in terms of the water molecules and limited knowledge of how the pouring takes place) relations involve averages associated with all possible flows. Clearly relations involving ensemble averages are associated with a much greater variety of behavior than is describable in terms of deterministic continuum mechanics.
\end{quotation}
We share the same concern regarding the ambiguity in the definition of an ensemble. For example, in an experiment where an austenite-martensite phase transformation occurs, the resulting micro-structure consists of a complex spatial configuration of martensitic variants, and this depends largely on the microscopic details of the system, such as cracks, lattice defects, etc. Therefore, in this case, macroscopic variables cannot completely describe the ensemble of interest. To avoid this difficulty, Murdoch proposes a time average in place of ensemble average. Nevertheless, it should be noted that from classical statistical mechanics, the ensemble of interest and its corresponding distribution exists \emph{in principle}. Therefore the framework described in \sref{sec:spatial_average} is a correct framework in which to phrase the problem. A practical calculation can then be performed, for example, by replacing the ensemble averages with time averages in a molecular dynamics calculation (see \sref{ch:compare}). We stress the importance of writing a continuum field variable as an ensemble average followed by spatial average, rather than a spatial average followed by a time average, as is done in the Murdoch--Hardy procedure, because it helps to give a unified picture of all the previous definitions for continuum fields and stress in particular. This is discussed in the next section.

It is interesting to note that by relaxing the connection with statistical mechanics, the Murdoch--Hardy procedure allows for a much wider class of definitions for the stress tensor \cite{murdoch2007} in addition to the non-uniqueness characterized so far, due to the presence of multiple extensions for the potential energy and allowing non-straight bonds. In this section we intend to systematize this procedure. The source of non-uniqueness resulting from multiple definitions of the stress tensor is studied, thus helping us to identify a much larger class of possible definitions. In this new systematic approach, the steps involved in the Murdoch--Hardy procedure are as follows:
\begin{enumerate}
\item Develop a continuum system by smearing out the discrete system using \eref{eqn:define_s}.
\item Introduce a \emph{non-local} constitutive law for the continuum that is consistent with the discrete version of force balance given later in \eref{eqn:balance}.
\item For each constitutive law, define a stress tensor, which satisfies the equation of motion for the continuum.
\end{enumerate}

To understand the above three steps, we explore the Murdoch--Hardy procedure in more detail. The continuity equation is satisfied in a trivial way \cite{murdoch2003}. We now look at the equation of motion.

\subsubsection*{Equation of motion}
The motion of particle $\alpha$ is governed by Newton's second law,
\begin{equation}
\label{eqn:balance}
\sum_{\substack{\beta \\ \beta \ne \alpha}} \bm{f}_{\alpha\beta}^{\rm d}(t) + \bm{b}_\alpha^{\rm d}(t) = m_\alpha \dot{\bm{v}}_\alpha(t),
\end{equation}
where $\bm{f}_{\alpha\beta}^{\rm d}(t) := \bm{f}_{\alpha\beta}(\bm{u}(t))$, $\bm{f}_{\alpha\beta}(\bm{u})$ are the terms in the central-force decomposition obtained from a multi-body potential with an extension (see \sref{sec:s_motion}) and $\bm{b}_{\alpha}^{\rm d}(t)$ is defined as
\begin{equation}
\notag
\bm{b}_\alpha^{\rm d}(t) := -\nabla_{\bm{x}_\alpha} \pot_{\rm{ext}}(\bm{x}_1(t),\dots,\bm{x}_N(t)).
\end{equation}
The superscript ``$\rm d$'' in \eref{eqn:balance} and the above equation are used to the stress the fact that the quantities are deterministic in nature. Equation \eref{eqn:balance} is a force balance equation for the discrete system. We now design an analogous force balance equation for the smeared continuum defined by \eref{eqn:define_s}, such that \eref{eqn:balance} always holds. 

For the sake of simplicity in notation, from here onwards we use $\bm{f}_{\alpha\beta}$ to denote both $\bm{f}_{\alpha\beta}(\bm{u})$ and $\bm{f}_{\alpha\beta}^{\rm d}(t)$, whenever it is clear from the context. The same goes with the usage of $\bm{b}_\alpha(t)$ for $\bm{b}_{\alpha}^{\rm d}(t)$.

\subsubsection*{Force balance for the smeared continuum}
Multiplying \eref{eqn:balance} by $w(\bm{x}_i-\bm{x})$ and summing over all particles, we have
\begin{equation}
\label{eqn:weighted_balance_1}
\fmh_w + \bmh_w = \sum_{\alpha} m_\alpha \dot{\bm{v}}_\alpha w(\bm{x}_\alpha - \bm{x}),
\end{equation}
where 
\begin{align}
\label{eqn:f_w_1}
\fmh_w(\bm{x},t) &:= \sum_{\substack{\alpha,\beta \\  \alpha \neq \beta}} \bm{f}_{\alpha\beta}(t)w(\bm{x}_\alpha(t) - \bm{x}), \\
\bmh_w(\bm{x},t) &:= \sum_{\beta} \bm{b}_\beta(t) w(\bm{x}_\beta(t) - \bm{x}).
\label{eqn:b_w}
\end{align}
To arrive at a form similar to the equation of motion of continuum mechanics given in \eref{eqn:motion}, equation \eref{eqn:weighted_balance_1} is rewritten as
\begin{align}
\fmh_w + \bmh_w &= \frac{\partial}{\partial t} \sum_{\alpha} m_\alpha w(\bm{x}_\alpha-\bm{x}) \bm{v}_\alpha  - \sum_{\alpha} m_\alpha \bm{v}_\alpha (\nabla w(\bm{x}_\alpha - \bm{x}) \cdot \bm{v}_\alpha) \notag \\
&= \frac{\partial \pmh_w}{\partial t} + \divr_{\bm{x}} \sum_{\alpha} m_\alpha w(\bm{x}_\alpha - \bm{x}) \bm{v}_\alpha \otimes \bm{v}_\alpha.  \label{eqn:weighted_balance_*}
\end{align}
Similar to \eref{eqn:define_velocity}, we define the continuum velocity as 
\begin{equation}
\label{eqn:cont_velocity}
\vmh_w(\bm{x},t) := \frac{\pmh_w(\bm{x},t)}{\rhomh_w(\bm{x},t)},
\end{equation}
and the relative velocity of a particle with respect to the continuum velocity as
\begin{equation}
\label{eqn:rel_velocity}
\vmh_{\alpha}^{\rm{rel}}(\bm{x},t) := \bm{v}_{\alpha}(t) - \vmh_w(\bm{x},t).
\end{equation}
Using \eref{eqn:define_s} and \eref{eqn:rel_velocity}, we obtain
\begin{align}
\sum_{\alpha} m_\alpha \vmh_{\alpha}^{\rm{rel}} w(\bm{x}_{\alpha}(t) - \bm{x}) &= \pmh_w(\bm{x},t) - \rhomh_w(\bm{x},t) \vmh_w(\bm{x},t) \notag \\
&= \bm{0}, \notag
\end{align}
the last equality being true by the definition \eref{eqn:cont_velocity}. From \eref{eqn:rel_velocity} and the above equation, it follows that
\begin{equation}
\notag
\sum_{\alpha} m_\alpha w(\bm{x}_\alpha - \bm{x}) \bm{v}_{\alpha} \otimes \bm{v}_\alpha  = \sum_{\alpha} m_\alpha w(\bm{x}_\alpha - \bm{x}) \vmh_{\alpha}^{\rm{rel}} \otimes \vmh_\alpha^{\rm{rel}}  + \rhomh_w \vmh_w \otimes \vmh_w.
\end{equation}
Substituting this into \eref{eqn:weighted_balance_*} and rearranging, we have
\[
\fmh_w - \divr_{\bm{x}} \sum_{\alpha} m_\alpha (\vmh_\alpha^{\rm{rel}} \otimes \vmh_\alpha^{\rm{rel}}) w(\bm{x}_\alpha - \bm{x}) + \bmh_w = \frac{\partial \pmh_w}{\partial t} + \divr_{\bm{x}} (\rhomh_w \vmh_w \otimes \vmh_w).
\]
Comparing the above equation with the equation of motion, \eref{eqn:motion}, we have
\begin{equation}
\label{eqn:weighted_balance_**}
\divr_{\bm{x}}\stressmh_{w}(\bm{x},t) = \fmh_w(\bm{x},t) - \divr_{\bm{x}} \sum_{\alpha} m_\alpha (\vmh_\alpha^{\rm{rel}} \otimes \vmh_\alpha^{\rm{rel}}) w(\bm{x}_\alpha - \bm{x}),
\end{equation}
where $\stressmh_w$ is the stress tensor corresponding to the weighting function $w$. From \eref{eqn:weighted_balance_**} it is clear that the kinetic part and the potential part of the stress tensor, $\stressmh_{w,\rm{k}}$ and $\stressmh_{w,\rm{v}}$, respectively, are given by
\begin{subequations}
\label{eqn:pde_w_stress}
\begin{align}
\stressmh_{w,\rm{k}}(\bm{x},t) &= - \sum_{\alpha} m_\alpha (\vmh_\alpha^{\rm{rel}} \otimes \vmh_\alpha^{\rm{rel}}) w(\bm{x}_\alpha - \bm{x}), \label{eqn:w_stress_kinetic} \\
\divr_{\bm{x}}\stressmh_{w,\rm{v}}(\bm{x},t) &= \fmh_w(\bm{x},t) \label{eqn:pde_w_stress_force}.
\end{align}
\end{subequations}
Any solution to \eref{eqn:pde_w_stress_force} is a valid candidate for the definition of $\stressmh_{w,\rm{v}}$. Murdoch \cite{murdoch2007} proposes several possible candidates, and highlights the possibility of having multiple definitions. To understand the connection between the different possible definitions, we look back at \eref{eqn:weighted_balance_1} and \eref{eqn:f_w_1}. Equation \eref{eqn:weighted_balance_1} is a force balance equation for any ``continuum particle'' at $\bm{x}$, and $\fmh_w$, defined in \eref{eqn:f_w_1}, is the force per unit volume acting on it. It is not immediately clear from \eref{eqn:f_w_1} how two continuum particles at positions $\bm{x}$ and $\bm{y}$ interact with each other. This interaction can be given by a non-local constitutive law. The main idea is to recast \eref{eqn:f_w_1} as
\begin{equation}
\label{eqn:interaction}
\fmh_w(\bm{x},t) = \int_{\real{3}} \bm{g}(\bm{x},\bm{y},t) \, d\bm{y},
\end{equation}
for some $\bm{g}(\bm{x},\bm{y},t)$, which we call the \emph{generator of the non-local constitutive law}. This function describes the interaction between the continuum particles at $\bm{x}$ and $\bm{y}$. To satisfy Newton's third law, we also need $\bm{g}$ to be anti-symmetric with respect to its arguments $\bm{x}$ and $\bm{y}$. Unfortunately the representation given in \eref{eqn:interaction} is not unique and, since every choice of $\bm{g}$ leads to a different stress definition, this is one of the sources of non-uniqueness in the definition for the stress tensor in the Murdoch--Hardy procedure. We describe two different constitutive laws, which lead to the Hardy stress and the doubly-averaged stress (\murdoch\ stress)\footnote{Murdoch \cite{murdoch2007} refers to this stress as ``Noll's choice''. To avoid confusion with the stress derived through the Irving--Kirkwood--Noll procedure in \sref{ch:phase}, we name it the ``\murdoch\ stress''.} \cite{murdoch1994}.

\medskip
For the case of Hardy stress, the generator $\bm{g}^{\rm{H}}$ is given by the equation
\begin{equation}
\label{eqn:g_1}
\bm{g}^{\rm{H}}(\bm{x},\bm{y},t) = \sum_{\substack{\alpha,\beta \\  \alpha \neq \beta}} \bm{f}_{\alpha\beta} w(\bm{x}_\alpha - \bm{x}) \delta(\bm{x}_\beta - \bm{x}_\alpha + \bm{x} - \bm{y}),
\end{equation}
where $\delta$ denotes the Dirac delta distribution.

\medskip
The generator $\bm{g}^{\rm{D}}$ for the \murdoch\ stress is given by
\begin{equation}
\label{eqn:g_2}
\bm{g}^{\rm{D}}(\bm{x},\bm{y},t) = \sum_{\substack{\alpha,\beta \\  \alpha \neq \beta}} \bm{f}_{\alpha\beta} w(\bm{x}_\alpha - \bm{x}) w(\bm{x}_\beta - \bm{y}).
\end{equation}

\begin{figure}
\begin{center}
\subfigure[]{\label{fig:hardy}\includegraphics[totalheight=0.2\textheight]{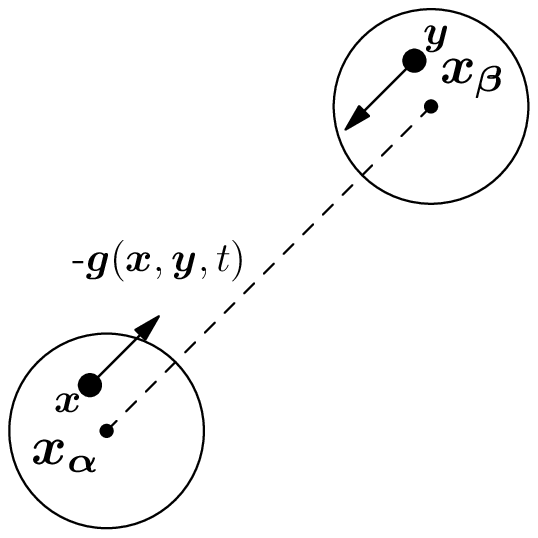}}
\subfigure[]{\label{fig:murdoch}\includegraphics[totalheight=0.2\textheight]{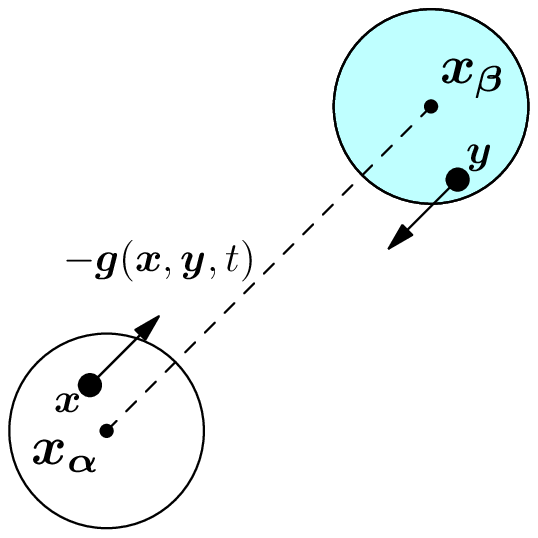}}
\subfigure[]{\label{fig:admal}\includegraphics[totalheight=0.2\textheight]{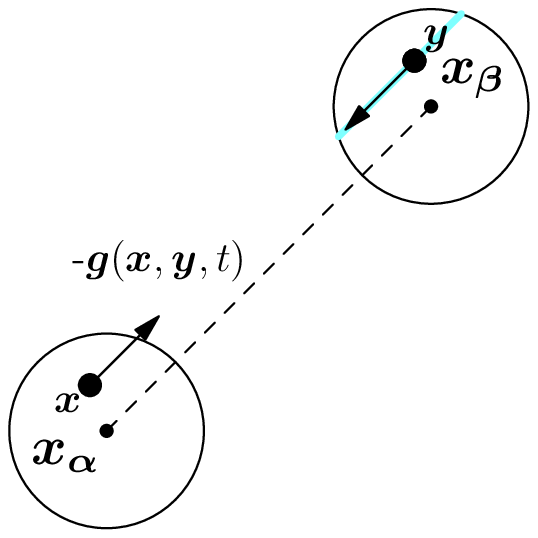}}
\end{center}
\caption{A continuum particle $\bm{x}$ interacts with: (a) only that continuum particle at $\bm{y}$, which is identically oriented to $\bm{x}_\beta$ as $\bm{x}$ is oriented to $\bm{x}_\alpha$, when the interaction is given by $\bm{g}^{\rm{H}}$; (b) any continuum particle in the shaded region, when the interaction is given by $\bm{g}^{\rm{D}}$;  (c) any continuum particle on the shaded line, when the interaction is given by $\bm{g}^{\rm{HD}}$.}
\label{fig:generators}
\end{figure}

\smallskip
\fref{fig:generators} shows the interaction between two continuum particles, with positions $\bm{x}$ and $\bm{y}$, that are in a neighborhood of two interacting particles $\alpha$ and $\beta$ respectively and not in a neighborhood of any other particle in the system. In this setup, it is clear from the generator for the Hardy stress, given in \eref{eqn:g_1}, that two continuum particles at $\bm{x}$ and $\bm{y}$ interact only when $\bm{y}-\bm{x}_\beta = \bm{x} - \bm{x}_\alpha$, as shown in \fref{fig:hardy}. On the other hand, there is no such restriction on the generator for the \murdoch\ stress, described by \eref{eqn:g_2} (see \fref{fig:murdoch}).\footnote{Note that for the \murdoch\ stress, the force between two continuum particles at $\bm{x}$ and $\bm{y}$ is \emph{not} parallel to $\bm{x} - \bm{y}$ in general. This is not a violation of  the strong law of action and reaction, because the strong law only applies to discrete systems. It has been used in this derivation by requiring that $\bm{f}_{\alpha\beta} = -\bm{f}_{\beta\alpha}$ and $\bm{f}_{\alpha\beta} \times (\bm{x}_\alpha - \bm{x}_\beta) = \bm{0}$.}  Although at this point there is no systematic way of suggesting additional possible generators, we can suggest a third generator, $\bm{g}^{\rm{HD}}$, which has properties that lie in between $\bm{g}^{\rm{H}}$ and $\bm{g}^{\rm{D}}$. As shown in \fref{fig:admal}, when interaction is governed by $\bm{g}^{\rm{HD}}$, a continuum particle $\bm{x}$ interacts with $\bm{y}$ only when $\bm{y}$ lies on the line passing through $\bm{x}$ and parallel to $\bm{x}_\alpha - \bm{x}_\beta$. In all the three cases, the interaction force is always directed along the vector $\bm{x}_\alpha - \bm{x}_\beta$. Therefore by \eref{eqn:interaction} we have three different integral representations for $\fmh_w$ with generators $\bm{g}^{\rm{D}}$, $\bm{g}^{\rm{H}}$, $\bm{g}^{\rm{HD}}$.

Now, in each of the integral representations of $\fmh_w$ given by \eref{eqn:g_1} and \eref{eqn:g_2}, the integrand satisfies all the necessary conditions for the application of Lemma \ref{lem1} or Lemma \ref{lem2} in Appendix \ref{ch:noll}.\footnote{The integrand should be continuously differentiable for Noll's lemma to be applicable. Although $\bm{g}^{\rm{H}}$ is not continuously differentiable due to the presence of the Dirac delta distribution, this does not hinder us from applying the lemma since we can replace the Dirac delta distribution by an appropriate infinitely differentiable delta sequence and take a limit. See Appendix \ref{sec:hardy_limit} for a rigorous derivation of this.\label{fn:mollifier}} For instance, using Lemma \ref{lem1} we obtain an expression for the potential part of the stress tensor, given by
\begin{equation}
\label{eqn:w_stress_force}
\stressmh_{w,\rm{v}}(\bm{x},t) = -\frac{1}{2} \int_{\real{3}} \left [ \int_{s= 0}^{1} \bm{g}(\bm{x}+s\bm{z}, \bm{x}-(1-s)\bm{z}, t) \, ds\right] \otimes \bm{z} \, d\bm{z}.
\end{equation}
Substituting \eref{eqn:g_1} into \eref{eqn:w_stress_force}, we have the potential part of the Hardy stress:
\begin{align}
\stressmh_{w,\rm{v}}^{\rm{H}} &= -\frac{1}{2} \sum_{\substack{\alpha,\beta \\  \alpha \neq \beta}} \int_{\real{3}} \left [ \int_{s= 0}^{1} \bm{f}_{\alpha\beta} w(\bm{x}_\alpha-\bm{x} - s \bm{z}) \delta(\bm{x}_\beta - \bm{x}_\alpha + \bm{z}) \, ds \right ] \otimes \bm{z} \, d\bm{z} \notag \\
&= \frac{1}{2} \sum_{\substack{\alpha,\beta \\  \alpha \neq \beta}} \int_{s= 0}^{1} [-\bm{f}_{\alpha\beta} w((1-s)\bm{x}_\alpha + s\bm{x}_\beta-\bm{x}) \otimes (\bm{x}_\alpha - \bm{x}_\beta)] \, ds. \label{eqn:hardy_stress}
\end{align}
Substituting \eref{eqn:g_2} into \eref{eqn:w_stress_force} we have the potential part of the \murdoch\ stress:
\begin{equation}
\label{eqn:murdoch_stress}
\stressmh_{w,\rm{v}}^{\rm{\murdoch}} =  \frac{1}{2} \sum_{\substack{\alpha,\beta \\ \alpha \ne \beta}} \int_{z \in \real{3}} \int_{s=0}^{1} [-\bm{f}_{\alpha\beta} w(\bm{x}_\alpha - \bm{x} - s\bm{z}) w(\bm{x}_\beta - \bm{x} + (1-s)\bm{z}) \otimes \bm{z}]  \, ds \, d\bm{z},
\end{equation}
which was derived by Murdoch \cite{murdoch1994}. The conclusion is that the non-uniqueness of the generator in the systematized Murdoch--Hardy procedure leads to a non-unique definition for the stress tensor. Further sources of non-uniqueness can by introduced by having a different force decomposition corresponding to a different potential extension, or using curved paths of interaction instead of straight bonds and applying Lemma \ref{lem2}, which is a generalization of Lemma \ref{lem1} in Appendix \ref{ch:noll}. We do not pursue this generalization further here.

\medskip
It is important to point out that the systematized Murdoch--Hardy procedure presented here does \emph{not} describe all possible solutions to \eref{eqn:pde_w_stress_force}. An example of a solution that cannot be obtained via our systematized Murdoch--Hardy procedure is the following definition suggested in \cite{murdoch2007}:
\begin{equation}
\stressmh_{w,\rm{v}}^*(\bm{x},t):=\sum_{\substack{\alpha,\beta \\  \alpha \neq \beta}} \bm{f}_{\alpha\beta} \otimes (\bm{x} - \bm{x}_\alpha) \hat{a}(\vnorm{\bm{x}-\bm{x}_\alpha}),
\label{eqn:mur2}
\end{equation}
where $\hat{a}(u) := \frac{1}{u^3} \int_{0}^{u} s^2 \hat{w}(s) \, ds$. As is pointed out in \cite{murdoch2007}, the expansion in \eref{eqn:mur2} is not a physically-relevant definition for stress due to the following test case. Consider a stationary deformed body at zero temperature (i.e., where the particles occupy fixed positions without vibrating). In this case, the net force acting on any particle is zero. Since $\bm{f}_{\alpha\beta}$ is the only term in the summand of \eref{eqn:mur2} which depends on $\beta$, \eref{eqn:mur2} is equivalent to 
\begin{equation}
\stressmh_{w,\rm{v}}^*(\bm{x},t):=\sum_\alpha (\sum_{\substack{\beta \\ \beta \ne \alpha}} \bm{f}_{\alpha\beta}) \otimes (\bm{x} - \bm{x}_\alpha) \hat{a}(\vnorm{\bm{x}-\bm{x}_\alpha}).
\end{equation}
In our case, $\sum_{\beta}\bm{f}_{\alpha\beta} = \bm{f}_\alpha = \bm{0}$, for each particle $\alpha$ in the \emph{interior} of the body which is considerably away from the surface compared to the interatomic distance. Hence, the only non-zero contribution to the stress is due to those particles close to the surface on which the net force due to other particles is non-zero. Moreover, although $\hat{a}(u)$ decays to zero as $u$ increases, $\hat{a}(u) \ne 0$ for all $u \ne 0$. Thus, there is a non-zero stress at every point $\bm{x} \in \real{3}$ (even outside the body!) due to particles close to the surface of the body, which is obviously not physically reasonable. Nevertheless, $\stressmh_{w,\rm{v}}^*$ is mathematically still a valid definition since it satisfies the force balance equation \cite{murdoch2007}. However, it cannot be derived using the systematized Murdoch--Hardy procedure proposed here. Thus, the systematized Murdoch--Hardy procedure does not lead to all possible definitions that satisfy the force balance equation.

\medskip
Finally, it is also worth noting that the fact that all balance laws are satisfied under the Murdoch--Hardy procedure should not come as a surprise, since $w$ in \eref{eqn:define_s} serves the same purpose as $W$ in \eref{eqn:define_density}. In this view, $w$ is seen as a function defined on a phase space, although one that does not evolve according to a flow described by Hamilton's equations of motion, but still satisfies \eref{eqn:liouville}.\footnote{This is only a mathematical argument. No physical significance should be drawn from this analogy.} The corresponding flow in phase space is described as follows. Continuing with our notation introduced  \eref{eqn:define_X}, let $\bm{\Xi}(0)=(\bm{x}^{\rm s}(0),\bm{v}^{\rm s}(0))=(\bm{x}_1^{\rm s}(0),\dots,\bm{x}_N^{\rm s}(0),\bm{v}_1^{\rm s}(0),\dots,\bm{v}_N^{\rm s}(0))$ denote any arbitrary point in phase space. We add a superscript ``$\rm s$'' to stress the fact that an element in phase space is stochastic in nature.  Consider the flow in phase space given by the mapping 
\begin{align}
\notag
\bm{\Xi}(0) = (\bm{x}^{\rm s}(0),\bm{v}^{\rm s}(0)) \mapsto &(\bm{x}^{\rm s}(0)+\bm{x}(t)-\bm{x}(0), \bm{v}^{\rm s}(0)+\bm{v}(t)-\bm{v}(0)) \\
&=(\bm{x}^{\rm s}(t),\bm{v}^{\rm s}(t)) = \bm{\Xi}(t),
\label{eqn:flowmap}
\end{align}
where the quantities $\bm{x}(t) = (\bm{x}_1(t),\cdots,\bm{x}_N(t)), \bm{v}(t) = (\bm{v}_1(t),\cdots,\bm{v}_N(t))$ denote the position and velocity of the particle and these are assumed to be known. (Typically these quantities are obtained from a molecular dynamics simulation.) Therefore the Murdoch--Hardy procedure can be interpreted as a probabilistic model constructed from the data, $\bm{x}(t)$ and $\bm{v}(t)$, obtained from a deterministic model -- a molecular dynamics simulation. Note that $\bm{x}^{\rm s}$ and $\bm{v}^{\rm s}$ in \eref{eqn:flowmap} denote the positions and velocities of the particles in the probabilistic model.  Then it is easy to see that if $W(\bm{\Xi};t)$ is given by
\begin{equation}
W(\bm{\Xi};t) = w(\bm{x}_1(t) - \bm{x}_1^{\rm s}) \cdots w(\bm{x}_N(t) - \bm{x}_N^{\rm s}),
\label{eqn:W_w}
\end{equation}
then the definitions given by \eref{eqn:define_density}, \eref{eqn:define_mom_density} and \eref{eqn:define_s} are consistent and $W$ given by the above formula satisfies Liouville's equation (see \eref{eqn:liouville}), which was used in deriving the balance equations in \sref{ch:phase}. Note that unlike \sref{ch:phase}, $W(\bm{\Xi};t)$ defined in \eref{eqn:W_w} is \emph{not} a probability density function. (Its integral over phase space diverges, since it is independent of $\bm{v}$.) The key difference between the two approaches is that all quantities in the Irving--Kirkwood--Noll procedure are probabilistic, while this is not true for the Murdoch--Hardy procedure, if the above probabilistic interpretation is adopted. For example, $\bm{f}_{\alpha\beta}$ in the Murdoch--Hardy procedure is deterministic. Therefore the structure inherent in \eref{eqn:stress_force_differential_**} through the marginal densities is absent in the Murdoch--Hardy procedure, thus giving additional non-uniqueness.  It is shown in \sref{ch:compare} that the Hardy stress can be derived using both approaches, while the \murdoch\ stress is a result of the Murdoch--Hardy procedure alone.

\subsection{Definition of the spatially-averaged traction vector}
\label{sec:define_traction_weight}
We close this section by defining the spatially-averaged traction vector, $\bm{t}_w(\bm{x},\bm{n};t)$, for a weighting function $w$, at a point $\bm{x}$ relative to a plane with normal $\bm{n}(\bm{x})$.  One possibility is to adopt the Cauchy relation using the spatially-averaged stress tensor,
\begin{equation}
\label{eqn:traction_pw}
\bm{t}_{w}(\bm{x},\bm{n};t):=\stress_{w}(\bm{x},t) \bm{n}.
\end{equation}
However, since $\stress_{w}$ is defined as a volume average, we immediately see that with this definition $\bm{t}_{w}$ depends not only on the bonds that cross the surface, but also on nearby bonds that do not cross it \cite{murdoch2003}. Hence, equation \eref{eqn:traction_pw} does not appear to be consistent with Cauchy's definition of traction. 

We therefore seek an alternative definition for the spatially-averaged traction vector. In \sref{sec:traction_vector}, we showed that the pointwise traction vector at a point on a surface is the expectation of the force per unit area of all the bonds that cross the surface, making it a property of the surface. We would like the spatially-averaged traction to have the same property.  We therefore define it as an average over a \emph{surface} rather than over a volume as for the stress.  For simplicity, we consider the weighting function $w_h$, defined to be constant on the averaging domain, which is taken to be a generalized cylinder of height $h$, with its axis parallel to $\bm{n}$ and enclosing $\bm{x}$. The traction $\bm{t}_w(\bm{x},\bm{n};t)$ is defined as 
\begin{equation}
\label{eqn:discrete_traction_w}
\bm{t}_{w}(\bm{x},\bm{n};t) := \lim_{h \to 0} \stress_{w_h}(\bm{x},t)\bm{n},
\end{equation}
where $\stress_{w_h}$ is the stress associated with the weighting function, $w_h$. In a more general case, an arbitrary averaging domain can be collapsed onto a surface passing through $\bm{x}$, in many ways. Although this can be made mathematically more precise, we do not pursue that in this work. Definition \eref{eqn:discrete_traction_w} has a two-fold advantage over the definition in \eref{eqn:traction_pw}:
\begin{enumerate}
\item The traction vector is defined to be non-local on a surface, thus making it a property of the surface. This is physically more meaningful, and closer to the continuum definition.
\item The above definition differs from the traction definition in \eref{eqn:traction_pw}, because only the bonds which cross the surface contribute to the traction field.
\end{enumerate}
In \sref{sec:tsai}, we use definition \eref{eqn:discrete_traction_w} to define the Tsai traction starting with the spatial averaging discussed in \sref{sec:spatial_average} and in this way establish a link between the Tsai traction and the Irving--Kirkwood--Noll procedure. 

\section{Derivation of different stress definitions and the issue of uniqueness}
\label{ch:compare}
In this section, we systematically derive various stress tensors commonly found in the literature from the methods developed in \sref{ch:phase} and \sref{ch:spatial}. The stress tensors discussed in this section are the Hardy, virial and \murdoch\ stress tensors and the Tsai traction. 

\subsection{Hardy stress tensor}
\label{sec:hardy}
The Murdoch--Hardy procedure described in \sref{ch:spatial} was independently developed by Murdoch \cite{murdoch1982} and Hardy \cite{hardy1982}. The motivation for Hardy's study was to test the validity of the continuum description of phenomena in shock waves. The formulas suggested by Irving and Kirkwood were not useful due to the lack of knowledge regarding the probability density function and the infinite series expansion in the definition of the stress. As an alternative, Hardy used what we now term as the ``Murdoch--Hardy procedure'' to propose an instantaneous definition for stress, for the special case of pair potential, given by
\begin{subequations}
\label{eqn:hardy}
\begin{align}
\stress^{\rm H}_{\rm{v}}(\bm{x},t) &= \frac{1}{2} \sum_{\substack{\alpha,\beta \\  \alpha \neq \beta}} \frac{(\bm{x}_\alpha(t)-\bm{x}_\beta(t)) \otimes (\bm{x}_\alpha(t)-\bm{x}_\beta(t))}{\vnorm{\bm{x}_\alpha-\bm{x}_\beta}} \pot'_{\alpha\beta} b(\bm{x};\bm{x}_\alpha,\bm{x}_\beta), \label{eqn:hardy_force}\\
\stress^{\rm H}_{\rm{k}}(\bm{x},t) &= -\sum_{\alpha} w(\bm{x}_\alpha(t) - \bm{x}) m_\alpha \bm{v}_{\alpha}^{\rm{rel}}(t) \otimes \bm{v}_{\alpha}^{\rm{rel}}(t),
\label{eqn:hardy_kinetic}
\end{align}
\end{subequations}
where $b$ is the bond function defined in \eref{eqn:bond_function} and $\bm{v}_\alpha^{\rm{rel}}$ is the velocity of particle $\alpha$ with respect to the continuum velocity, as defined in \eref{eqn:cont_velocity}. To simplify the notation, the explicit dependence of $\bm{x}_\alpha$ and $\bm{v}^{\rm{rel}}_\alpha$ on time is dropped from here onwards. Equations \eref{eqn:hardy_force} and \eref{eqn:hardy_kinetic} may look familiar. They are similar to the spatially-averaged generalized stress in \eref{eqn:stress_kinetic_w} and \eref{eqn:stress_force_w_hardy_straight} (for the special case of a pair potential). If in these relations, the ensemble average is replaced by a time average, we obtain a time-averaged Hardy stress. However, in performing such an operation, we must note the following:
\begin{enumerate}
\item Under conditions of thermodynamic equilibrium (see footnote~\ref{foot:tdequil} on page~\pageref{foot:tdequil}), ensemble averages can be replaced by time averages provided that the system is assumed to be ergodic. Strictly speaking this time average should be done for infinite time, but for practical reasons we are restricted to finite time.
\item The Hardy stress tensor is valid under non-equilibrium conditions assuming that the system is in {\em local thermodynamic equilibrium}\footnote{Local thermodynamic equilibrium is a weaker condition than uniform thermodynamic equilibrium (see footnote \ref{foot:tdequil}). The assumption is that 
the microscopic domain associated with each continuum particle is locally in a state of uniform
thermodynamic (or at least metastable) equilibrium. This is the reason why concepts like temperature can be defined as field variables in continuum mechanics. See for example \cite{evansmorriss}.} at all points at every instant of time. This is plausible only when there is a clear separation of time scales between the microscopic equilibration time scale $\tau$ and macroscopic times. Here, $\tau$ is not being defined rigorously. Roughly speaking, $\tau$ must be sufficiently small so that macroscopic observables do not vary appreciably over it.
\end{enumerate}
Under these assumptions, we may replace ensemble averages with time averages in \eref{eqn:stress_kinetic_w} and \eref{eqn:stress_force_w_hardy_straight}  to obtain
\begin{subequations}
\label{eqn:hardy_stress_tavg}
\begin{align}
\stress_{w,\rm{k}}(\bm{x},t) &= -\frac{1}{\tau} \sum_{\alpha} \int_{t}^{t+\tau} w(\bm{x}_\alpha - \bm{x}) m_\alpha \bm{v}_{\alpha}^{\rm{rel}} \otimes \bm{v}_{\alpha}^{\rm{rel}} dt, 
\label{eqn:hardy_stress_kin_tavg}
\\
\stress_{w,\rm{v}}(\bm{x},t) &= \frac{1}{2\tau}\sum_{\substack{\alpha,\beta \\ \alpha \neq \beta}} \int_{t}^{t+\tau} [-\bm{f}_{\alpha\beta} \otimes (\bm{x}_\alpha-\bm{x}_\beta) b(\bm{x};\bm{x}_\alpha,\bm{x}_\beta)] dt 
\label{eqn:hardy_stress_force_tavg}, 
\end{align}
\end{subequations}
where $\bm{f}_{\alpha\beta}$, corresponding to a given potential extension, is defined in \eref{eqn:define_fij}, and $\tau$ represents a microscopic time scale. We see that the Hardy stress is obtained through a rigorous process beginning with the statistical mechanics concepts introduced in \sref{ch:phase}. From here on, we will denote the stress in \eref{eqn:hardy_stress_tavg} as the ``Hardy stress'', although we note that this definition constitutes a generalization of the original Hardy stress to arbitrary potentials and includes time averaging.  Incidentally, the Hardy stress can also be derived from the systematized Murdoch--Hardy procedure described in \sref{sec:murdoch_proc}. The kinetic part of the Hardy stress is the same as that obtained in the Murdoch--Hardy procedure. The potential part of the Hardy stress is derived using the generator $\bm{g}^{\rm{H}}$ given in \eref{eqn:g_1}. This was done in \sref{ch:spatial} (see \eref{eqn:hardy_stress}).  

Note that the Hardy stress tensor is symmetric. One could modify this to a general form by choosing an arbitrary path of interaction, thus leading to a non-symmetric form (see \sref{sec:murdoch_proc}). Also note that the stress tensor resulting from the generator $\bm{g}^{\rm{HD}}$ (see \fref{fig:generators}) would be symmetric, because the interaction force between two continuum particles is always aligned with the line connecting them. It is very important to observe that under non-equilibrium conditions, where we assume a local thermodynamic equilibrium at every instant of macroscopic time, we may assume that the averaging domain centered at a position $\bm{x}$ moves with the continuum velocity $\bm{v}(\bm{x},t)$. This fact will be used in \sref{sec:tsai}.  

\subsection{Tsai traction}
\label{sec:tsai}
Cauchy's original definition of stress emerges from the concept of traction acting across the internal surfaces of a solid via the bonds that cross the surface. It is therefore natural to attempt to define traction at the atomic level in a similar vein in terms of the force in bonds intersecting a given plane. This approach actually goes back to Cauchy himself as part of his effort in the 1820s  to define the stress in crystalline systems \cite{Cauchy1828a,Cauchy1828b}, which is described in detail in Note B in Love's classical book on the theory of elasticity \cite{love}. Cauchy's derivation is limited to zero temperature equilibrium where the atoms are stationary. This approach was extended by Tsai \cite{tsai1979} to the dynamical setting by also accounting for the momentum flux of atoms moving across the plane. The expression for the traction given in \cite{tsai1979} appears to be based on intuition. 

In this section, we show how the Tsai traction can be systematically derived from the Hardy stress tensor, which itself was derived from the generalized stress tensor defined in \sref{sec:spatial_average}.  We will see that the potential part of Tsai's original definition agrees with the results of our unified framework. However, Tsai's expression for the kinetic part of the traction depends on the absolute velocity of the particles and therefore is not invariant with respect to Galilean transformations.  We show below that the correct expression for the Tsai traction vector $\bm{t}(\bm{x},\bm{n};t)$ across a plane $P$ with normal $\bm{n}$ is
\begin{align}
\bm{t}_{w}(\bm{x},\bm{n};t) &= \frac{1}{A\tau} \int_{t}^{t+\tau} \sum_{\alpha\beta \cap P} \bm{f}_{\alpha\beta} \frac{(\bm{x}_\alpha - \bm{x}_\beta) \cdot \bm{n}} {\abs{(\bm{x}_\alpha-\bm{x}_\beta) \cdot \bm{n}}} \,dt \notag \\
&-\frac{1}{A\tau} \sum_{\alpha \leftrightarrow P} \frac{m_\alpha \bm{v}_\alpha^{\rm{rel}}(t_\leftrightarrow)( \bm{v}_{\alpha}^{\rm{rel}}(t_\leftrightarrow) \cdot \bm{n} ) }{\abs{\bm{v}_{\alpha}^{\rm{rel}}(t_\leftrightarrow) \cdot \bm{n} }}, 
\label{eqn:tsai_traction}
\end{align}
where $\tau$ indicates the microscopic time scale, $\sum_{\alpha\beta \cap P}$ indicates the summation over all bonds $\alpha-\beta$ crossing the plane $P$, $\sum_{\alpha \leftrightarrow P}$ indicates summation over all particles that cross $P$ in the time interval $[t,t+\tau]$,\footnote{A particle is counted multiple times if it crosses the plane multiple times.} $\bm{v}_{\alpha}^{\rm{rel}}$ denotes the local relative velocity of particle $\alpha$, and $t_\leftrightarrow$ indicates the time at which the particle crosses the plane. The correct form for $\bm{v}_{\alpha}^{\rm{rel}}$ is not immediately obvious. Below, we derive equation \eref{eqn:tsai_traction} and obtain an explicit expression for $\bm{v}_{\alpha}^{\rm{rel}}$.

\medskip
We start with the Hardy stress in \eref{eqn:hardy_stress_tavg}. Recall from \sref{sec:define_traction_weight} that if the averaging domain is taken to be a generalized cylinder $\mathcal{C}_h$ of height $h$, the spatially-averaged traction field, $\bm{t}_w(\bm{x},\bm{n};t)$, on a surface passing through $\bm{x}$ with normal $\bm{n}$ is
\begin{align}
\bm{t}_w(\bm{x},\bm{n};t) = \lim_{h \to 0} \stress_{w_h} \bm{n} &= \lim_{h \to 0} (\stress_{w_h,\rm{v}} \bm{n} + \stress_{w_h,\rm{k}} \bm{n}) \label{eqn:traction} \\
&=: \bm{t}_{w,\rm{v}} + \bm{t}_{w,\rm{k}}.  \notag
\end{align}
Using \eref{eqn:hardy_stress_tavg}, we rewrite the potential part and kinetic part of \eref{eqn:traction} as 
\begin{subequations}
\begin{align}
\bm{t}_{w,\rm{v}}(\bm{x},\bm{n};t) &= \frac{1}{2\tau} \lim_{h \to 0} \int_{t}^{t+\tau} \sum_{\substack{\alpha,\beta \\ \alpha \neq \beta}} [-\bm{f}_{\alpha\beta} \otimes (\bm{x}_\alpha - \bm{x}_\beta) b_h(\bm{x};\bm{x}_\alpha,\bm{x}_\beta)] dt \label{eqn:traction_contact},\\
\bm{t}_{w,\rm{k}}(\bm{x},\bm{n};t) &= -\frac{1}{\tau} \lim_{h \to 0} \int_{t}^{t+\tau} \sum_{\alpha} m_\alpha w(\bm{x_\alpha} - \bm{x};h)  \bm{v}_{\alpha}^{\rm{rel}}(t;h) ( \bm{v}_{\alpha}^{\rm{rel}}(t;h) \cdot \bm{n} ) dt \label{eqn:traction_kin},
\end{align}
\end{subequations}
where $b_h$ denotes the bond function for a generalized cylinder of height $h$. Also note the dependence of $\bm{v}_{\alpha}^{\rm{rel}}$ on $h$ in \eref{eqn:traction_kin}. Let us first consider the potential part of the traction in \eref{eqn:traction_contact}. As $h$ approaches zero, the generalized cylinder will no longer contain complete bonds. Assuming a constant weighting function, the bond function $b_h$ equals the fraction of the length of the bond lying within the generalized cylinder per unit volume:
\begin{equation}
b_h(\bm{x};\bm{x}_\alpha,\bm{x}_\beta) = \frac{1}{hA}\frac{h}{\abs{\left (\bm{x}_\alpha - \bm{x}_\beta \right) \cdot \bm{n}}} =\frac{1}{A \abs{\left (\bm{x}_\alpha - \bm{x}_\beta \right) \cdot \bm{n}}},
\end{equation}
for any bond $\alpha-\beta$ crossing the cylinder. Therefore \eref{eqn:traction_contact} takes the form
\begin{equation}
\bm{t}_{w,\rm{v}}(\bm{x},\bm{n};t) = \frac{1}{A\tau} \int_{t}^{t+\tau} \sum_{\alpha\beta \cap P} \left [-\bm{f}_{\alpha\beta} \frac{(\bm{x}_\alpha - \bm{x}_\beta) \cdot \bm{n}}{\abs{(\bm{x}_\alpha - \bm{x}_\beta) \cdot \bm{n}}}\right ] \, dt.
\end{equation}
Note that the $1/2$ factor is dropped because of the definition of the summation in the above equation. This is the first term in \eref{eqn:tsai_traction}. Turning to the kinetic part of the traction in \eref{eqn:traction_kin}, we interchange the summation and integral to obtain
\begin{equation}
\bm{t}_{w,\rm{k}}(\bm{x},\bm{n};t) = -\frac{1}{\tau} \lim_{h \to 0} \sum_{\alpha \in \mathcal{C}_h} \int_{t_1(\alpha;h)}^{t_2(\alpha;h)}  m_\alpha w(\bm{x_\alpha} - \bm{x};h)  \bm{v}_{\alpha}^{\rm{rel}}(t;h) ( \bm{v}_{\alpha}^{\rm{rel}}(t;h) \cdot \bm{n} ) dt,
\end{equation}
where $t_1(\alpha;h)$ and $t_2(\alpha;h)$ are the times of entry and exit of particle $\alpha$, respectively, from a cylinder of height $h$. The summation in the above equation is over all particles that are in the generalized cylinder during the time interval $[t,t+\tau]$, with a particle counted $k$ times if it enters and exits the cylinder $k$ times. Multiplying and dividing the above equation by $t_2(\alpha;h)-t_1(\alpha;h)$ and substituting in $w$, we have
\begin{align}
\bm{t}_{w,\rm{k}}(\bm{n}) &= -\frac{1}{A\tau} \lim_{h \to 0} \sum_{\alpha \in \mathcal{C}_h} \frac{t_2(\alpha;h)-t_1(\alpha;h)}{h} \frac{\int_{t_1(\alpha;h)}^{t_2(\alpha;h)}  m_\alpha \bm{v}_{\alpha}^{\rm{rel}}(t;h) ( \bm{v}_{\alpha}^{\rm{rel}}(t;h) \cdot \bm{n} ) dt}{t_2(\alpha;h)-t_1(\alpha;h)} \notag  \\
& = -\frac{1}{A\tau} \sum_{\alpha \leftrightarrow P}\lim_{h \to 0} \frac{t_2(\alpha;h)-t_1(\alpha;h)}{h} m_\alpha \bm{v}_{\alpha}^{\rm{rel}}(t_\leftrightarrow) ( \bm{v}_{\alpha}^{\rm{rel}}(t_\leftrightarrow) \cdot \bm{n}),
\label{eqn:twk}
\end{align}
where we have used the Lebesgue differentiation theorem \cite{folland} in the last equality. Note that the interchange of limit and summation in the above step is valid since we can assume that the summation for any $\mathcal{C}_h$ is a finite summation which is physically meaningful. Since the averaging domain moves with a continuum velocity we note that
\begin{equation}
\lim_{h \to 0} \frac{t_2(\alpha;h)-t_1(\alpha;h)}{h} = \frac{1}{\abs{\bm{v}_{\alpha}^{\rm{rel}}(t_\leftrightarrow) \cdot \bm{n}}}.
\label{eqn:toh}
\end{equation}
In words, this equality states that the net time spent by particle $\alpha$ in the cylinder, divided by its height, is equal to the inverse of the velocity of particle $\alpha$ along the axis of the cylinder. This is correct in the limit, $h\to 0$, where particles only enter and exit the cylinder at its ends.  Substituting \eref{eqn:toh} into \eref{eqn:twk}, we have
\begin{align}
\bm{t}_{w,\rm{k}}(\bm{n}) &= -\frac{1}{A\tau} \sum_{\alpha \leftrightarrow P} \frac{m_\alpha \bm{v}_\alpha^{\rm{rel}}(t_\leftrightarrow)( \bm{v}_{\alpha}^{\rm{rel}}(t_\leftrightarrow) \cdot \bm{n} ) }{\vert \bm{v}_{\alpha}^{\rm{rel}}(t_{\leftrightarrow}) \cdot \bm{n} \vert} \notag \\
&= -\frac{1}{A\tau} \sum_{\alpha \leftrightarrow P} m_\alpha \bm{v}_\alpha^{\rm{rel}}(t_{\leftrightarrow}) \textrm{sign}( \bm{v}_{\alpha}^{\rm{rel}}(t_\leftrightarrow) \cdot \bm{n} ).
\end{align}
This is the second term in \eref{eqn:tsai_traction}. Note that 
\begin{equation}
\bm{v}_{\alpha}^{\rm{rel}}(t) = \lim_{h \to 0}\bm{v}_{\alpha}^{\rm{rel}}(t;h) = \bm{v}_\alpha(t) - \lim_{h \to 0}\bm{v}(\bm{x};h).
\end{equation}
Hence, we have implicitly assumed that $\lim_{h \to 0}\bm{v}(\bm{x};h)$ is well-defined for our averaging domain (plane $P$) which is a limit of the generalized cylinder $\mathcal{C}_h$. In the following calculation, we show that $\bm{v}(\bm{x};h)$ is well-defined and its exact form is derived.

We know that for a generalized cylinder 
\begin{align}
\bm{v}(\bm{x};h) &:= \frac{\frac{1}{\tau} \int_{t}^{t+\tau} \sum_{\alpha} m_\alpha w(\bm{x}_\alpha(t) - \bm{x};h) \bm{v}_\alpha(t) dt}{\frac{1}{\tau}\int_{t}^{t+\tau} \sum_{\beta} m_\beta w(\bm{x}_\beta(t)-\bm{x}) dt} \notag\\ 
&= \frac{ \sum^{'}_{\alpha} \int_{t_1(\alpha;h)}^{t_2(\alpha;h)}  m_\alpha \bm{v}_\alpha(t) dt}{\sum^{'}_{\beta} \int_{t_1(\beta;h)}^{t_2(\beta;h)}  m_\beta dt} \notag \\
&= \frac{\sum^{'}_{\alpha} m_\alpha \left [ \bm{x}_\alpha(t_2(\alpha;h)) - \bm{x}_\alpha(t_1(\alpha;h)) \right ]}{\sum^{'}_{\beta} m_\beta \left [ t_2(\beta;h) - t_1(\beta;h) \right ]},\label{eqn:cont_velocity_h} 
\end{align}
where $\sum'$ indicates summation over those particles that cross $P$ in the time interval $[t,t+\tau]$, including multiple entries and exits.

Considering the limit $h \to 0$ of the first partial fraction of the last equation we have
\begin{equation}
\lim_{h \to 0} \frac{ m_\alpha \frac{\bm{x}_\alpha(t_2(\alpha;h)) - \bm{x}_\alpha(t_1(\alpha;h))}  {t_2(\alpha;h) - t_1(\alpha;h)} }  { \sum_{\beta}^{'} m_\beta \frac{t_2(\beta;h) - t_1(\beta;h)}  {t_2(\alpha;h) - t_1(\alpha;h)} } = \frac{m_\alpha \bm{v}_\alpha( t_\leftrightarrow)}  {\sum^{'}_{\beta} m_\beta \abs{ \bm{v}_\alpha(t_\leftrightarrow) \cdot \bm{n} / \bm{v}_\beta(t_\leftrightarrow) \cdot \bm{n} }},
\label{partialfraction}
\end{equation}
using the fact that in the limit $h \to 0$, ($t_2(\beta;h) - t_1(\beta;h))/(t_2(\alpha;h) - t_1(\alpha;h))$, which is the ratio of the times spent by particles $\beta$ and $\alpha$ in one of their sojourns into the cylinder, is equal to the inverse ratio of their normal velocities. Using \eref{partialfraction} and taking the limit $h \to 0$ of \eref{eqn:cont_velocity_h}, we obtain
\begin{equation}
\bm{v}(\bm{x}) = \lim_{h \to 0} v(\bm{x};h) = \sum_{\alpha \leftrightarrow P} \frac{m_\alpha \bm{v}_\alpha(t_\leftrightarrow)}  {\sum^{'}_{\beta} m_\beta \abs{ \bm{v}_\alpha(t_\leftrightarrow) \cdot \bm{n} / \bm{v}_\beta(t_\leftrightarrow) \cdot \bm{n}}}.
\end{equation}
Note that the above expression for the continuum velocity is far from intuitive. One might expect the continuum velocity to be the average velocity of particles crossing the surface, but this is not true. It is clear from the above equation that the averaging is not trivial.

\medskip
From the relationship between the Tsai traction in \eref{eqn:tsai_traction} and the Hardy stress tensor in \eref{eqn:hardy_stress_tavg}, it is apparent that the Tsai traction is a more local quantity than the Hardy stress tensor. The Tsai traction performs better than the Hardy stress in systems with free surfaces. This was studied by Cheung and Yip \cite{cheung1991} for a one-dimensional case, in which virial stress and Tsai stress are compared (the virial stress is a special case of Hardy stress as shown in the next section). 

The Tsai traction definition can be used to evaluate the stress tensor at a point by evaluating the traction on three perpendicular planes.\footnote{For example, if the normals to the planes are aligned with the axes of a Cartesian coordinate system with basis vectors $\bm{e}_i$, then $\bm{t}(\bm{e}_1)$ would give the components $\sigma_{11}$, $\sigma_{21}$, $\sigma_{31}$, $\bm{t}(\bm{e}_2)$ would give the components $\sigma_{12}$, $\sigma_{22}$, $\sigma_{32}$, and $\bm{t}(\bm{e}_3)$ would give the components $\sigma_{13}$, $\sigma_{23}$, $\sigma_{33}$.} However, it is not clear from the perspective put forward by Tsai \cite{tsai1979} whether the resulting stress tensor would be symmetric or even well-defined, i.e., it is not clear if another choice of planes will give suitably transformed components of the same stress tensor. Our derivation suggests that a stress tensor constructed from the Tsai traction should be well-defined and symmetric, at least in a weak sense, since it is a limit of the Hardy stress, which has these properties. The numerical experiments presented in \sref{ch:experiment}, suggest that the Tsai traction is invariant with respect to the position of the Tsai plane $P$ and the resulting stress tensor is symmetric. 

\subsection{Virial stress tensor}
\label{sec:virial}
In this section, we show that the virial stress tensor derived in \sref{ch:canonical} and in Appendix \ref{ch:virial} can be re-derived from the time-averaged version of the Hardy stress given in \eref{eqn:hardy_stress_tavg}. The expression for the virial stress tensor is obtained from \eref{eqn:hardy_stress_tavg} as a special case for a weighting function which is constant on its support. The bond function, $b$, in \eref{eqn:hardy_stress_tavg} is evaluated approximately using its definition \eref{eqn:bond_function} by only counting those bonds $\alpha-\beta$ that lie entirely within the averaging domain and neglecting the bonds that cross the averaging domain. Hence, $b(\bm{x};\bm{x}_\alpha,\bm{x}_\beta)$ is given by
\begin{equation}
\label{eqn:bond_function_virial}
b(\bm{x};\bm{x}_\alpha,\bm{x}_\beta) = \left \{ \begin{array}{ll}
1/\vol(\Omega_{\bm{x}}) & \mbox{if bond $\alpha-\beta \in \Omega_{\bm{x}}$}, \\
0  & \mbox{otherwise}, \end{array} \right.
\end{equation}
where $\Omega_{\bm{x}}$ denotes the averaging domain centered at $\bm{x}$. Substituting \eref{eqn:bond_function_virial} into \eref{eqn:hardy_stress_tavg}, we have
\begin{equation}
\stress(\bm{x},t) = \frac{1}{\tau \vol(\Omega_{\bm{x}})} \int_{t}^{t+\tau} \Bigg [ -\sum_{\alpha \in \Omega_{\bm{x}}} m_\alpha \bm{v}_{\alpha}^{\rm{rel}} \otimes \bm{v}_{\alpha}^{\rm{rel}} + \frac{1}{2} \sum_{\substack{\alpha,\beta \in \Omega_{\bm{x}} \\ \alpha \ne \beta}} [-\bm{f}_{\alpha\beta} \otimes (\bm{x}_\alpha - \bm{x}_\beta)] \Bigg ] \,dt ,
\label{eqn:virial_stress_tavg}
\end{equation}
which is identical to \eref{eqn:virial_2} in Appendix \ref{ch:virial}. It is clear from this that the virial stress tensor is only an approximation and tends to the Hardy stress as the volume of the averaging domain is increased. This is because the ratio of the measure of bonds that cross the surface to those which are inside the averaging domain decreases as the size of the domain increases. The difference between the virial stress tensor and the Tsai traction was analytically calculated for a one-dimensional chain by Tsai (see \cite{tsai1979}). Since, taking the averaging domain size to infinity is equivalent to taking the thermodynamic limit in this context, the Hardy and virial stress expressions become identical in this limit. Since the virial theorem was also derived in \sref{ch:canonical} for the case of equilibrium statistical mechanics, it follows that the Irving--Kirkwood--Noll procedure is consistent with the results of equilibrium statistical mechanics in the thermodynamic limit. 

\subsection{\murdoch\ stress tensor}
It was seen in \sref{sec:murdoch_proc}, that the \murdoch\ stress tensor, defined in \eref{eqn:murdoch_stress}, is derived using an appropriate generator in the systematized Murdoch--Hardy procedure. However, unlike the Hardy stress, the \murdoch\ stress cannot be derived from the Irving--Kirkwood--Noll procedure. It is also worth noting that the stress tensor given by \eref{eqn:murdoch_stress} is in general non-symmetric, and only under very special conditions yields a symmetric tensor \cite{murdoch2007}. 

\subsection{Uniqueness of the macroscopic stress tensor}
\label{sec:unique_macro_stress}
Three possible sources of non-uniqueness for the stress tensor have been identified in our discussion:
\begin{enumerate}
\item Given that there are multiple potential extensions (see Page~\pageref{page:altext}), different force decompositions are possible and hence different pointwise stress tensors can be obtained.
\item For a given pointwise stress tensor, a new pointwise stress, which also satisfies the balance of linear momentum, can be obtained by adding on an arbitrary tensor field with zero divergence.
\item The generalization of the Irving--Kirkwood--Noll procedure in \sref{sec:gen_stress} to arbitrary ``paths of interaction'' leads to the possibility of non-symmetric expressions for the pointwise stress tensor.
\end{enumerate}

We address the first two issues in this section. The third source of non-uniqueness is only possible in systems where the discrete particles making up the system possess internal structure, such as internal polarization or spin. For systems of discrete particles without internal structure only straight bonds are possible due to symmetry arguments. We leave the discussion of particles with internal structure to future work.

\subsubsection*{Uniqueness and potential energy extensions}
The first source of non-uniqueness of the stress tensor is related to the potential energy extension discussed in \sref{sec:s_motion}. We show below that the macroscopic stress tensor, calculated as a spatial average of the pointwise stress tensor with constant weighting function, is always unique in the thermodynamic limit (see footnote~\ref{foot:tdlimit} on page~\pageref{foot:tdlimit}), i.e., the difference between the spatially-averaged pointwise stress tensors resulting from two different extensions tends to zero, as the volume of the averaging domain is increased.

The discussion below is limited to 5-body potentials since it can be easily extended to any interatomic potential. We first show that the contribution due to any cluster of $5$ particles within the averaging domain is zero. Without loss of generality, we may assume that our system consists of $5$ particles interacting with an interatomic potential energy given by
\begin{equation}
\pot_{\rm{int}} = \potfisher(\bm{x}_1,\dots,\bm{x}_5).
\end{equation}
Let $\pot_{\rm{int}}(\zeta_{12},\dots,\zeta_{45})$ and $\pot^*_{\rm{int}}(\zeta_{12},\dots,\zeta_{45})$ be two different extensions of $\pot_{\rm int}$ from the shape space $\mcal{S}$ to $\mathbb{R}^{10}$ (see \sref{sec:s_motion}), and for any $\bm{s}=(r_{12},\dots,r_{45}) \in \mcal{S}$, let
\begin{align}
\bm{f}_{\alpha\beta}(\bm{x}_1,\dots,\bm{x}_5) &:= \frac{\partial \pot_{\rm{int}}}{\partial \zeta_{\alpha\beta}}(\bm{s}) \frac{\bm{x}_\beta - \bm{x}_\alpha}{r_{\alpha\beta}}, \\
\bm{f}^*_{\alpha\beta} (\bm{x}_1,\dots,\bm{x}_5)&:= \frac{\partial \pot^*_{\rm{int}}}{\partial \zeta_{\alpha\beta}}(\bm{s})\frac{\bm{x}_\beta - \bm{x}_\alpha}{r_{\alpha\beta}},
\end{align}
be their corresponding force decompositions. Let $\stress$ and $\stress^*$ denote the resulting pointwise stress tensors in the Irving--Kirkwood--Noll procedure from $\pot_{\rm{int}}$ and $\pot^*_{\rm{int}}$, respectively. Let $\Omega_{\bm{x}}$ denote the averaging domain\footnote{For simplicity assume that the averaging domain is convex.} centered at $\bm{x}$ that is used to calculate the Hardy stress tensor. Using \eref{eqn:hardy_stress_force_tavg} and noting that all the bonds lie within $\Omega$, the difference between the Hardy stress tensors resulting from these two representations, for the special case of a constant weighting function, is given by
\begin{equation}
\label{eqn:delta_stress}
\Delta \stress(\bm{x},t) := \stress_w - \stress^*_w = \frac{1}{2\tau \vol(\Omega_{\bm{x}})}\sum_{\substack{\alpha,\beta \\ \alpha \neq \beta}} \int_{t}^{t+\tau} [-\Delta \bm{f}_{\alpha\beta} \otimes (\bm{x}_\alpha-\bm{x}_\beta)] dt,
\end{equation}
where $\Delta \bm{f}_{\alpha\beta} := \bm{f}_{\alpha\beta} - \bm{f}^*_{\alpha\beta}$. 

\begin{figure}
\centering
\includegraphics[scale=0.7]{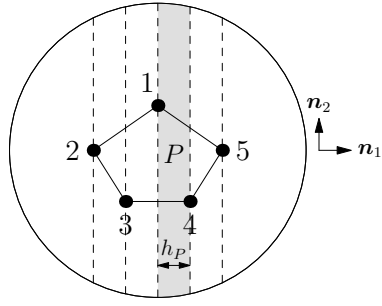}
\caption{A cluster of $5$ particles that lie completely inside the averaging domain, does not contribute to the ambiguity in the stress tensor.}
\label{fig:delta_sigma}
\end{figure}

We would like to show that $\Delta \stress \bm{n}_1=\bm{0}$, where $\bm{n}_1$ is the normal vector as shown in \fref{fig:delta_sigma}. The essential idea to is to interchange the integration and summation in \eref{eqn:delta_stress} and split the terms appearing in the summation into fractions, such that each fraction yields a zero contribution to $\Delta\stress\bm{n}_1$.  In order to show this, we partition the averaging domain into regions such that no region contains a particle in its interior and the partition surfaces are perpendicular to the normal (see \fref{fig:delta_sigma}). Let $h_P$ denote the width of the partition $P$. Using \eref{eqn:delta_stress}, we can now write $\Delta \stress \bm{n}_1$ as
\begin{align}
\Delta \stress \bm{n}_1 &= \frac{1}{2 \tau \vol(\Omega_{\bm{x}})} \int_t^{t+\tau} \sum_{P} \sum_{\alpha\beta \cap P} \left [-\Delta \bm{f}_{\alpha\beta} (\bm{x}_\alpha - \bm{x}_\beta) \cdot \bm{n}_1\frac{h_P}{\abs{(\bm{x}_\alpha - \bm{x}_\beta) \cdot \bm{n}_1}} \right], \notag  \\
&= \frac{1}{2 \tau \vol(\Omega_{\bm{x}})} \int_t^{t+\tau} \sum_{P} h_P 
\Delta\bm{F}_P,
\end{align}
where $\sum_{\alpha\beta \cap P}$ denotes the summation over the bonds crossing 
partition $P$, and
\begin{equation}
\Delta\bm{F}_P =
\sum_{\alpha\beta \cap P} \left [-\Delta \bm{f}_{\alpha\beta} \frac{(\bm{x}_\alpha - \bm{x}_\beta) \cdot \bm{n}_1}{\abs{(\bm{x}_\alpha - \bm{x}_\beta) \cdot \bm{n}_1}} \right ]
\end{equation}
is the \emph{net} force on particles on one side of the partition due to particles on the other side. Since both representations give the same total force on each particle, the force difference, or net force, on each particle is zero and therefore, $\Delta\bm{F}_P=\bm{0}$.  For example, for the partition shown in the figure, 
\begin{equation}
\Delta\bm{F}_P = -2(\Delta \bm{f}_{51} + \Delta \bm{f}_{43}).
\end{equation}
Since $\Delta \bm{f}_{45} = - \Delta \bm{f}_{54}$, we have
\begin{align}
\Delta\bm{F}_P &= -2(\Delta \bm{f}_{51} + \Delta \bm{f}_{43} + \Delta \bm{f}_{45} + \Delta \bm{f}_{54})  \notag \\
&= -2(\Delta \bm{f}_{43} + \Delta \bm{f}_{45}) - 2(\Delta \bm{f}_{51} + \Delta \bm{f}_{54}) \notag \\
&= -2\Delta \bm{f}^{\rm int}_4 - 2\Delta \bm{f}^{\rm int}_5 = \bm{0} + \bm{0} = \bm{0}.
\end{align}
Hence, $\Delta\stress \bm{n}_1 = \bm{0}$. Undertaking a similar argument in the other directions, we see that $\Delta\stress \bm{n}_i=\bm{0}$. These results together imply that $\Delta\stress=\bm{0}$.  Given this, we can conclude that any cluster of particles that lies entirely within the averaging domain does not contribute to the spatial average of the difference between two stress definitions.  Consequently, the only non-zero contribution comes from those clusters for which the bonds connecting its particles cross the averaging domain. Since this contribution scales as surface area, it tends to zero as volume tends to infinity.

\subsubsection*{Uniqueness and the addition of a divergence-free field to the stress}
The second source of non-uniqueness of the stress tensor involves the addition to it of a divergence-free field. This issue is partly addressed by the result (shown in \sref{sec:virial}) that the spatially-averaged pointwise stress converges to the virial stress in the thermodynamic limit (see footnote~\ref{foot:tdlimit} on page~\pageref{foot:tdlimit}). Consider the pointwise stress, $\stress$, obtained through the Irving--Kirkwood--Noll procedure, which satisfies the balance of linear momentum, and a new pointwise stress, $\hat{\stress}=\stress+\tilde{\stress}$, where $\divr_{\bm{x}}\tilde{\stress}=\bm{0}$. Clearly, $\hat{\stress}$ also satisfies the balance of linear momentum and is therefore also a valid solution. The spatially-averaged stress obtained from the new definition is
\begin{equation}
\hat{\stress}_w(\bm{x},t) 
= \int_{\real{3}} w(\bm{y} - \bm{x}) \hat{\stress}(\bm{y},t) \, d\bm{y}
= \int_{\real{3}} w(\bm{y} - \bm{x}) 
(\stress(\bm{y},t)+\tilde{\stress}(\bm{y},t))\, d\bm{y}.
\label{eqn:tildesigw}
\end{equation}
We showed in \sref{sec:virial} that in the thermodynamic limit, the spatially-averaged pointwise stress, $\stress$, converges to the virial stress. We also expect $\hat{\stress}_w$ to equal the virial stress in this limit (since any macroscopic stress must converge to this value under equilibrium conditions). Therefore, \eref{eqn:tildesigw} reduces to
\begin{equation}
\lim_{\rm TD} \int_{\real{3}} w(\bm{y} - \bm{x}) \tilde{\stress}(\bm{y},t)\, d\bm{y}=\bm{0},
\label{eqn:tildesigconstraint}
\end{equation}
where $\lim_{\rm TD}$ refers to the thermodynamic limit. Equation~\eref{eqn:tildesigconstraint} places a strong constraint on allowable forms for $\tilde{\stress}$, the implications of which are left for future work.

\section{Numerical Experiments}
\label{ch:experiment}
In this section, we describe several numerical experiments, involving molecular dynamics and lattice statics simulations, conducted to capture differences in the spatially-averaged stress measures derived in \sref{ch:compare}.  We consider the Hardy stress defined in \eref{eqn:hardy_stress_tavg}, the Tsai traction defined in \eref{eqn:tsai_traction}, the virial stress defined in \eref{eqn:virial_stress_tavg} and the \murdoch\ stress defined in \eref{eqn:murdoch_stress}.  We will sometimes refer to these as the ``microscopic definitions'' or the ``microscopically-based stress tensors''. 

\subsection{Experiment 1}
\label{sec:exp_1}
We begin with the study of the kinetic part of the stress tensor. From the discussion in \sref{ch:compare}, it is clear that unlike the definition for the potential part of the stress tensor, there is no ambiguity in the definition for the kinetic part of stress. However, the kinetic part of the stress may appear to be at odds with the continuum definition of stress that is stated solely in terms of the forces acting between different parts of the body. The need for the kinetic part of stress becomes apparent when considering an ideal gas, where the potential interaction term is zero by definition and therefore it is the kinetic term that is wholly responsible for the transmission of pressure.

\begin{figure}
\centering
\includegraphics[totalheight=0.3\textheight]{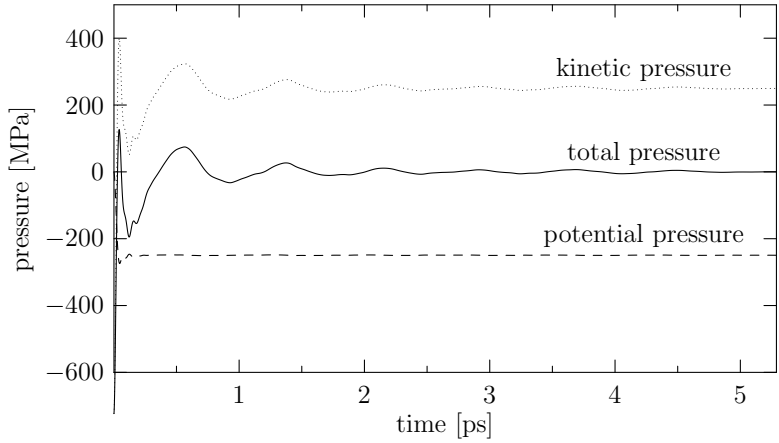}
\caption{The virial pressure as a function of time is plotted for an isolated cube of aluminum at $300\rm{K}$. The total pressure/virial pressure is the sum of the kinetic and potential pressures.}
\label{fig:isolated}
\end{figure}

To demonstrate that the kinetic term in the stress tensor does indeed exist, we perform the following constant energy molecular dynamics simulation of an isolated cube. The cube, consisting of $4000$ aluminum atoms in a face-centered cubic (fcc) arrangement ($10\times10\times10$ unit cells), is floating freely in a vacuum. The atoms interact according to an EAM potential for aluminum due to Ercolessi and Adams \cite{ercolessi1994}:
\begin{align}
\poteam_{\rm int} = \frac{1}{2}\sum_{\substack{\alpha,\beta \\ \alpha \ne \beta}} \pot_{\alpha\beta}(r_{\alpha\beta}) + \sum_\alpha \mcal{U}_\alpha(\rho_\alpha), \quad \rho_\alpha=\sum_{\substack{\beta \\ \beta \ne \alpha}} f_\beta(r_{\alpha\beta}).
\end{align}
Here $\mcal{U}_\alpha$, called the \emph{embedding function}, is the energy required to embed particle $\alpha$ in the electron density, $\rho_\alpha$, due to the surrounding particles, and $f_\beta(r_{\alpha\beta})$ is the electron density of particle $\beta$ at $\bm{x}_\alpha$. The initial positions of the atoms are randomly perturbed by a small amount relative to their zero temperature equilibrium positions and the system is evolved by integrating the equations of motion. The initial perturbation is adjusted so that the temperature of the cube is about $300\rm{K}$ (small fluctuations in temperature are expected since temperature is not controlled in the simulation). Since the block is unconstrained, we expect the stress, $\stress$, in the box and consequently the pressure, defined by $p = - \frac{1}{3}\tr \stress$, to be zero. The virial expression for calculating the pressure follows from \eref{eqn:virial_1} as
\begin{equation}
\label{eqn:virial_pressure}
p= \frac{1}{3V}\Bigg [ \sum_\alpha m_\alpha \ol{\vnorm{\bm{v}_\alpha}^2} - \frac{1}{2} \sum_{\substack{\alpha,\beta \\ \alpha \ne \beta}} \ol{\vnorm{\bm{f}_{\alpha\beta}} r_{\alpha\beta}} \Bigg ],
\end{equation}
where
\begin{align}
\bm{f}_{\alpha\beta} = \frac{\partial \poteam_{\rm int}}{\partial r_{\alpha\beta}} \frac{\bm{x}_\beta-\bm{x}_\alpha}{r_{\alpha\beta}}.
\end{align}

The three curves shown in \fref{fig:isolated} are the potential and kinetic parts of the pressure and the total pressure as a function of time, calculated using \eref{eqn:virial_pressure}. As expected the total pressure tends to zero as the system equilibrates. However, the potential and kinetic parts are \emph{non-zero}, converging to values that are equal and opposite such that their sum is zero. More interestingly, the kinetic part is not insignificant for our system. This clearly shows that kinetic part cannot be neglected even when considering solid systems. This can be quantified by noting that the kinetic part in \eref{eqn:virial_pressure} is simply the temperature per unit volume given by the equipartition theorem \cite{huang}, $k_B T = 2 \mcal{T}/3N$, where $\mcal{T}$ is the kinetic energy. Therefore \eref{eqn:virial_pressure} reduces to
\begin{equation}
p= \frac{1}{V}\Bigg [ Nk_BT - \frac{1}{6} \sum_{\substack{\alpha,\beta \\ \alpha \ne \beta}} \ol{\pot'_{\alpha\beta} r_{\alpha\beta}} \Bigg ].
\end{equation}
For example at 300~K, $\kb T = 0.02585~\rm{eV}$. The lattice spacing for the system considered is equal to $4.032 \ang$. Hence, the volume per atom is $V/N = 4.023^3/4 = 16.387\ang$ and the kinetic pressure is $1.577~\rm{meV}/\ang$. This translates to $252.394$ MPa, which is a considerable stress.

\subsection{Experiment 2}
\label{sec:exp2}
It is clear from \sref{sec:exp_1}, that the kinetic stress is a sizable quantity and cannot be neglected. In this experiment, we further explore the interplay between the potential and kinetic parts of the stress. 

\begin{figure}
\centering
\subfigure[]{\label{fig:midway}\includegraphics[totalheight=0.25\textheight]{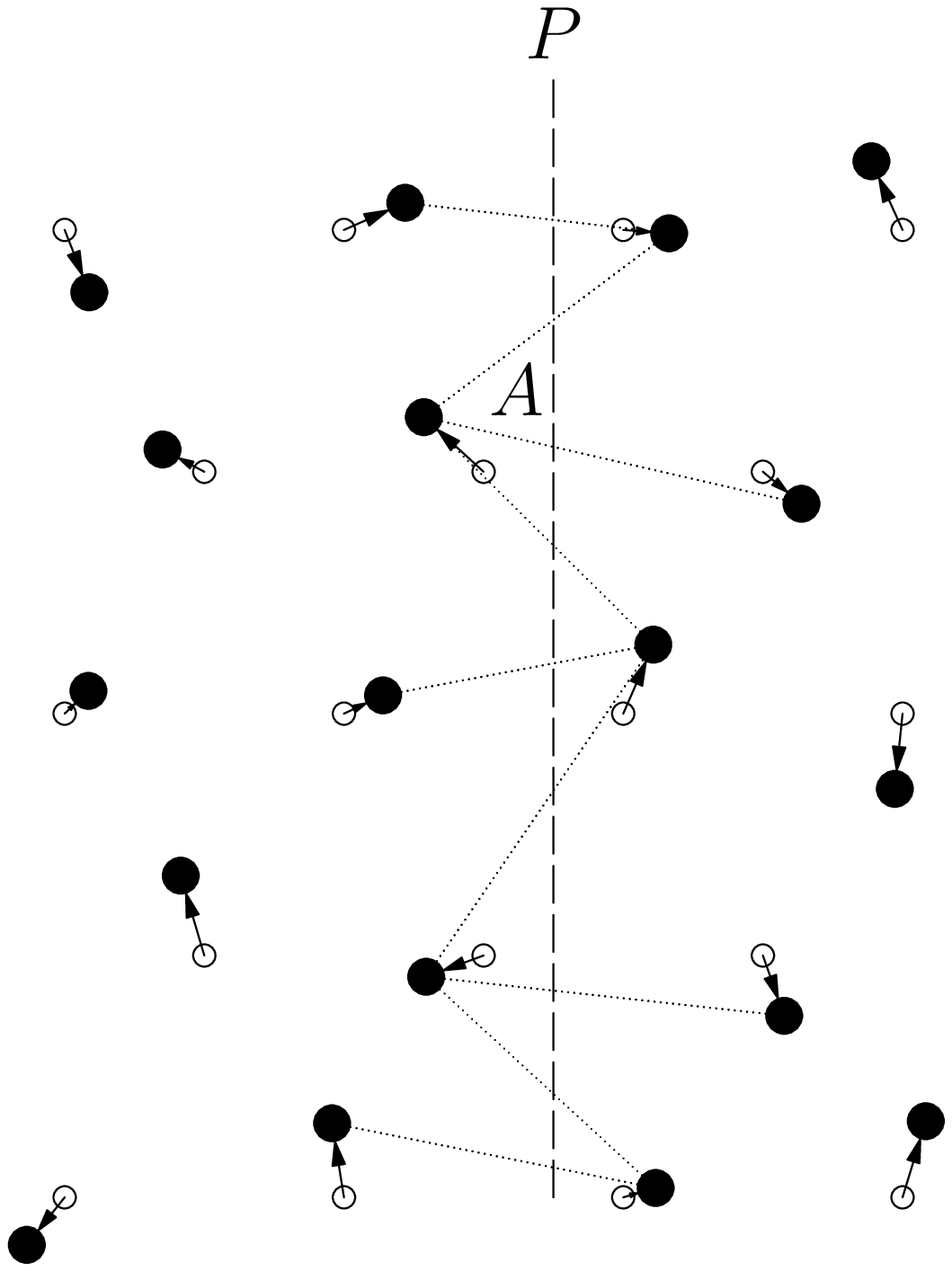}}
\subfigure[]{\label{fig:ontop}\includegraphics[totalheight=0.25\textheight]{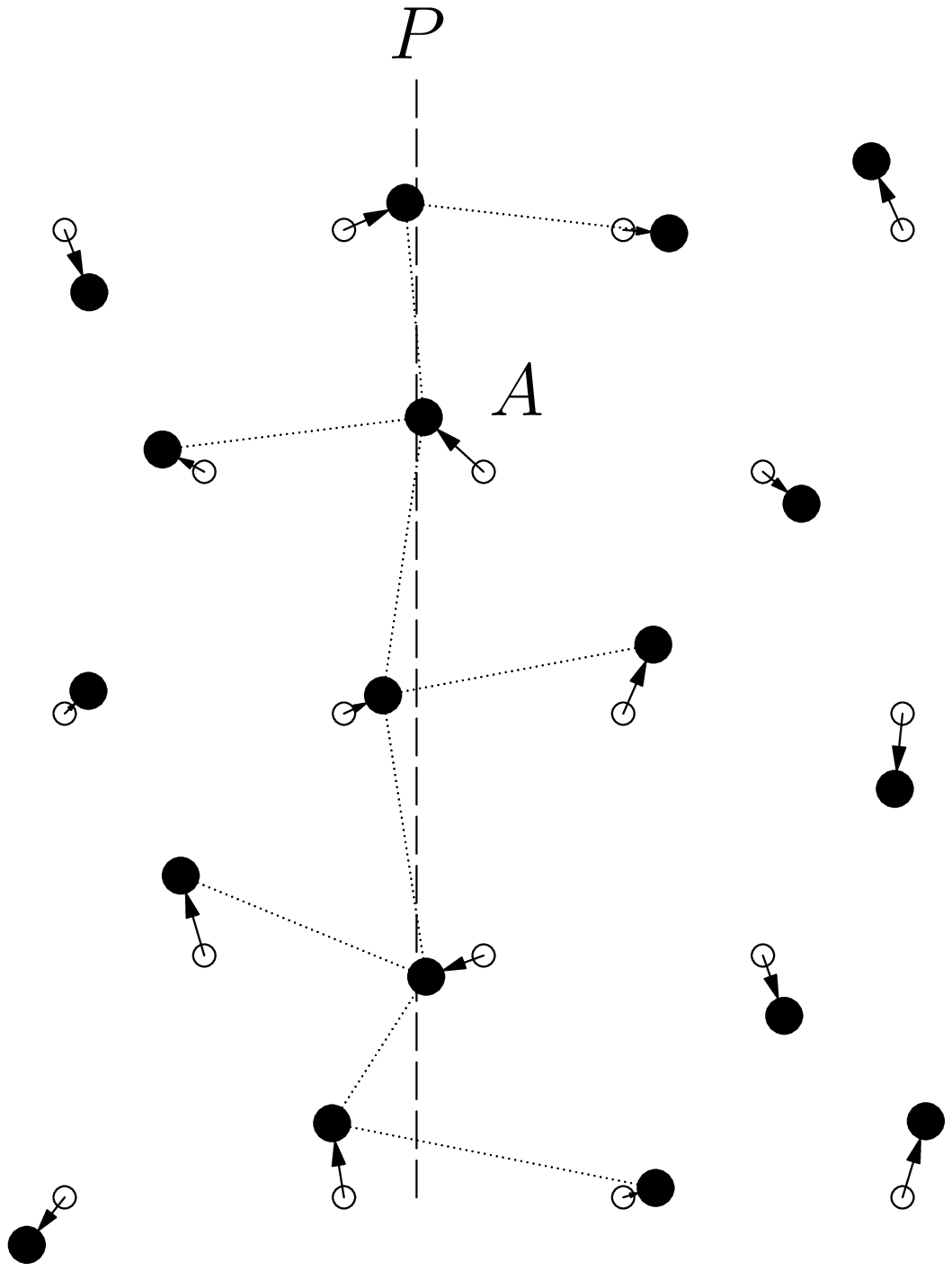}}
\subfigure[]{\label{fig:al_4000_periodic}\includegraphics[totalheight=0.25\textheight]{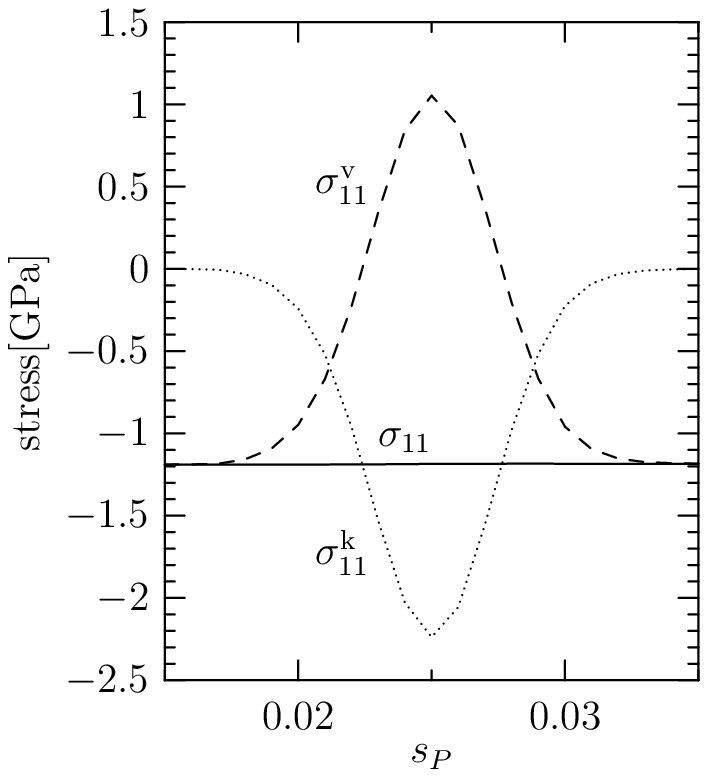}}
\caption{The effect of the position of the Tsai plane on the potential and kinetic parts of stress. Frames (a) and (b) show schematic diagrams of a two-dimensional triangular lattice with (a) the Tsai plane positioned midway between the lattice planes and (b) the Tsai plane positioned almost on top of a lattice plane. The open circles correspond to the ideal lattice positions. The black circles are atoms that are shown in mid-vibration relative to the lattice site as indicated by the arrow. The Tsai plane is indicated by a vertical dashed line. The bonds crossing it appear as dotted lines. Frame (c) shows the plot of the kinetic part of stress $\sigma^{\rm{k}}_{11}$, potential part of stress $\sigma^{\rm{v}}_{11}$ and the total stress $\sigma_{11}$, as a function of the normalized position $s_P = (x_P - x_L)/\Delta x$ of the Tsai plane P, where $x_P$ is the position of $P$, $x_L$ is the position of the lattice plane, and $\Delta x$ is the spacing between the lattice planes.}
\end{figure}

Consider a crystalline solid at a relatively low temperature under uniform stress. The atoms will vibrate about their mean positions with an amplitude that is small relative to the nearest-neighbor spacing. Now imagine placing a Tsai plane $P$ between two crystal lattice planes and measuring the traction across it. If $P$ is midway between the lattice planes (see \fref{fig:midway}), we expect that relatively few atoms will cross $P$ and that consequently the kinetic stress will be small or even zero. In contrast, if $P$ is close to the lattice plane there will be many such crossings and the kinetic stress will be large in magnitude. This seems to suggest that the traction will change as a function of the position of $P$, which would be incorrect since the system is under uniform stress. The reason that this does not occur is that every time an atom crosses $P$, the bonds connected with it reverse directions with respect to $P$, changing a positive contribution to the contact stress to a negative one and vice versa (see the bonds connected with atom $A$ in \fref{fig:midway} and \fref{fig:ontop}). This effect on the potential part of the stress exactly compensates for the change in magnitude of the kinetic stress leaving the total stress constant.  This is demonstrated numerically in \fref{fig:al_4000_periodic}. This graph shows the results obtained from a molecular dynamics simulation of the system described in \sref{sec:exp_1}, with periodic boundary conditions. The periodic length of the box is set based on the zero temperature lattice spacing. Consequently upon heating by a temperature change of $\Delta T$, a compressive stress is built up in the box according to
\begin{equation}
\label{eqn:strain}
\bm{\epsilon} = \bm{s}:\stress + \bm{I}\alpha_T \Delta T = \bm{0},
\end{equation}
where $\bm{s}$ is the elastic compliance tensor and $\alpha_T$ is the coefficient of thermal expansion. Inverting this relation for an fcc crystal with cubic symmetry oriented along the crystallographic axes, we have
\begin{equation}
\label{eqn:constitutive}
\sigma_{11} = \sigma_{22} = \sigma_{33} = -(c_{11} + 2c_{12})\Delta T = \sigma,
\end{equation}
with the rest of the stress components zero. In \eref{eqn:constitutive}, $c_{ij}$ are the elastic constants of the material. Substituting in the appropriate values for Ercolessi-Adams EAM aluminum \cite{ercolessi1994} ($c_{11}=118.1$ GPa, $c_{12} = 62.3$ GPa, $\alpha_T = 1.6 \times 10^{-5} \rm{K}^{-1}$) and $\Delta T = 310 \rm{K}$, gives $\sigma=-1.2$ GPa. We see that the total stress in \fref{fig:al_4000_periodic} is constant regardless of the position of the Tsai plane and equal to the expected value of $-1.2$ GPa computed above. However, the kinetic and potential parts change dramatically. When the Tsai plane is away from the lattice planes ($s_P = \pm 0.1$), the kinetic stress is zero and the entire stress is given by the potential part of the stress. As the Tsai plane moves closer to a lattice plane ($\abs{s_P} \to 0$), the kinetic stress becomes more negative (increasing in magnitude) and the potential part of stress increases in exactly the right amount to compensate for this effect. When the Tsai plane is right on top of a lattice plane ($s_P=0$), both the kinetic stress and potential stress are maximum in magnitude, but their sum remains equal to the constant total stress. This is a striking demonstration of the interplay between the kinetic and potential parts of the stress tensor.

\subsection{Experiment 3 and 4}
In this section, the predictions of the microscopically-based stress tensors are compared with analytical solutions from elasticity theory for two simple boundary-value problems. This is a revealing test, since stress is a continuum concept and therefore the microscopic definitions should reproduce the results of a continuum theory under the same conditions.  We perform two numerical experiments. In each experiment, an atomistic boundary-value problem is set up, and the values computed from the discrete system are compared with the ``exact'' result computed from elasticity theory for the same problem using material properties predicted by the interatomic potential used in the atomistic calculations.  The numerical experiments are conducted at zero-temperature since there is no controversy regarding the form of of the kinetic stress which is the same for all stress definitions. Therefore, a comparison at zero temperature is sufficient to probe the differences between the stress measures, at least under equilibrium conditions. The properties we are interested in studying are:
\begin{enumerate}
\item Symmetry of the stress tensor.
\item Convergence of the stress tensor to the continuum value with the size of the averaging domain (a three-dimensional volume in the case of virial, Hardy and \murdoch\ stresses and a plane in the case of the Tsai traction).
\end{enumerate}

\subsubsection*{Inter-atomic Model}
The numerical experiments in this section are carried out using a Lennard-Jones potential. The exact choice of material parameters is unimportant, since the objective of the experiment is to compare the values obtained from the microscopically-based stress for the discrete system with the ``exact'' values obtained from the continuum elasticity theory for the same material. The Lennard-Jones parameters, $\epsilon$ and $\sigma$, are therefore arbitrarily set to 1. The potential has the following form:
\begin{equation}
\label{eqn:lennard_jones}
\phi(r) = 4 \left [ \frac{1}{r^{12}} - \frac{1}{r^6} \right ] - 0.0078r^2 +0.0651.
\end{equation}
Note that the above equation has been rendered dimensionless by expressing lengths in units of $\sigma$ and energy in units of $\epsilon$. As seen in the above equation, the Lennard-Jones potential is modified by the addition of a quadratic term. This is done to ensure that $\phi(r_{\rm{cut}}) = 0$ and $\phi'(r_{\rm{cut}})=0$, where $r_{\rm{cut}} = 2.5$, denotes the cutoff radius for the potential. We refer to this as the ``modified Lennard-Jones potential''. The ground state of this potential is an fcc crystal with a lattice constant of $a=1.556$ and elastic constants, $c_{11}=87.652, c_{12}=c_{44}=50.379$. The conventional elastic moduli associated with the cubic elastic constants are \cite{lekhnitskii}:
\begin{align}
E&=(c_{11}^2 + c_{11}c_{12}-2c_{12}^2) / (c_{11} + c_{12}) = 50.877,\\
\mu &= c_{44} = 50.379, \\
\nu &= c_{12}/(c_{11}+c_{12}) = 0.365,
\end{align}
where $E$ is Young's modulus, $\mu$ is the shear modulus, and $\nu$ is Poisson's ratio. (In the above, elastic constants are given in units of $\epsilon/\sigma^3$. Poisson's ratio is dimensionless.)

\subsubsection*{Experiment 3: Dependence of the microscopically-based stress on the averaging domain size} 
The main aim of this experiment is to study the dependence of the stress given by various definitions (Hardy, Tsai, virial and \murdoch) on the size of the averaging domain. We consider the special case of uniform uniaxial loading with $\sigma_{11}=1$ (all other stress components zero). Our system is a cube of $10\times10\times10$ unit cells ($4000$ atoms) with periodic boundary conditions applied in all directions. To impose the uniaxial loading, the periodic lengths $l_i$ $(i=1,2,3)$ in the three directions are modified according to the linear elastic solution for uniform straining:
\begin{align}
l_1 &= 10a(1+\sigma_{11}/E) = 10.197,\\
l_2 = l_3 &= 10a(1-\nu \sigma_{11}/E) = 9.928.
\end{align}
We then compute the stress at the center of the periodic cell, while increasing the size of the averaging domain. In comparing the different stress definitions, the domain size is set by the Tsai plane which is taken to be a square normal to the 1-direction with the same dimension $w$ in the 2 and 3-directions. The averaging domain for the Hardy, virial and \murdoch\ stresses is a sphere of diameter $d=w$. The weighting function, $w(r)$, for the Hardy stress is taken to be constant with a suitable mollifying function,
\begin{equation}
\label{eqn:weight_molly}
w(r) = \left \{ \begin{array}{ll}
c & \mbox{if $r<R-\epsilon$}\\
\frac{1}{2}c\left[ 1- \cos \left( \frac{R-r}{\epsilon}\pi \right ) \right ] & \mbox{if $R-\epsilon<r<R$} \\
0 & \mbox{otherwise} \end{array} \right.,
\end{equation}
where $c$ is chosen appropriately to normalize $w$. The results are presented in \fref{fig:window}, where the stress $\sigma_{11}$ is plotted as a function of the normalized domain size, $s=w/10a$. (Recall that the applied value is $\sigma_{11}=1$.) We make the following qualitative observations based on these results:
 
\begin{figure}
\centering
\includegraphics[totalheight=0.27\textheight]{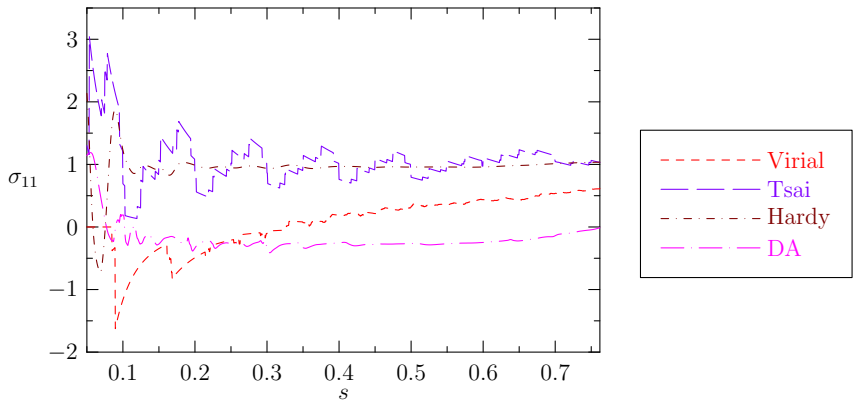}
\caption{Plot showing the dependence of $\sigma_{11}$, calculated using different definitions on the averaging domain size. The variable $s$ represents the ratio of the domain size to the length of the system ($10$ unit cells).}
\label{fig:window}
\end{figure}

\begin{enumerate}
\item The Hardy stress converges to the exact value most quickly of all stress definitions and has the least noise.
\item The normal stress computed from the Tsai traction oscillates about the exact value with a fluctuation amplitude that decays rather slowly with domain size. The oscillations reflect the symmetry of the crystal as new bonds enter the calculation with increasing plane size.
\item The virial stress is always smaller than the Hardy and Tsai stresses since it does not take into account the bonds that cross out of the averaging domain. It appears to be converging towards the exact value, but convergence is slow and even at the maximum domain size studied, the virial stress still has a significant error.
\item The \murdoch\ stress is much smaller than all other stresses due to greater averaging.
\end{enumerate}

\subsubsection*{Experiment 4: A plate with a hole under tension} We now consider one of the classical elasticity boundary-value problems: an infinite plate with a hole subjected to uniaxial tension $\sigma_{\infty}$ at infinity. This is traditionally named the \emph{Kirsch} problem for an isotropic material model. Our objective is to compare the microscopically-based stresses computed for a discrete system set up for the Kirsch problem with the exact solution. A complication in making this comparison is that the fcc Lennard-Jones material we are considering is crystalline with cubic symmetry and is not isotropic. We must therefore compare the discrete solution with the more general solution for the Kirsch problem from the theory of elasticity for anisotropic media \cite{lekhnitskii}. For anisotropic materials, the stress concentration\footnote{The stress concentration is defined as the ratio of the maximum stress to the applied stress $\sigma_\infty$. The maximum stress for the Kirsch problem occurs at the circumference of the hole.} at the hole is no longer $3$ (as it is for an isotropic material), but depends on the elastic constants of the material.  For the elastic constants of the Lennard-Jones model in \eref{eqn:lennard_jones}, we obtain a stress concentration of $2.408$. In addition to the overall stress concentration, the analytical solution provides the complete stress field about the hole. We can therefore compare the microscopically-based stress fields with the continuum result.

In order to model an infinite elastic space, we consider a large square plate oriented along the crystallographic axes consisting of $367{,}590$ atoms, with a hole of radius $25a$, where $a$ is the lattice constant. The plate is constructed by stacking $100\times100\times10$ unit cells and excluding the atoms that lie withing the radius of the hole. The relatively large system size helps to ensure that the variation of the continuum stress is small on the lengthscale of the lattice spacing and minimizes boundary effects near the hole. The atoms interact according to the modified Lennard-Jones potential given in \eref{eqn:lennard_jones}.

\begin{figure}
\centering
\subfigure{\includegraphics[totalheight=0.130\textheight]{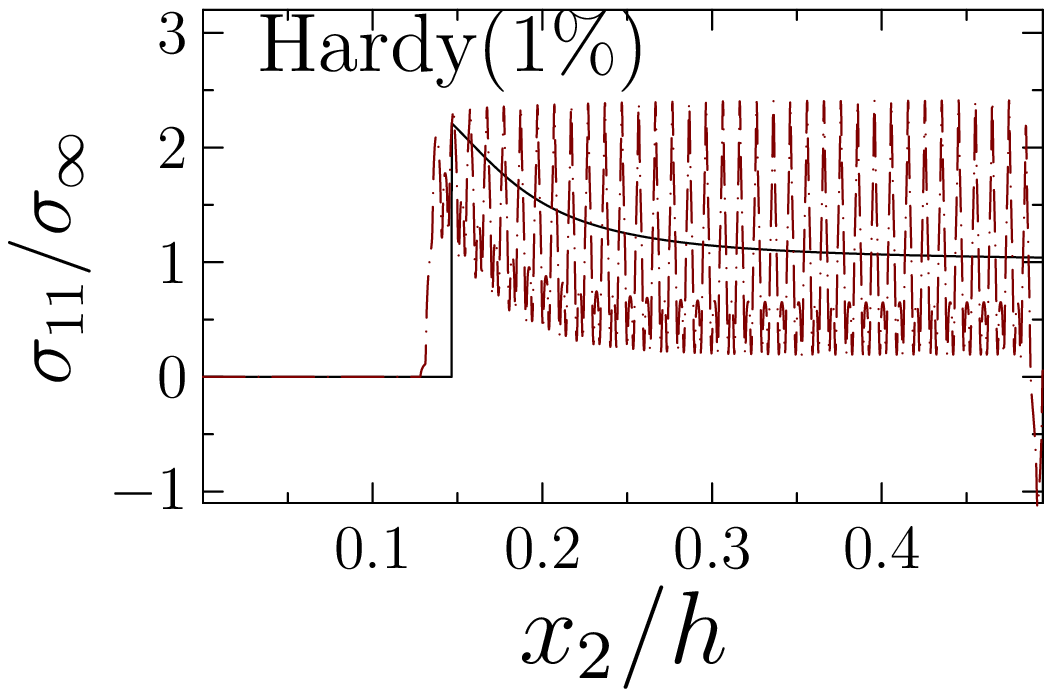}} 
\subfigure{\includegraphics[totalheight=0.130\textheight]{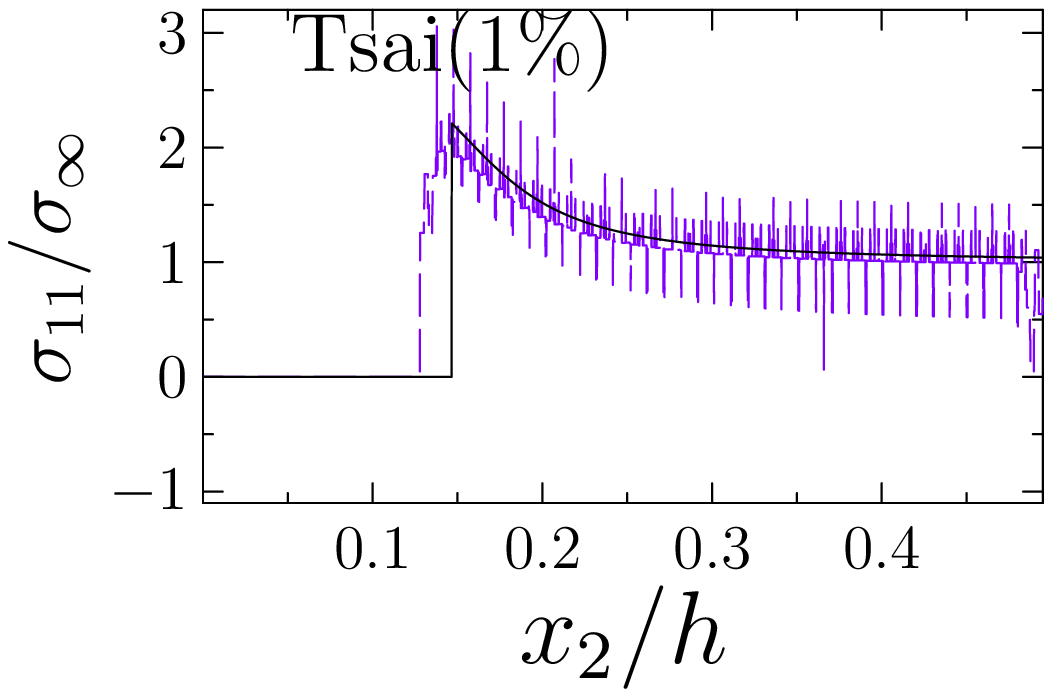}}
\subfigure{\includegraphics[totalheight=0.130\textheight]{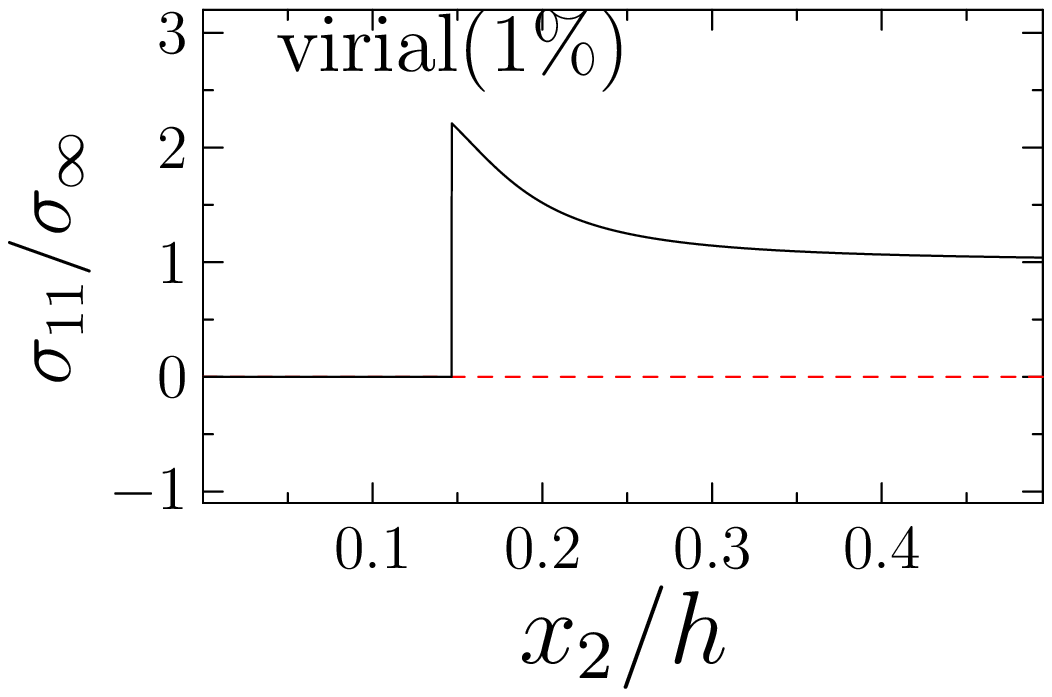}} \\

\subfigure{\includegraphics[totalheight=0.130\textheight]{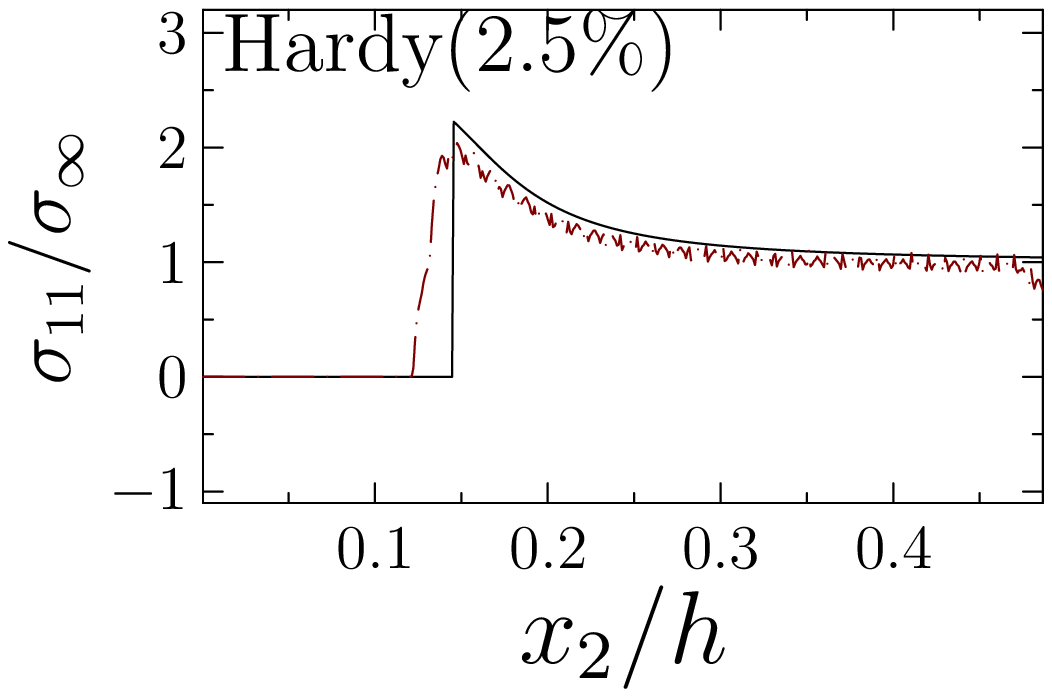}}
\subfigure{\includegraphics[totalheight=0.130\textheight]{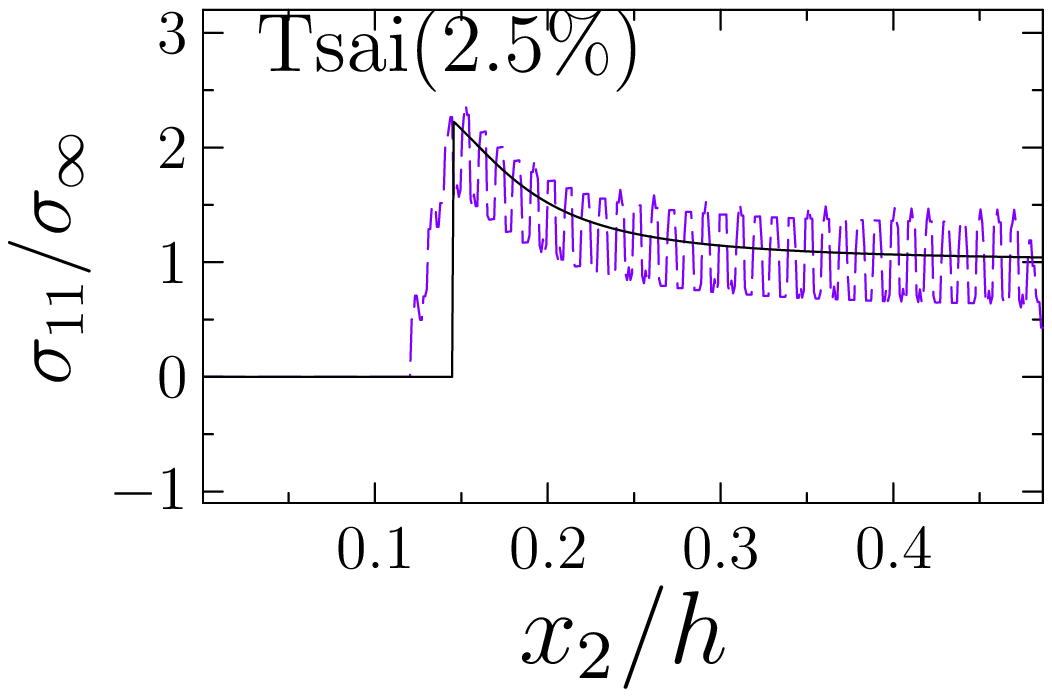}}
\subfigure{\includegraphics[totalheight=0.130\textheight]{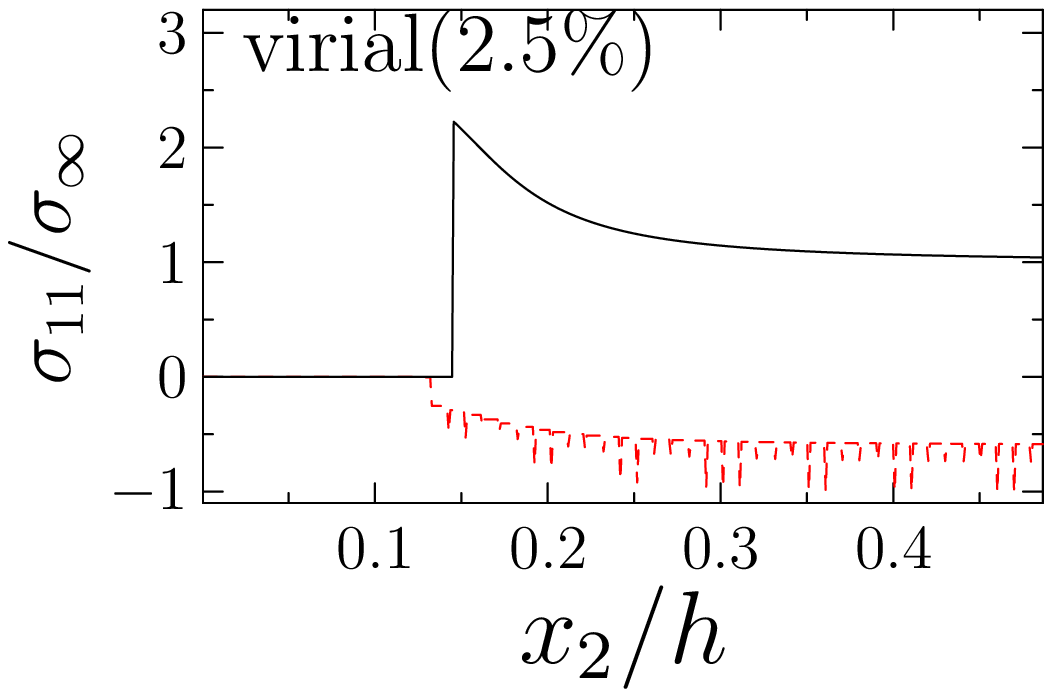}} \\

\subfigure{\includegraphics[totalheight=0.130\textheight]{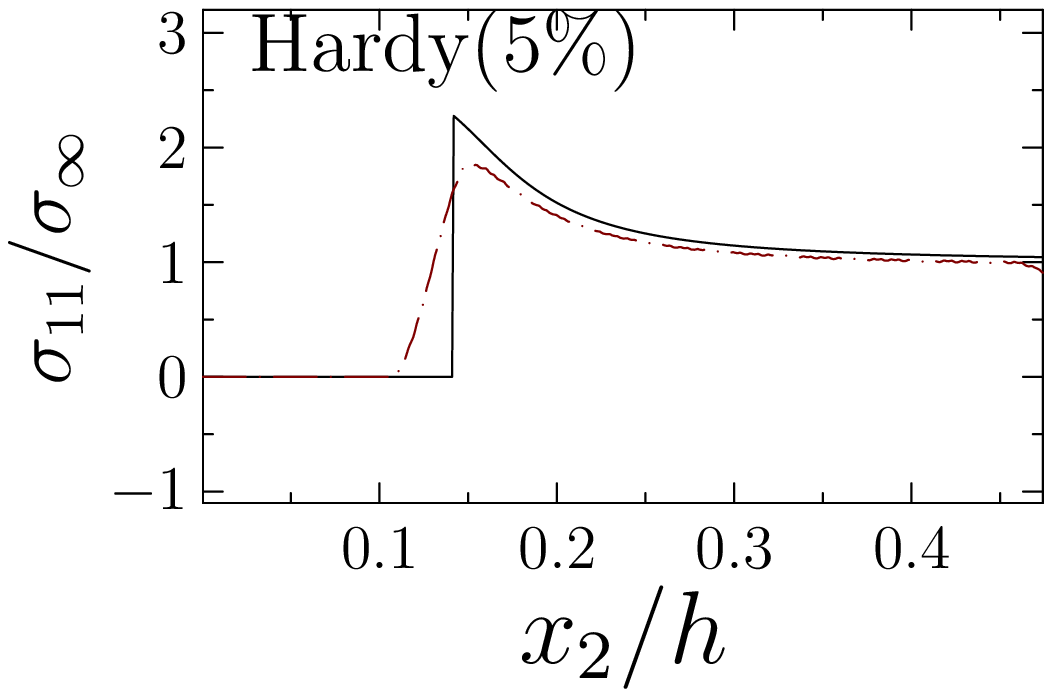}}
\subfigure{\includegraphics[totalheight=0.130\textheight]{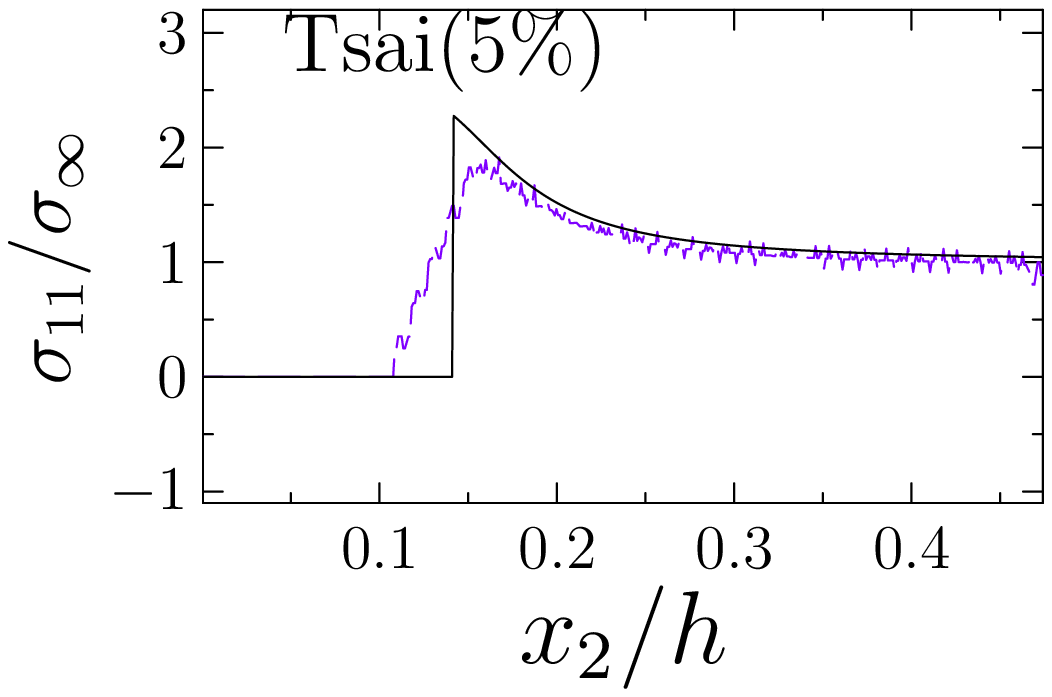}}
\subfigure{\includegraphics[totalheight=0.130\textheight]{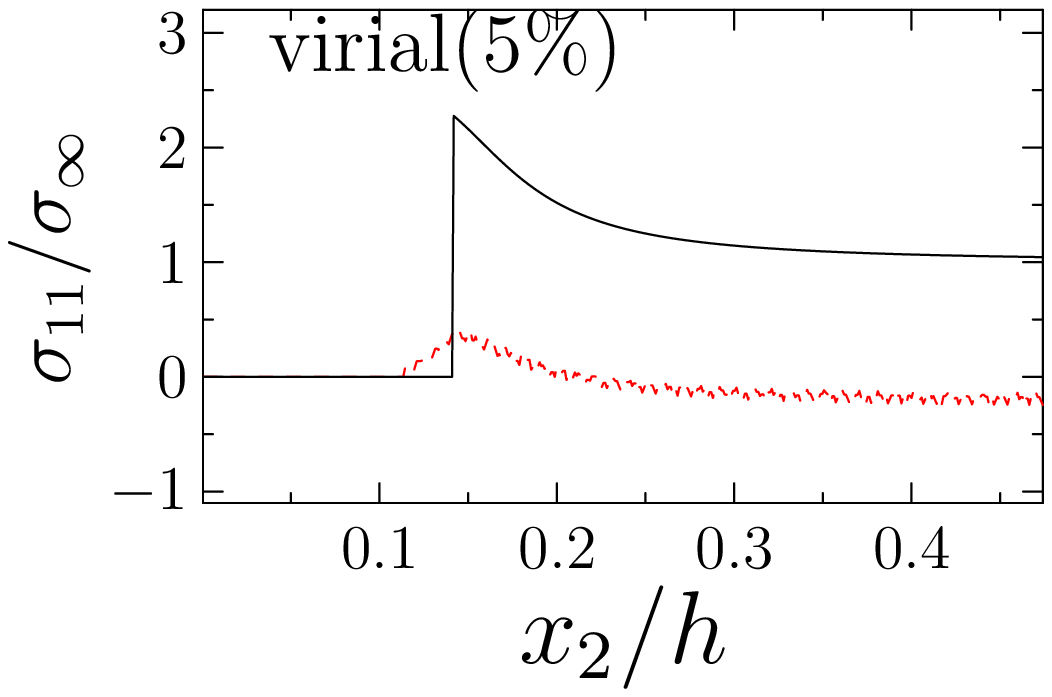}}  \\

\subfigure{\includegraphics[totalheight=0.130\textheight]{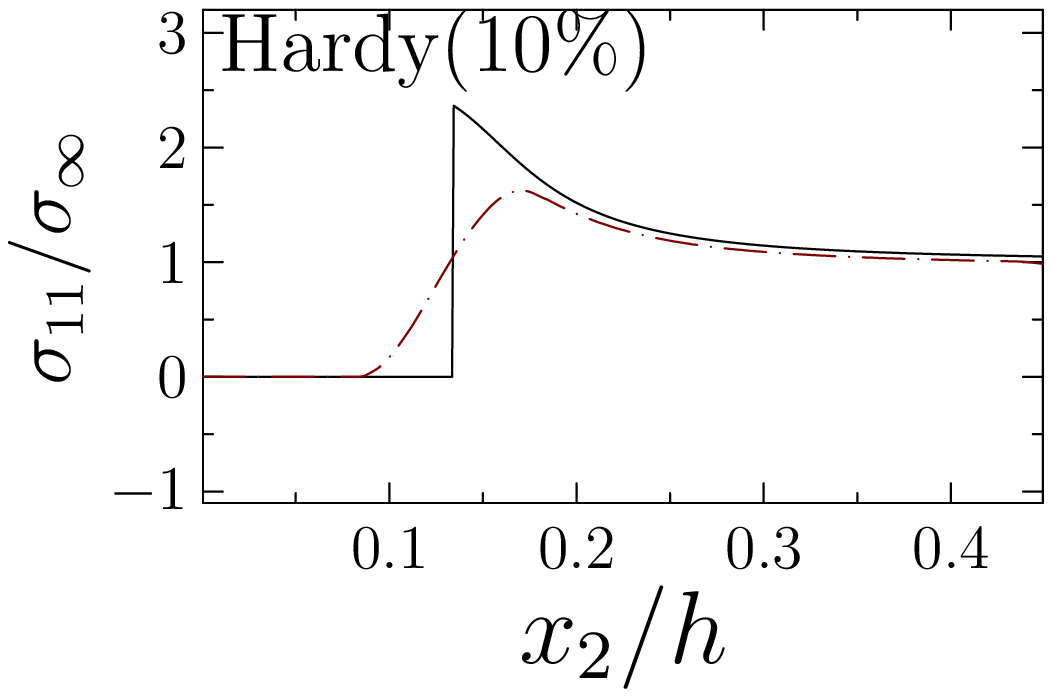}} 
\subfigure{\includegraphics[totalheight=0.130\textheight]{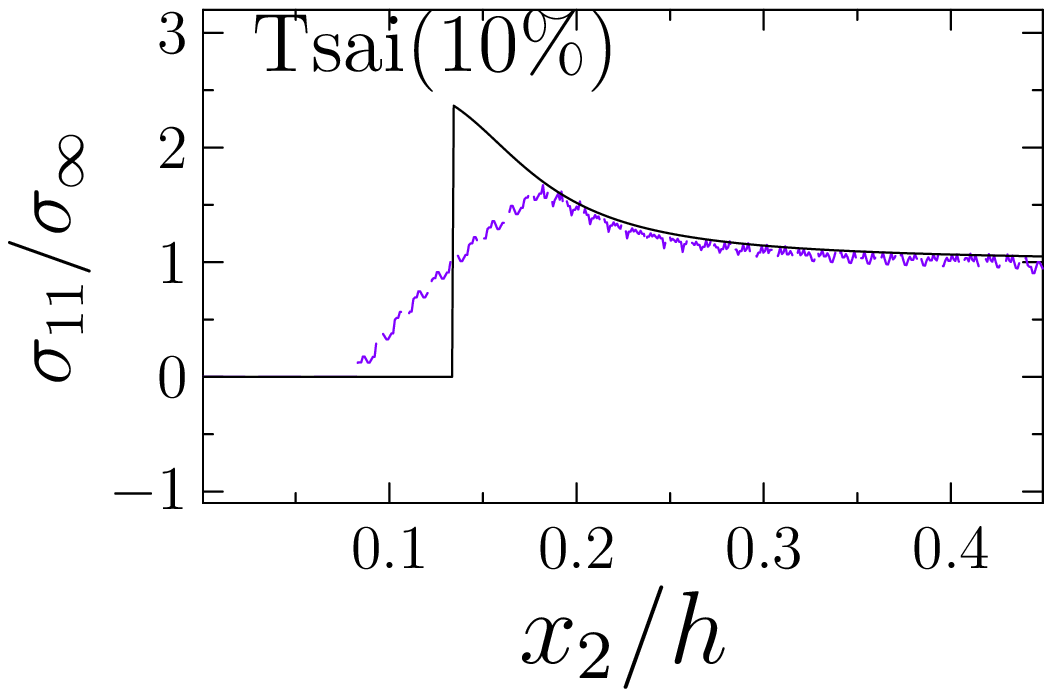}} 
\subfigure{\includegraphics[totalheight=0.130\textheight]{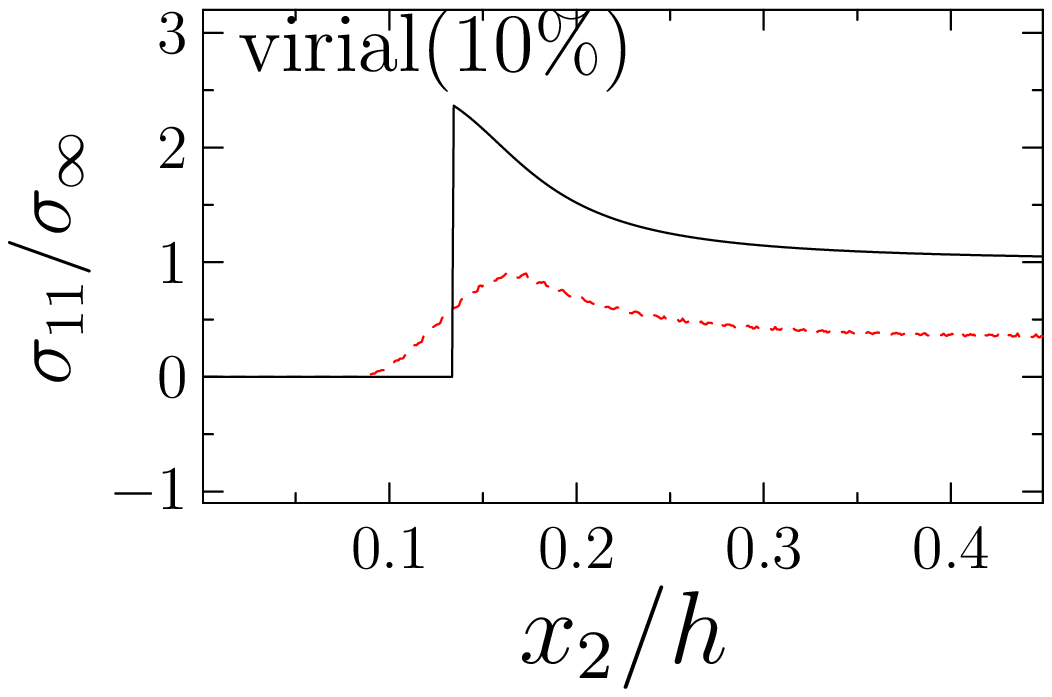}} 
\caption{Normalized $\sigma_{11}$ component of the stress along the $x_1=0$ line for an anisotropic plate with a hole subjected to uniaxial tension in the 1-direction. The $x_2$ coordinate is normalized by the height $h$ of the atomistic model. The exact solution for an infinite plate obtained from anisotropic linear elasticity (black solid line) is compared with the results obtained from the three microscopic definitions in the three columns: Hardy, Tsai and virial. The four rows correspond to four different averaging domains constituting $1\%$, $2.5\%$, $5\%$ and $10\%$ of $h$.}
\label{fig:11_1}
\end{figure}

As before, the averaging domain size is set by the length and width of the Tsai plane, with the Hardy, virial and \murdoch\ stress using a sphere of diameter equal to the length of the Tsai plane. The system is loaded by $\sigma_{11}=\sigma_\infty$, by displacing the atoms according to the exact solution from continuum mechanics for linear elastic anisotropic media \cite{lekhnitskii}. The applied stress is sufficiently small, so that the assumption of material linearity and the small strain approximation inherent in the elasticity theory provide a good approximation for the behavior of the system. After the atoms are displaced, the stress tensor is evaluated on a uniformly-distributed grid of points on the mid-plane of the plate located at $x_3=0$. A grid of $100\times100$ points is chosen to evaluate the virial, Tsai and Hardy expressions and a grid of $30\times30$ is chosen to evaluate the \murdoch\ stress tensor. A coarser grid is used for the \murdoch\ stress due to the higher computational cost of this calculation.

First, we consider the $\sigma_{11}$ component of the stress along the $x_1=0$ line, where we expect the maximum value at the hole surface. The results are plotted in \fref{fig:11_1}, which shows a comparison between the exact value and the three microscopic stress definitions (Hardy, Tsai and virial) for four different averaging domains ranging from $1\%$ to $10\%$ of the height $h$ of the atomistic model. We see that the Hardy and Tsai stresses faithfully follow the exact solution, but then drop-off as their averaging domain overlaps with the hole. This drop-off reflects the fact that the microscopically-based stress measures are bulk expressions. The smaller the averaging domain, the closer the microscopic measures can approach the exact stress concentration at the hole surface, however, this increased fidelity comes at the cost of significantly large fluctuation about the exact value. The virial stress is identically zero for the smallest averaging domain because it is too small to contain complete bonds. For the same reason, the Hardy stress experiences very large fluctuations and a nearly constant average value. For larger averaging domains, the Hardy stress has smaller fluctuations than the other stress definitions.

The reason that the drop-off effect described above is so pronounced in this simulation, is that the system is very small by continuum standards. If instead of a hole with a radius of $25a$, we studied a plate with a hole $100$ or $1000$ times larger, using the same-sized averaging domain, the spatially-averaged expressions would get much closer to the correct value before dropping off over the same lengthscale as seen in \fref{fig:11_1}. However, microscopic stress measures are often  computed for small system sizes and therefore the difficulties presented in the figure are typical of realistic atomistic simulation.

\begin{figure}
\centering
\subfigure[]{\label{fig:11_exact}\includegraphics[totalheight=0.14\textheight]{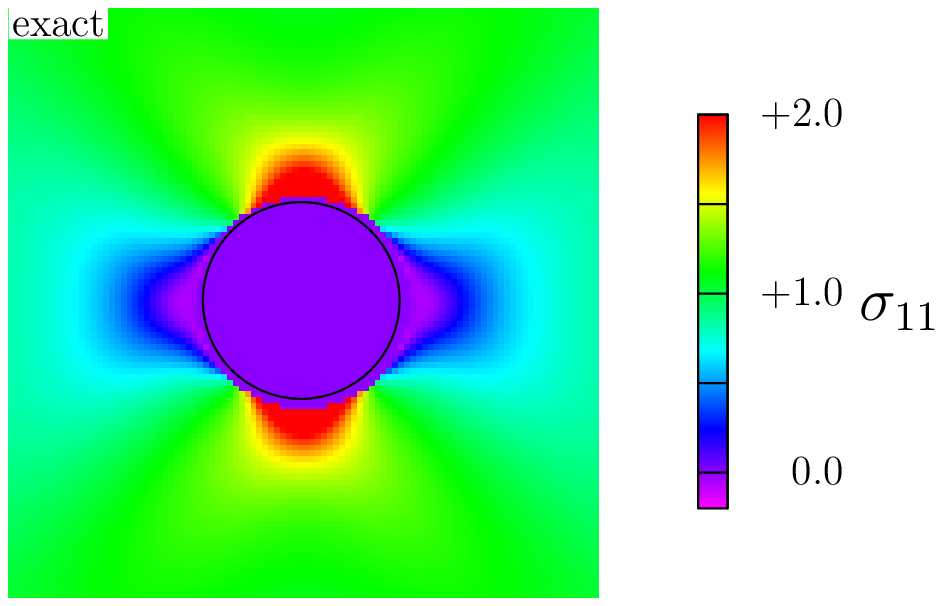}} \\
\subfigure[]{\label{fig:11_virial}\includegraphics[totalheight=0.14\textheight]{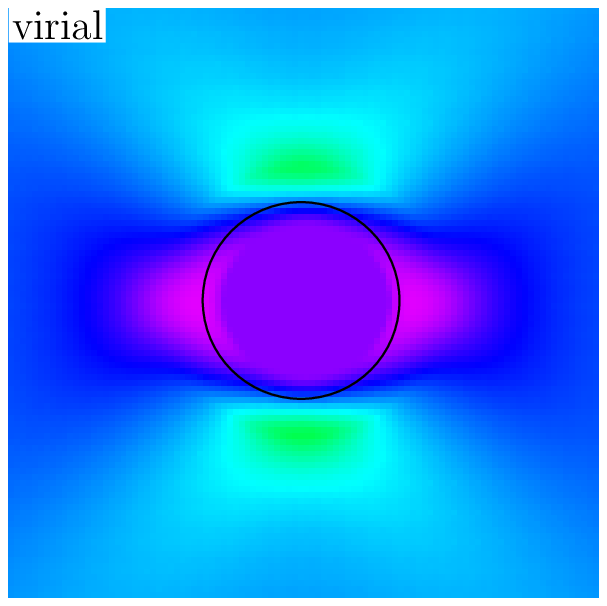}}
\subfigure[]{\label{fig:11_tsai}\includegraphics[totalheight=0.14\textheight]{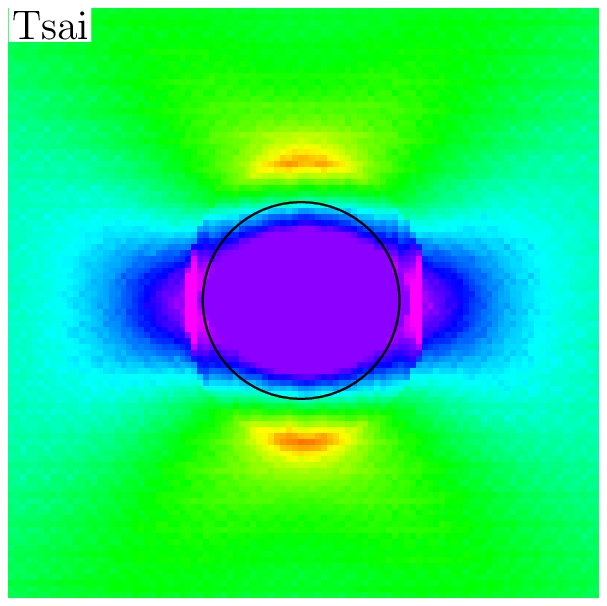}}
\subfigure[]{\label{fig:11_hardy}\includegraphics[totalheight=0.14\textheight]{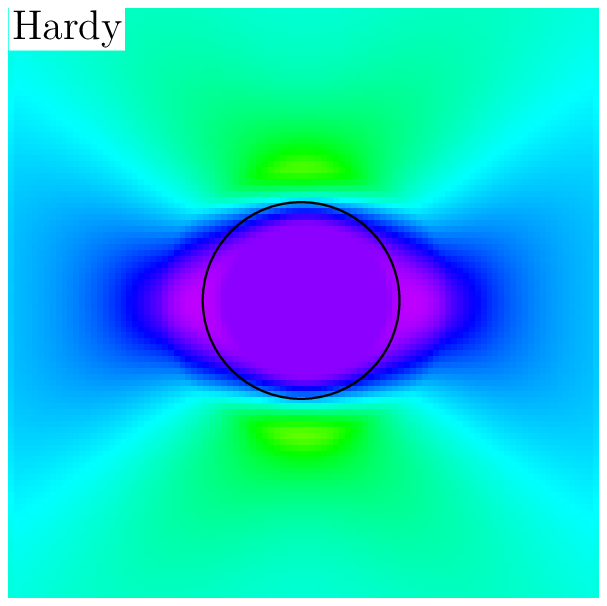}}
\subfigure[]{\label{fig:11_murdoch}\includegraphics[totalheight=0.14\textheight]{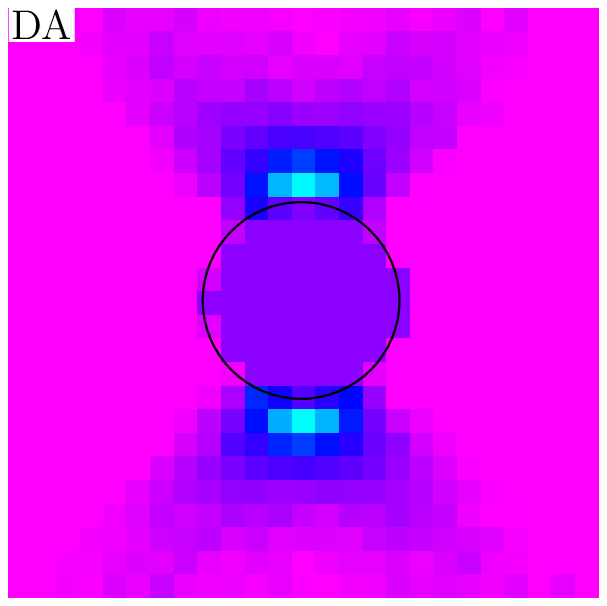}}
\caption{Color density plots of $\sigma_{11}$ are plotted on a common scale: (a) exact (b) virial (c) Tsai (d) Hardy (e) \murdoch. Results plotted for an averaging domain size of $10\%$ of the height of the model.}
\label{fig:11}
\end{figure}

Next, we explore the stress field over the entire plane. The color density plots given in \fref{fig:11}, show variation of $\sigma_{11}$ in the mid-plane of the plate. It can be seen that the stress within the hole is zero. Comparing the plots for $\sigma_{11}$ of the virial stress (\fref{fig:11_virial}), Tsai stress (\fref{fig:11_tsai}), Hardy stress (\fref{fig:11_hardy}) and the \murdoch\ stress (\fref{fig:11_murdoch}) with the exact solution (\fref{fig:11_exact}), we see that the first three definitions capture the overall variation in $\sigma_{11}$, whereas the \murdoch\ stress does not. However, it is clear that the microscopically-based stress in all of the cases is smeared relative to the stress given by the exact solution and none reach the exact stress concentration of $2.408$. This is a result of the averaging procedure involved in all the definitions as explained above. Although the \murdoch\ stress tensor plotted in \fref{fig:11_murdoch}\footnote{\fref{fig:11_murdoch} and \fref{fig:shear_murdoch} are generated from a much coarser grid compared to the other plots due to the computational expense of the \murdoch\ stress definition.} captures the variations in the field, it is much smaller in magnitude compared to the exact solution. This is because of the greater degree of averaging involved in the \murdoch\ stress tensor. Overall, the Hardy stress is less noisy than the virial or Tsai definitions due to the smoothing afforded by the weighting function. (This is hard to see in the figure.) The stress computed from the Tsai traction, in particular, is more noisy since the averaging is limited to a plane compared with the volume averaging of the Hardy and virial definitions. However, this more localized definition enables the Tsai stress to approach the exact stress concentration most closely of all of the microscopic definitions. Similar results were observed by Cheung and Yip \cite{cheung1991} for the stress near a free surface.

\begin{figure}
\centering
\subfigure[]{\label{fig:11_virial_e}\includegraphics[totalheight=0.17\textheight]{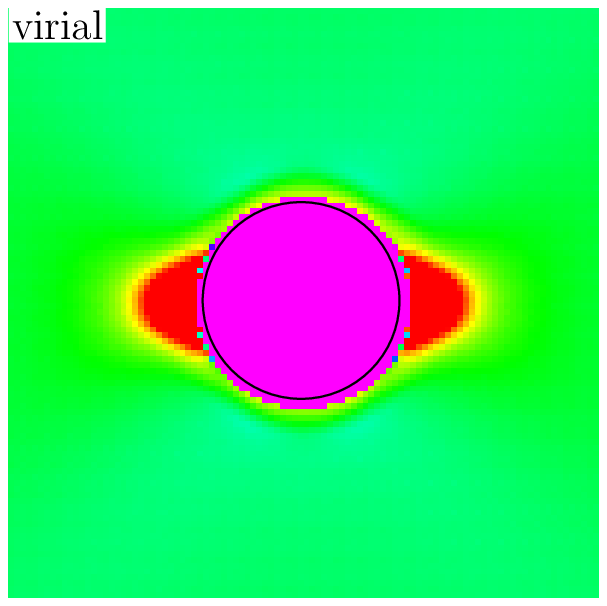}}
\subfigure[]{\label{fig:11_tsai_e}\includegraphics[totalheight=0.17\textheight]{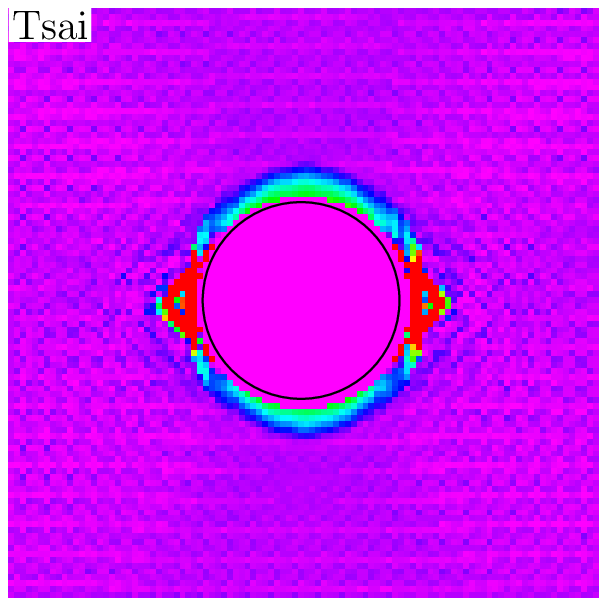}}
\subfigure[]{\label{fig:11_hardy_e}\includegraphics[totalheight=0.17\textheight]{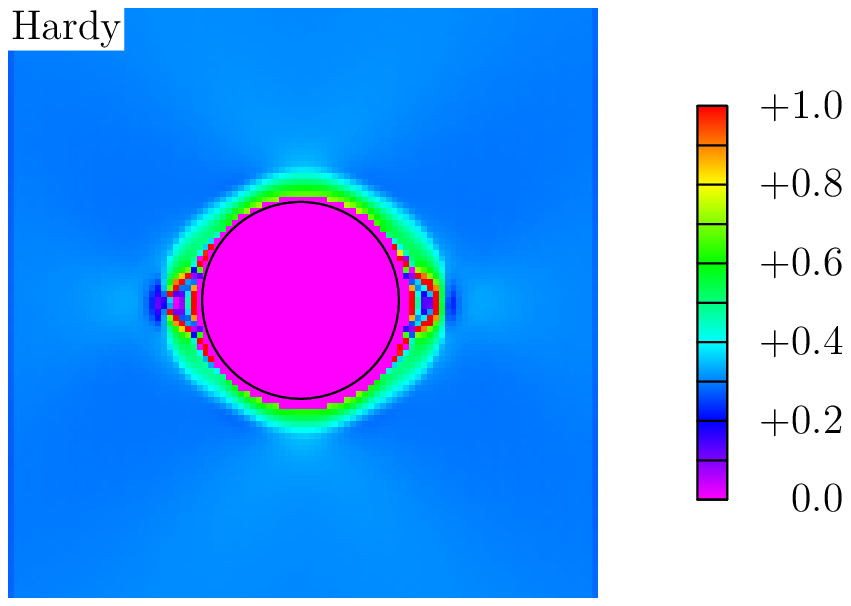}}
\caption{Color density plots of error in $\sigma_{11}$, defined as the absolute value of $(\sigma_{11} - \sigma_{11}^{\rm{exact}}) / \sigma_{11}^{\rm{exact}}$, where $\sigma_{11}$ is the stress calculated using (a) virial (b) Tsai and (c) Hardy stress definitions.}
\label{fig:11_e}
\end{figure}

The relative error in $\sigma_{11}$ for the three microscopic definitions is shown in \fref{fig:11_e}. Of the three definitions, the stress computed from the Tsai traction is generally more accurate, followed by the Hardy stress and then the virial stress. As noted above, the Tsai stress does particularly well in capturing the variations in the stress field close to the hole where the fact that it is localized in one direction is particularly helpful.

\begin{figure}
\centering
\subfigure[]{\label{fig:12_tsai}\includegraphics[totalheight=0.17\textheight]{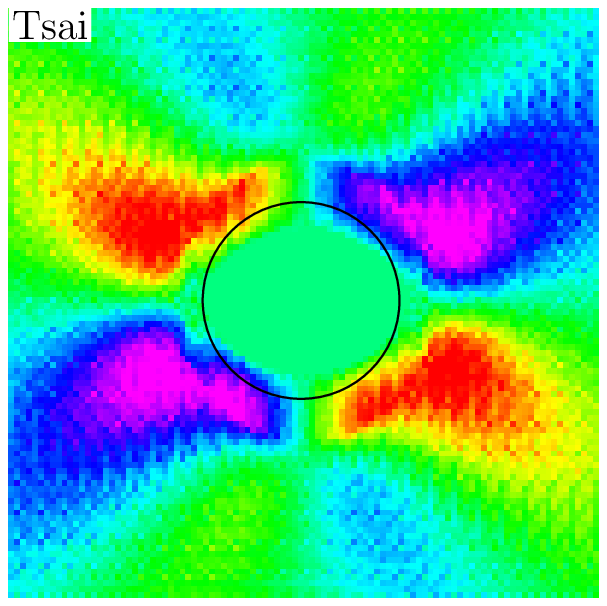}}
\subfigure[]{\label{fig:21_tsai}\includegraphics[totalheight=0.17\textheight]{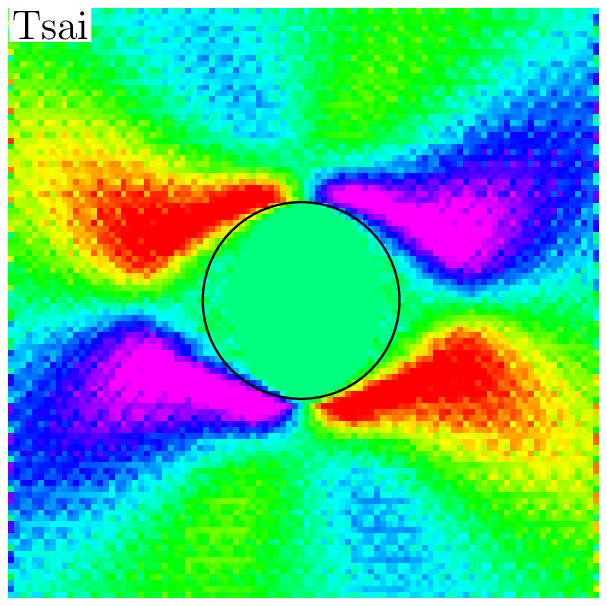}}
\subfigure[]{\label{fig:12_exact}\includegraphics[totalheight=0.17\textheight]{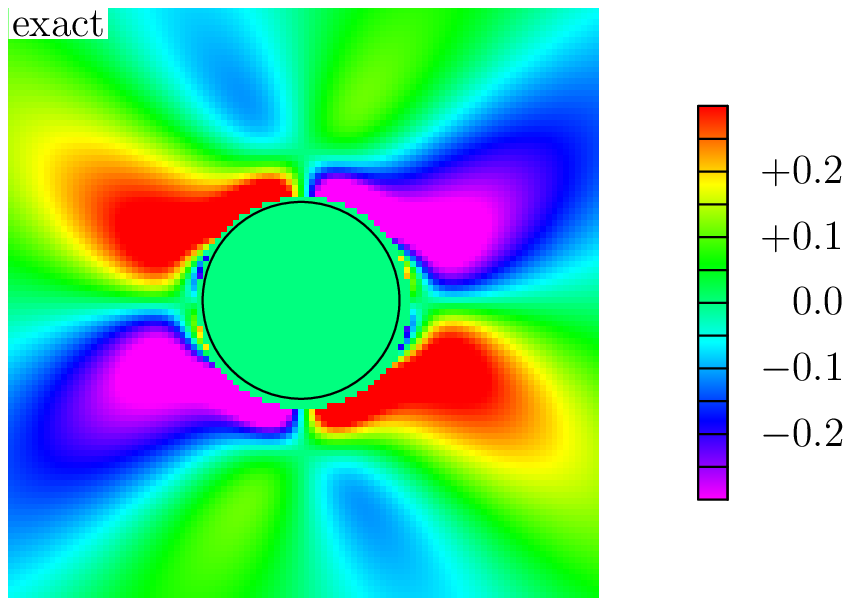}}
\caption{Color density plots of the shear stress components computed from the Tsai traction, (a) $\sigma_{12}$ and (b) $\sigma_{21}$, and (c) the exact shear stress, plotted on a common scale.}
\label{fig:12}
\end{figure}

\begin{figure}
\centering
\subfigure[]{\label{fig:12_murdoch}\includegraphics[totalheight=0.19\textheight]{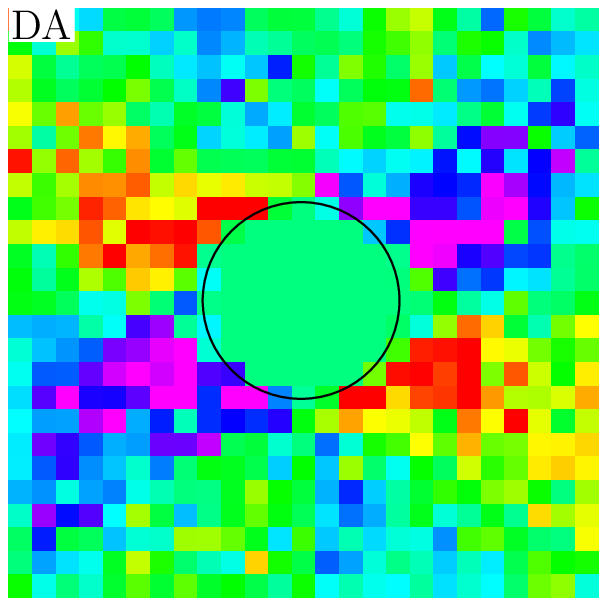}}
\subfigure[]{\label{fig:21_murdoch}\includegraphics[totalheight=0.19\textheight]{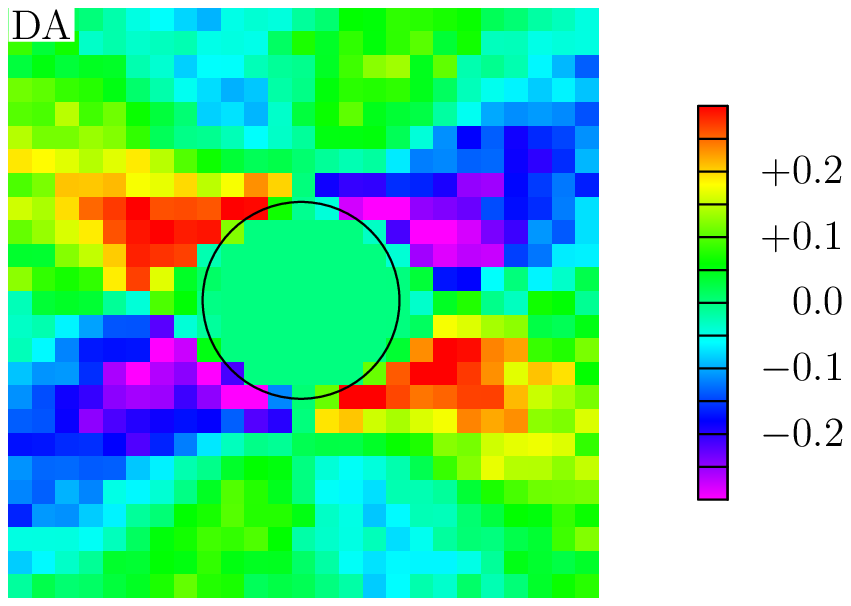}}
\caption{Color density plots of the \murdoch\ shear stress components, (a) $\sigma_{12}$ and (b) $\sigma_{21}$, plotted on a common scale.}
\label{fig:shear_murdoch}
\end{figure}

It is also interesting to examine the shear stress components. \fref{fig:12} shows the exact result from the continuum solution and the $\sigma_{12}$ and $\sigma_{21}$ components of the stress tensor computed from the Tsai traction from two different planes, one normal to the $1$ direction and the other normal to the $2$ direction. We see that the Tsai stress reproduces the exact distribution and appears generally symmetric. This suggests that the symmetry of the Hardy stress is preserved while taking the limit to arrive at the Tsai traction. The \murdoch\ stress tensor is in general non-symmetric \cite{murdoch2007}, but from \fref{fig:12_murdoch} and \fref{fig:21_murdoch} we observe that in this case, $\sigma_{12}$ and $\sigma_{21}$ appear similar. Interestingly, in contrast to the normal stress, the magnitude of shear stress is captured by the \murdoch\ stress definition, at least for the case studied here.  The reason for this is not obvious.

\medskip
Overall, we can summarize our results as follows. Of the three definitions studied, the Hardy stress is generally preferred. It tends to be the smoothest and provides good accuracy away from surfaces as long as the lengthscale over which the continuum fields vary is large relative to the atomic spacing. In situations where either of those conditions break down, the Tsai traction provides a better localized measure of stress. The virial stress is less accurate than both.  From a computational standpoint, the virial stress has the advantage of being easiest to compute. The evaluation of the bond function in the Hardy stress makes it slightly more expensive to compute, but comparable to the virial stress. The Tsai traction is most difficult and time consuming to compute, since it requires the detection of bonds and atoms that cross a given plane during the averaging process. Furthermore, this evaluation must be performed for three separate planes in order to obtain the full stress tensor in three dimensions.

\section{Summary and future work}
\label{ch:conclusions}
In this paper, we provide a unified interpretation and possible generalization of all commonly used definitions for the stress tensor for a discrete system of particles. The macroscopic stress in a system under conditions of thermodynamic equilibrium are derived using the ideas of canonical transformations within the framework of classical statistical mechanics. The stress in non-equilibrium systems is obtained in a two-step procedure: 
\begin{enumerate}

\item The Irving--Kirkwood--Noll procedure \cite{ik1950,noll1955} is applied to obtain a \emph{pointwise} (microscopic) stress tensor. The stress consists of a kinetic part $\stress_{\rm k}$ and a potential part $\stress_{\rm v}$. The potential part of the stress is obtained for multi-body potentials which have a continuously differentiable extension from the shape space of the system to a higher-dimensional Euclidean space.\footnote{Most practical interatomic potentials satisfy this conditions.} This generalizes the original Irving--Kirkwood--Noll approach that was limited to pair potentials.  This generalization is obtained based on the important result that for any multi-body potential with a continuously differentiable extension, the force on a particle in a discrete system can always be expressed as a sum of \emph{central} forces. In other words, the \emph{strong law of action and reaction} is always satisfied.

\item The pointwise stress obtained in the first step is spatially averaged to obtain the macroscopic stress.  
\end{enumerate}
This two-step procedure provides a general unified framework from which various stress definitions can be derived including \emph{all} of the main definitions commonly used in practice. In particular, it is shown that the two-step procedure leads directly to the stress tensor derived by Hardy in \cite{hardy1982}. The traction of Cauchy and Tsai \cite{Cauchy1828a,Cauchy1828b,tsai1979} is obtained from the Hardy stress in the limit that the averaging domain is collapsed to a plane. The virial stress of Clausius and Maxwell \cite{clausius1870,maxwell1870,maxwell1874} is an approximation of the Hardy stress tensor for a uniform weighting function where bonds that cross the averaging domain are neglected. The Hardy stress and virial stress become identical in the thermodynamic limit.  In this manner, clear connections are established between all of the major stress definitions in use today.

The unified framework described above yields a \emph{symmetric} stress tensor for \emph{all} interatomic potentials which have an extension, when used with the standard Irving--Kirkwood--Noll procedure. However, there are materials in nature, such as liquid crystals, which can have non-symmetric stress tensors. In order to explore the possibility of non-symmetric stress, the Irving--Kirkwood--Noll procedure is generalized to \emph{curved paths of interaction} as suggested in \cite{schofield1982}. This involves the generalization of Noll's lemmas in \cite{noll1955}, originally derived for straight bonds, to arbitrary curved paths as defined in this paper.  These generalized lemmas, lead to a non-symmetric stress tensor when applied within the Irving--Kirkwood--Noll procedure. It is postulated that curved paths of interaction may be important in systems with internal degrees of freedom, such as liquid crystals and objective structures \cite{james2006}. This is left for future work.

One of the key points addressed in this paper is the uniqueness of the stress tensor.  Three possible sources of non-uniqueness of the stress tensor are identified and addressed: 
\begin{enumerate}

\item Different pointwise stress definitions can be obtained for different potential energy extensions. This is demonstrated through a simple one-dimensional example. We also show that regardless of the uniqueness of the pointwise stress tensor, the \emph{macroscopic} stress tensor obtained through a procedure of spatial averaging is unique since the difference resulting from alternative pointwise stress tensors tends to zero as the volume of the averaging domain is increased.

\item The pointwise stress tensor is obtained by solving the balance of linear momentum, $\divr_{\bm{x}} \stress_{\rm{v}} = \bm{h}(\bm{x})$, where $\stress_{\rm{v}}$ is the potential part of the stress tensor, and $\bm{h}(\bm{x})$ is a known function. The Irving--Kirkwood--Noll procedure leads to a closed-form solution to this problem. However, an arbitrary tensor field $\tilde{\stress}$ with zero divergence can be added to $\stress_{\rm{v}}$ without violating the balance of linear momentum. We argue that in the thermodynamic limit, the non-equilibrium stress obtained through our unified two-step process must converge to the virial stress of equilibrium statistical mechanics. This is similar to the argument made by Wajnryb et al. \cite{dahler1995}. This condition is satisfied by the general stress expression that we obtain. Any divergence-free stress $\tilde{\stress}$ added to this stress must therefore also disappear under equilibrium conditions. This greatly restricts the allowable forms of $\tilde{\stress}$.

\item The generalization of the Irving--Kirkwood--Noll procedure from straight bonds to arbitrary curved paths of interaction implies the existence of multiple stress tensors for a given system. However, the existence of curved bonds implies the existence of internal structure for the discrete particles, a possibility already discussed by Kirkwood in \cite{kirkwood1946}. For a system of point masses without internal structure, only straight bonds are possible due to symmetry arguments, and therefore this source of non-uniqueness is removed. The general case of non-symmetric stress must be addressed within the context of an appropriate multipolar theory as discussed by Pitteri \cite{Pitteri1990}. We leave this to future work.

\end{enumerate}

In addition to the unified framework described above which is based on the Irving--Kirkwood--Noll procedure, we also investigated the Murdoch--Hardy procedure \cite{hardy1982,murdoch1982} of defining continuum fields as direct spatial averages of microscopic quantities. We demonstrate that this approach can be systematized by adopting a non-local continuum perspective and introducing suitable generator functions.  The various stress definitions resulting from the Murdoch--Hardy procedure, such as the Hardy, virial and the ``double-averaged'' stress (suggested by Murdoch in \cite{murdoch1994}) can be derived from this unified framework. Although we share the concern regarding the ambiguity of the probability density functions used in the Irving--Kirkwood--Noll procedure that led Murdoch to develop the direct spatial averaging approach \cite{murdoch1993}, we feel that since these probability density functions exist \emph{in principle}, the Irving--Kirkwood--Noll formalism is the correct framework to phrase the problem in with approximations introduced later to derive practical expressions.

Finally, numerical experiments involving molecular dynamics and lattice statics simulations are conducted to study the various stress definitions derived in this paper. It is generally observed that the Hardy stress definition appears to be most accurate and converges most quickly with the averaging domain size. In situations where a more localized measure of stress is needed, such as near surfaces or defects, the Tsai traction can be used instead. The virial stress is less accurate than the other two definitions and converges most slowly with averaging domain size. Its main advantage is its simple form and low computational cost.  One of the most interesting results, which requires further study, comes from Experiment~2 of a crystalline system under uniform hydrostatic stress. \fref{fig:al_4000_periodic} shows that although the potential and kinetic parts of the Tsai traction largely depend on the position of the Tsai plane between two adjacent lattice planes, the total stress remains constant. This calculation provides a striking demonstration of the interplay between the kinetic and potential parts of the stress tensor.

\vfill

\appendix

\section{Derivation of the virial stress tensor}
\label{ch:virial}
The original virial theorem was a scalar equation credited to Clausius \cite{clausius1870}, which gives a definition for the pressure in a gas. Maxwell \cite{maxwell1870} extended this result to a tensor version which gives a definition for stress. We present here a more modern version of the virial theorem partly based on the derivation by Marc and McMillan \cite{marc1985}.

Consider a system of $N$ interacting point masses. The position of each point mass is given by
\begin{equation}
\bm{x}_{\alpha} = \bm{x} + \bm{r}_{\alpha},
\label{eqn:separation}
\end{equation}
where $\bm{x}$ is the position of the center of mass of the system of particles and $\bm{r}_\alpha$ is the position of each point mass relative to the center of mass. From Newton's second law, we have
\begin{equation}
\bm{f}_\alpha = \bmd{p}_\alpha,
\label{eqn:momentum_alpha}
\end{equation}
where $\bm{f}_\alpha$ is the force acting on particle $\alpha$, and $\bm{p}_\alpha$ is its linear momentum given by
\begin{equation}
\bm{p}_\alpha = m_\alpha(\bmd{x} + \bmd{r}_\alpha) = m_\alpha(\bmd{x} + \bm{v}_\alpha^{\rm{rel}}).
\label{eq:momalf}
\end{equation}
In \eref{eq:momalf}, $m_\alpha $ is the mass of particle $\alpha$ and $\bm{v}_\alpha ^{\rm{rel}} = \bmd{r}_\alpha $ is its relative velocity with respect to the center of mass. Since the position of the center of mass $\bm{x}$ is given by
\begin{equation}
\bm{x} = \frac{\sum_\alpha  m_\alpha  \bm{x}_\alpha  }{\sum_\alpha  m_\alpha },
\end{equation}
we have the following identities which will be used later:
\begin{equation}
\label{eqn:identities}
\sum_\alpha  m_\alpha  \bm{r}_\alpha  = \bm{0}, \quad \sum_\alpha  m_\alpha  \bm{v}_\alpha ^{\rm{rel}} = \bm{0}.
\end{equation}
Next, we apply a tensor product with $\bm{r}_\alpha $ to both sides of \eref{eqn:momentum_alpha} to give
\begin{equation}
\bm{r}_\alpha  \otimes \bm{f}_\alpha  = \bm{r}_\alpha  \otimes \bmd{p}_\alpha.
\label{eqn:tensor_product}
\end{equation}
On using the identity
\begin{equation}
\frac{d}{dt}(\bm{r}_\alpha  \otimes \bm{p}_\alpha ) = \bm{v}_\alpha ^{\rm{rel}} \otimes \bm{p}_\alpha  + \bm{r}_\alpha  \otimes \bmd{p}_\alpha ,
\end{equation}
equation \eref{eqn:tensor_product} becomes
\begin{equation}
\label{eqn:dyn_virial}
\frac{d}{dt}(\bm{r}_\alpha  \otimes \bm{p}_\alpha ) = \bm{\mcal{W}}_\alpha  + 2\bm{\mcal{T}}_\alpha ,
\end{equation}
where $\bm{\mcal{W}}_\alpha  = \bm{r}_\alpha  \otimes \bm{f}_\alpha $ is the \emph{virial tensor} and $\bm{\mcal{T}}_\alpha  = \frac{1}{2}(\bm{v}_\alpha ^{\rm{rel}} \otimes \bm{p}_{\alpha})$ is the \emph{kinetic tensor} of particle $\alpha$. Equation \eref{eqn:dyn_virial} is called the \emph{dynamical tensor virial theorem}. This ``theorem'', which is simply an alternative form for the balance of linear momentum, becomes useful in a modified form after making the assumption that the atoms in the system are in \emph{stationary motion}. This means that there exists a time scale $\tau$, which is short relative to macroscopic processes but long relative to the characteristic time of the atomic system, over which the atoms remain close to their original positions with bounded positions and velocities. This condition is satisfied for a system of atoms undergoing thermal vibrations about their mean positions as expected in a solid at moderate temperature. To exploit this property of the system, we define the time average of any quantity $f$ over the time $\tau$ as
\begin{equation}
\bar{f} = \frac{1}{\tau}\int_0^\tau f(t) dt
\end{equation}
and apply this averaging to \eref{eqn:dyn_virial}:
\begin{equation}
\frac{1}{\tau} (\bm{r}_\alpha  \otimes \bm{p}_\alpha ) \bigg |_0^\tau = \ol{\bm{\mcal{W}}}_\alpha  + 2\ol{\bm{\mcal{T}}}_\alpha .
\end{equation}
Assuming that $\bm{r}_\alpha  \otimes \bm{p}_\alpha $ is bounded, and a separation of time scales between microscopic and continuum processes exists, the term on the left-hand side can be made as small as desired by taking $\tau$ sufficiently large. Therefore the above equation is reduced to $\ol{\bm{\mcal{W}}}_\alpha  = -2\ol{\bm{\mcal{T}}}_\alpha $. Summing over all particles gives the \emph{tensor virial theorem}:
\begin{equation}
\ol{\bm{\mcal{W}}} = -2\ol{\bm{\mcal{T}}},
\label{eqn:virial_thm}
\end{equation}
where $\ol{\bm{\mcal{W}}}$ is the time-averaged virial tensor of the system,
\begin{equation}
\label{eqn:virial_tensor}
\ol{\bm{\mcal{W}}} = \sum_\alpha  \ol{\bm{r}_\alpha  \otimes \bm{f}_\alpha },
\end{equation}
and $\ol{\bm{\mcal{T}}}$ is the time-averaged kinetic tensor of the system,
\begin{equation}
\label{eqn:kinetic_tensor}
\ol{\bm{\mcal{T}}} = \frac{1}{2}\sum_\alpha  m_\alpha  \ol{\bm{v}_\alpha ^{\rm{rel}} \otimes \bmd{x}_\alpha }.
\end{equation}
The expression for $\ol{\bm{\mcal{T}}}$ can be simplified by substituting in \eref{eqn:separation} and noting that due to separation of time scales $\bmd{x}$ is constant on the atomistic time scale $\tau$, so that
\begin{equation}
\ol{\bm{\mcal{T}}} = \frac{1}{2}\sum_\alpha  m_\alpha  \ol{\bm{v}_\alpha ^{\rm{rel}} \otimes \bm{v}_\alpha ^{\rm{rel}}} + \left [ \ol{\sum_\alpha  m_\alpha  \bm{v}_\alpha ^{\rm{rel}}} \right ] \otimes \bmd{x}.
\end{equation}
The term in the square brackets is zero due to \eref{eqn:identities}, and thus 
\begin{equation}
\ol{\bm{\mcal{T}}} = \frac{1}{2}\sum_\alpha  m_\alpha \ol{\bm{v}_\alpha ^{\rm{rel}} \otimes \bm{v}_\alpha ^{\rm{rel}}}.
\label{eqn:kinetic_tensor_2}
\end{equation}
It is important to note that the virial theorem in \eref{eqn:virial_thm} applies equally to continuum systems at rest as well as those that are not in macroscopic equilibrium and are therefore in a state of motion. This statement hinges on the separation of scales assumption according to which continuum motion occurs so slowly relative to atomistic processes so as to be essentially constant on that time scale.\\

The virial theorem leads to a definition for stress by considering the idea that the forces on the particles in the system can be divided into two parts: an internal part, $\bm{f}_\alpha ^{\rm{int}}$, resulting from the interaction of particle $\alpha$ with other particles in the system and an external part, $\bm{f}_\alpha ^{\rm{ext}}$, due to its interaction with atoms outside the system and due to external fields,
\begin{equation}
\bm{f}_\alpha  = \bm{f}_\alpha ^{\rm{int}} + \bm{f}_\alpha ^{\rm{ext}}.
\label{eqn:split}
\end{equation}
In terms of continuum variables, the external part of the force can be identified with the traction, $\bm{t}$, that the surrounding medium applies to the system of particles and the body force, $\rho\bm{b}$, acting on it, where $\rho$ is the mass density and $\bm{b}$ is body force per unit mass.\footnote{These density, traction and body forces are {\em pointwise} continuum fields, i.e., they are defined at all points but do not include the spatial averaging implicit in the macroscopic continuum description. See \sref{sec:phase} for more details on pointwise fields.} Substituting \eref{eqn:split} into \eref{eqn:virial_tensor}, divides the virial tensor for the system into internal and external parts,
\begin{equation}
\label{eqn:w_split}
\ol{\bm{\mcal{W}}} = \ol{\bm{\mcal{W}}}_{\rm{int}} + \ol{\bm{\mcal{W}}}_{\rm{ext}} = \sum_\alpha  \ol{\bm{r}_\alpha  \otimes \bm{f}_\alpha ^{\rm{int}}} + \sum_\alpha  \ol{\bm{r}_\alpha  \otimes \bm{f}_\alpha ^{\rm{ext}}}.
\end{equation}
Rewriting the external virial as\footnote{We accept this step as an ansatz due to the many assumptions involved, one of which being that $\Omega$ is large enough to express the external forces acting on $\Omega$ in the form of the continuum traction $\bm{t}$.}
\begin{equation}
\ol{\bm{\mcal{W}}}_{\rm{ext}} 
= \sum_\alpha  \ol{\bm{r}_\alpha  \otimes \bm{f}_\alpha ^{\rm{ext}}} 
:= \int_\Omega \bm{x} \otimes \rho\bm{b}\,dV
 + \int_{\partial \Omega} \bm{x} \otimes \bm{t} \, dA,
\end{equation}
where $\Omega$ is the domain occupied by the system of particles and $\partial \Omega$ is a continuous closed surface bounding the particles and separating them from surrounding particles. The variable $\bm{x}$ is a position vector within $\Omega$ and on the surface $\partial \Omega$.
Substituting the Cauchy relation, $\bm{t} = \bm{\stress} \bm{n}$, where $\stress$ is the pointwise Cauchy stress, into the above equation and applying the divergence theorem, we have
\begin{equation}
\ol{\bm{\mcal{W}}}_{\rm{ext}} 
= \int_{\Omega} \left[
  \bm{x}\otimes\rho\bm{b} 
+ \divr_{\bm{x}} (\bm{x} \otimes \stress) 
\right] \, dV 
= \int_{\Omega} \left [ 
  \stress^{\T} + \bm{x} \otimes (\divr_{\bm{x}} \stress + \rho\bm{b}) 
\right ] \, dV.
\end{equation}
We assume that the pointwise fields satisfy the same balance laws as the macroscopic continuum fields. Therefore, the term $(\divr_{\bm{x}} \stress + \rho\bm{b})$ is zero under equilibrium conditions (see \eref{eqn:motion}). We therefore have that
\begin{equation}
\label{eqn:w_ext}
\ol{\bm{\mcal{W}}}_{\rm{ext}} = V\stress^{\T},
\end{equation}
where we have defined the continuum stress field as the average value of $\stress$ over the domain $\Omega$:
\begin{equation}
\stressave := \frac{1}{V}\int_\Omega \stress\, dV,
\end{equation}
Here $V$ is the volume of $\Omega$.\footnote{The definition of this volume is somewhat arbitrary. One can possibly define $V$ as the total volume of the Voronoi cells of the atoms lying within $\Omega$.} Substituting \eref{eqn:w_ext} into \eref{eqn:w_split} and then into the virial theorem in \eref{eqn:virial_thm} gives
\begin{equation}
\stressave = -\frac{1}{V}\left [ \ol{\bm{\mcal{W}}}_{\rm{int}} + 2 \ol{\bm{\mcal{T}}} \right ] ^{\T}.
\end{equation}
Substituting in \eref{eqn:kinetic_tensor_2} and the definition of $\ol{\bm{\mcal{W}}}_{\rm{int}}$ from \eref{eqn:w_split}, we have
\begin{equation}
\label{eqn:virial_1}
\stressave = - \frac{1}{V}\left [ \sum_\alpha  \ol{\bm{f}_\alpha ^{\rm{int}} \otimes \bm{r}_\alpha } + \sum_\alpha  m_\alpha  \ol{\bm{v}_\alpha ^{\rm{rel}} \otimes \bm{v}_\alpha ^{\rm{rel}}} \right ].
\end{equation}
This is the virial stress tensor. It is an expression for the Cauchy stress tensor given entirely in terms of atomistic quantities. Finally to demonstrate the symmetry of the virial stress, we rewrite $\bm{f}^{\rm{int}}_\alpha $ as the sum over its decomposition:
\begin{equation}
\label{eqn:f_int}
\bm{f}^{\rm{int}}_\alpha  
= \sum_{\substack{\beta \\ \beta \ne \alpha}} \bm{f}_{\alpha\beta},
\end{equation}
where $\bm{f}_{\alpha\beta}$ are the terms in the central-force decomposition corresponding to a potential energy extension as explained in \sref{sec:s_motion}. Substituting \eref{eqn:f_int} into \eref{eqn:virial_1}, we obtain
\begin{equation}
\stressave = -\frac{1}{V}\Bigg [ \sum_{\substack{\alpha,\beta \\ \alpha \ne \beta}} \ol{\bm{f}_{\alpha\beta} \otimes \bm{r}_\alpha } + \sum_\alpha  m_\alpha  \ol{\bm{v}_\alpha ^{\rm{rel}} \otimes \bm{v}_\alpha ^{\rm{rel}}} \Bigg ].
\label{eqn:virial_int}
\end{equation}
Now recalling that $\bm{f}_{\alpha\beta}=-\bm{f}_{\beta\alpha}$, we have the following identity
\begin{equation}
\sum_{\substack{\alpha,\beta \\ \alpha \ne \beta}} 
\bm{f}_{\alpha\beta} \otimes \bm{r}_\alpha
=
\frac{1}{2}
\sum_{\substack{\alpha,\beta \\ \alpha \ne \beta}} 
\left(
\bm{f}_{\alpha\beta} \otimes \bm{r}_\alpha +
\bm{f}_{\beta\alpha} \otimes \bm{r}_\beta
\right)
=
\frac{1}{2}
\sum_{\substack{\alpha,\beta \\ \alpha \ne \beta}} 
\bm{f}_{\alpha\beta} \otimes (\bm{r}_\alpha-\bm{r}_\beta).
\label{eqn:fr_identity}
\end{equation}
Substituting \eref{eqn:fr_identity} into \eref{eqn:virial_int}, we have
\begin{equation}
\stressave = -\frac{1}{V}\Bigg [ \frac{1}{2} \sum_{\substack{\alpha,\beta \\ \alpha \ne \beta}} \ol{\bm{f}_{\alpha\beta} \otimes (\bm{r}_\alpha  - \bm{r}_\beta)} + \sum_\alpha  m_\alpha  \ol{\bm{v}_\alpha ^{\rm{rel}} \otimes \bm{v}_\alpha ^{\rm{rel}}} \Bigg ]. 
\label{eqn:virial_2}
\end{equation}
This expression shows that the virial stress is symmetric, since $\bm{f}_{\alpha\beta}$ is parallel to $\bm{r}_\alpha  - \bm{r}_\beta$.

\section{Distance geometry}
\label{sec:geometry}
In \sref{sec:s_motion}, we saw that the interatomic potential energy, $\pot_{\rm int}$, of a set of $N$ particles can be defined as a function on a $3N-6$ dimensional manifold embedded in $\mathbb{R}^{N(N-1)/2}$, where each point on this manifold represents an $N(N-1)/2$-tuple of real numbers which correspond to the distances between the $N$ particles in $\real{3}$. We also noted that these $N(N-1)/2$ distances must satisfy certain geometric constraints in order to be physically meaningful (see footnote \ref{fn:phys_dist}). In this appendix, we discuss the nature of these geometric constraints.

An entire field of geometry, referred to as {\em distance geometry}, has emerged to describe the geometry of sets of points in terms of the distances between them. The subject was first treated systematically by Karl Menger in the early twentieth century and summarized in the book by Blumenthal \cite{Blumenthal1970}.  More recent references include \cite{Crippen1977,HavelKuntz1983,SipplScheraga1986,Braun1987,CrippenHavel1988}. Our discussion below is partly based on \cite{PortaRos2005}, which provided a particularly clear explanation of the subject. 

The key function in distance geometry is the {\em Cayley-Menger determinant}, $\chi:\mathbb{R}^{N(N-1)/2} \to \mathbb{R}$, which is defined as
\begin{equation}
\chi(\zeta_{12},\dots,\zeta_{1N},\zeta_{23},\dots,\zeta_{(N-1)N})
= \det
\begin{bmatrix}
0             & s_{12} & s_{13} & \cdots & s_{1N} &      1 \\
s_{12} &             0 & s_{23} & \cdots & s_{2N} &      1 \\
s_{13} & s_{23} &             0 & \cdots & s_{3N} &      1 \\
       \vdots &        \vdots &        \vdots &        &        \vdots & \vdots \\
s_{1N} & s_{2N} & s_{3N} & \cdots &             0 &      1 \\
            1 &             1 &             1 & \cdots &             1 &      0 
\end{bmatrix},
\end{equation}
where 
\begin{equation}
s_{\alpha \beta} = \zeta_{\alpha \beta}^2.
\end{equation}

For a system of $N$ points, $\{\bm{x}_1,\dots,\bm{x}_N\}$, embedded in $\real{3}$, the Cayley-Menger determinant evaluated at the point, $(r_{12},\dots,r_{(N-1)N}))$, where $r_{\alpha\beta}=\vnorm{\bm{x}_\alpha-\bm{x}_\beta}$, is related to the volume $V_{N-1}$ of a simplex of $N$ points in an ($N-1$)-dimensional space through the relation:
\begin{equation}
\chi(r_{12},\dots,r_{(N-1)N}) = 
\frac {2^{N-1}((N-1)!)^2} {(-1)^N} V_{N-1}^2(\bm{x}_1,\dots,\bm{x}_N).
\end{equation}
For $N=2$,
\begin{equation}
\chi(r_{12})=2L^2,
\label{eq:CMD2}
\end{equation}
where $L=\sqrt{s_{12}}$ is the length of the segment defined by the two points. For $N=3$,
\begin{equation}
\chi(r_{12},r_{13},r_{23})=-16A^2,
\label{eq:CMD3}
\end{equation}
where $A$ is the area of the triangle defined by the three points. For $N=4$,
\begin{equation}
\chi(r_{12},\dots,r_{34})=288V^2,
\label{eq:CMD4}
\end{equation}
where $V$ is the volume of the tetrahedron defined by the four points. For $N\ge5$, we must have 
\begin{equation}
\chi(r_{12},\dots,r_{(N-1)N})=0,
\label{eq:CMD5}
\end{equation}
for points in $\real{3}$ since any simplex with five or more points has zero volume in three-dimensional space.\footnote{This is easier to visualize in two-dimensional space. In that case, a simplex with four vertices (a tetrahedron) would have zero volume since its vertices would be confined to a plane. The same applies to higher-order simplexes and the corresponding higher-order volumes.} 

We now seek to go in the opposite direction. Rather than computing the squared distances $\{s_{\alpha \beta}\}$ from a set of points and using the Cayley-Menger determinants to compute volumes, we seek to verify that a set of distances actually corresponds to a set of points in three-dimensional space. In technical terms, we want the points associated with the distances to be {\em embeddable} in $\real{3}$. In order for this to be the case
the following conditions must be satisfied:\footnote{Actually, a somewhat stronger theorem can be proved. If the points are numbered in such a way that the first four points satisfy conditions \ref{enum:cond1} and \ref{enum:cond2}, then conditions \ref{enum:cond3} and \ref{enum:cond4} need only be applied to groups of points that include these four points. See \cite{Blumenthal1970} for details.}
\begin{enumerate}


\item 
$\chi(r_{\alpha\beta},r_{\alpha\gamma},r_{\beta\gamma})\le 0, 
\qquad \forall \alpha<\beta<\gamma$,
\label{enum:cond1}

\item 
$\chi(r_{\alpha\beta},r_{\alpha\gamma},r_{\alpha\delta},\dots,r_{\gamma\delta})\ge 0, 
\qquad \forall \alpha<\beta<\gamma<\delta$,
\label{enum:cond2}

\item 
$\chi(r_{\alpha\beta},r_{\alpha\gamma},r_{\alpha\delta},r_{\alpha\epsilon},\dots,r_{\delta\epsilon})= 0, 
\qquad \forall \alpha<\beta<\gamma<\delta<\epsilon$,
\label{enum:cond3}

\item 
$\chi(r_{\alpha\beta},r_{\alpha\gamma},r_{\alpha\delta},r_{\alpha\epsilon},r_{\alpha\zeta},\dots,r_{\epsilon \zeta})= 0, 
\qquad \forall \alpha<\beta<\gamma<\delta<\epsilon<\zeta$.
\label{enum:cond4}

\end{enumerate}
For example, Condition~\ref{enum:cond1} states that the Cayley-Menger determinants computed for all distinct triplets of points must be negative or zero. Conditions~\ref{enum:cond2}-\ref{enum:cond4} apply similarly to sets of four, five and six points. The above conditions are easy to understand in terms of \eref{eq:CMD2}-\eref{eq:CMD5}. They enforce the correct sign of areas and volumes in three-dimensional space and the degeneracy condition in \eref{eq:CMD5}.

\begin{figure}
\begin{center}
\includegraphics[height=2in]{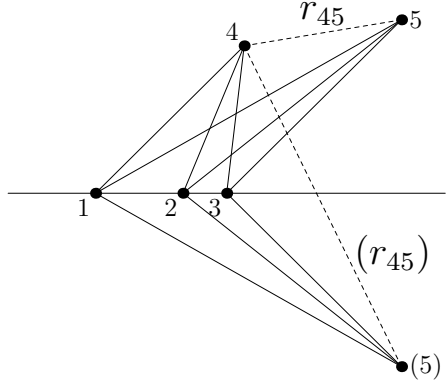}
\caption{A cluster of 5 particles. The cluster is shown projected onto the plane normal to the plane defined by particles 1, 2 and 3 shown as a horizontal line. Atom 4 lies above this plane. There are two possible positions for atom 5 (at the same height above or below the 1-2-3 plane), where the distances $\{r_{12},\dots,r_{35}\}$ are all the same and only distance $r_{45}$ differs. The alternative position of atom 5 and corresponding distance $r_{45}$ are shown in parentheses.  Based on Fig.~2 in \cite{martin1975}.}
\label{fig:r45:twosolns}
\end{center}
\end{figure}

\smallskip
As an example, let us see how this is applied for a cluster of $N=5$ particles.  There are $(5\times4)/2=10$ distances, and only $3\times5-6=9$ degrees of freedom.  There is therefore one more interatomic distance than is needed to describe the configuration and indeed there is one Cayley-Menger determinant coming from condition~\ref{enum:cond3}:
\begin{equation}
\chi(r_{12},\dots,r_{45})
= \det
\begin{bmatrix}
0             & s_{12} & s_{13} & s_{14} & s_{15} & 1 \\
s_{12} &             0 & s_{23} & s_{24} & s_{25} & 1 \\
s_{13} & s_{23} &             0 & s_{34} & s_{35} & 1 \\
s_{14} & s_{24} & s_{34} &             0 & s_{45} & 1 \\
s_{15} & s_{25} & s_{35} & s_{45} &             0 & 1 \\
            1 &             1 &             1 &             1 &             1 & 0 
\end{bmatrix} = 0.\notag
\end{equation}
This expression can be expanded out leading to the following explicit expression \cite{martin1975}:\footnote{Note that Martin \cite{martin1975} has a small typographical error in his relation.}
\begin{multline}
-\frac{1}{24}
\sum_{\alpha=1}^5
\sum_{\beta=1}^5
\sum_{\gamma=1}^5
\sum_{\delta=1}^5
\sum_{\epsilon=1}^5
\Big[
24
s_{\alpha \beta}
s_{\beta \gamma}
s_{\gamma \delta}
s_{\delta \epsilon}
-6
s_{\alpha \beta}
s_{\beta \gamma}
s_{\gamma \delta}
s_{\delta \alpha} \\
-8
s_{\alpha \beta}
s_{\gamma \delta}
s_{\delta \epsilon}
s_{\epsilon \gamma}
-12
(s_{\alpha \beta})^2
s_{\gamma \delta}
s_{\delta \epsilon}
+3
(s_{\alpha \beta})^2
(s_{\gamma \delta})^2
\Big]
=
0,
\label{eq:MartinIdentity}
\end{multline}
where in the above equation $s_{\alpha\beta}=s_{\beta\alpha}$ whenever $\alpha>\beta$. For a given set of nine squared distances, say $\{s_{11},\dots,s_{35}\}$, \eref{eq:MartinIdentity} provides a quadratic equation for the tenth squared distance, $s_{45}$,
\begin{equation}
A(s_{45})^2 + Bs_{45} + C = 0,
\label{eq:r45quadeqn}
\end{equation}
where $A$, $B$ and $C$ are functions of the other squared distances. (See \cite{KlapperDebrota1980} for explicit expressions for $A$, $B$ and $C$.) This means that there are two possible solutions for $s_{45}$ (and hence also for $r_{45}$) when the other distances are set. This situation is demonstrated in \fref{fig:r45:twosolns}.

The above example shows that although only 9 degrees of freedom are necessary to characterize a cluster of 5 particles, selecting a subset of 9 distances is not sufficient. For this reason, any interatomic potential energy extension (see \sref{sec:s_motion}) must be expressed as a function of $N(N-1)/2$ arguments and not just an arbitrary subset of $3N-6$ of them.

\section{Useful lemmas}
\label{ch:noll}
In his 1955 paper on the "Derivation of the fundamental equations of continuum thermodynamics from statistical mechanics" \cite{noll1955}, Walter Noll proves two lemmas that play an important role in his derivation. In this section, for completeness, we derive Noll's first lemma and then extend it to arbitrary curved ``paths of interaction''. We then derive Noll's second lemma for this more generalized case.

Let $\bm{f}(\bm{v},\bm{w})$ be a tensor-valued function of two vectors $\bm{v}$ and $\bm{w}$, which satisfies the following three conditions:
\begin{enumerate}
\item $\bm{f}(\bm{v},\bm{w})$ is defined for all $\bm{v}$ and $\bm{w}$ and is continuously differentiable \label{condition_1}.
\item There exists a $\delta > 0$, such that the auxiliary function $\bm{g}(\bm{v},\bm{w})$, defined through
\begin{equation}
\bm{g}(\bm{v},\bm{w}):=\bm{f}(\bm{v},\bm{w}) \vnorm{\bm{v}} ^{3+\delta} \vnorm{w}^{3+\delta},
\end{equation}
and its gradients $\nabla_{\bm{v}} \bm{g}$ and $\nabla_{\bm{w}}\bm{g}$ are bounded.\label{condition_2}
\item $\bm{f}(\bm{v},\bm{w})$ is antisymmetric, i.e.,
\begin{equation}
\bm{f}(\bm{v},\bm{w}) = -\bm{f}(\bm{w},\bm{v}).
\label{eqn:anti_sym_1}
\end{equation}
\label{condition_3}
\end{enumerate}
\begin{lemma}
\label{lem1}
Under the conditions mentioned above, the following equation holds:\footnote{The expression in Noll's paper appears transposed relative to \eref{eqn:lemma1}. This is because the gradient and divergence operations used by Noll are the transpose of our definitions. See end of \sref{ch:intro}.}
\begin{equation}
\label{eqn:lemma1}
\int_{\bm{y} \in \real{3}} \bm{f}(\bm{x},\bm{y}) \,d\bm{y} = -\frac{1}{2} \divr_{\bm{x}} \int_{\bm{z} \in \real{3}} \left [ \int_{s=0}^{1} \bm{f}(\bm{x} + s\bm{z},\bm{x}-(1-s)\bm{z}) \, ds \right ]\otimes \bm{z} \, d\bm{z}.
\end{equation}
\end{lemma}
\begin{proof}
Conditions (1) and (2) guarantee absolute convergence and allow the order of integration to be swapped. From \eref{eqn:anti_sym_1} we have
\begin{equation}
\int_{\bm{y} \in \real{3}}\bm{f}(\bm{x},\bm{y}) d\bm{y} = - \int_{\bm{y} \in \real{3}}\bm{f}(\bm{y},\bm{x}) d\bm{y}.
\end{equation}
Introduce new integration variables: on the left replace $\bm{y}$ with $\bm{z} = \bm{x} - \bm{y}$ and on the right replace $\bm{y}$ with $\bm{z} = \bm{y} - \bm{x}$. Thus,
\begin{equation}
\int_{\bm{z} \in \real{3}}\bm{f}(\bm{x},\bm{x}-\bm{z}) d\bm{z} = - \int_{\bm{z} \in \real{3}}\bm{f}(\bm{x}+\bm{z},\bm{x}) d\bm{z}.
\label{eqn:var_change}
\end{equation}
Note that a minus sign on the left-hand side is dropped by formally reversing the integration bounds. The two terms in \eref{eqn:var_change} are equal to $\int_{\bm{y} \in \real{3}} \bm{f}(\bm{x},\bm{y}) d\bm{y}$, we can therefore write
\begin{equation}
\label{eqn:average}
\int_{\bm{y} \in \real{3}}\bm{f}(\bm{x},\bm{y}) d\bm{y} = \frac{1}{2}\int_{\bm{z} \in \real{3}} \left [ \bm{f}(\bm{x},\bm{x} - \bm{z}) - \bm{f}(\bm{x} + \bm{z},\bm{x}) \right ] d\bm{z}.
\end{equation}
Next, from the chain rule, we have
\begin{equation}
\nabla_{\bm{x}} \bm{f}(\bm{x} + s \bm{z},\bm{x} - (1-s)\bm{z}) = \nabla_{\bm{v}} \bm{f} + \nabla_{\bm{w}} \bm{f},
\label{eqn:chain_rule}
\end{equation}
where $s \in \mathbb{R}$. Similarly,
\begin{equation}
\frac{d}{ds} \bm{f}(\bm{x} + s \bm{z},\bm{x} - (1-s)\bm{z}) = \left ( \nabla_{\bm{v}} \bm{f} + \nabla_{\bm{w}} \bm{f} \right ) \bm{z}.
\label{eqn:differential_f}
\end{equation}
Combining \eref{eqn:chain_rule} and \eref{eqn:differential_f}, we have
\begin{equation}
\left [\nabla_{\bm{x}} \bm{f}(\bm{x} + s \bm{z},\bm{x} - (1-s)\bm{z}) \right ] \bm{z} = \frac{d}{ds} \bm{f}(\bm{x} + s \bm{z},\bm{x} - (1-s)\bm{z}).
\end{equation}
Integrating the above equation with respect to $s$ from $0$ to $1$ gives
\begin{equation}
\label{eqn:final}
\left [ \nabla_{\bm{x}} \int_{s=0}^{1} \bm{f}(\bm{x} + s\bm{z},\bm{x} - (1-s)\bm{z}) ds\right ] \bm{z} = \bm{f}(\bm{x} + \bm{z},\bm{x}) - \bm{f}(\bm{x},\bm{x} - \bm{z}).
\end{equation}
Substituting \eref{eqn:final} into \eref{eqn:average}, and using the identity 
\begin{equation}
\label{eqn:identity_lemma}
(\nabla_{\bm{x}} \bm{T}) \bm{a} = \divr_{\bm{x}} (\bm{T} \otimes \bm{a}),
\end{equation}
where $\bm{T}(\bm{x})$ is a tensor of any order and $\bm{a}$  is a vector which is  not a function of $\bm{x}$, gives \eref{eqn:lemma1}.\qed
\end{proof}
Lemma \ref{lem1} provides us a solution $\bm{F}(\bm{x})$, to the equation
\begin{equation}
\int_{\bm{y} \in \real{3}} \bm{f}(\bm{x},\bm{y}) \, d\bm{y} = \divr_{\bm{x}} \bm{F}(\bm{x}).
\label{eqn:differential_F}
\end{equation}
It is clear that $\bm{F}(\bm{x})$ is only a single solution out of an infinite set of solutions that differ from it by a divergenceless second-order tensor field. The following lemma extends Lemma \ref{lem1} to accommodate other possible solutions of which the solution given in \eref{eqn:lemma1} is a special case. Before that we give the following definition:
\begin{definitions}
\label{def:poi}
For any $\bm{u}$ and $\bm{v} \in \real{3}$, the ``the path of interaction'' of $\bm{u}$ and $\bm{v}$ is a contour that joins $\bm{u}$ and $\bm{v}$ such that for any $\tbm{u}$ and $\tbm{v}$ in $\real{3}$, where $\vnorm{\bm{u}-\bm{v}} = \vnorm{\tbm{u}-\tbm{v}}$, the path of interaction of pairs $(\bm{u},\bm{v})$ and $(\tbm{u},\tbm{v})$ are related by a rigid body transformation $\langle \bm{Q} \mid \bm{c} \rangle$, for some $\bm{Q} \in \bm{SO}(3)$ and $\bm{c} \in \mathbb{R}^3$, with $\bm{Q}(\bm{u}-\bm{v})=\tilde{\bm{u}}-\tilde{\bm{v}}$ and $\bm{Q}=\bm{I}$, whenever $\bm{u}-\bm{v} = \tilde{\bm{u}}-\tilde{\bm{v}}$. See \fref{fig:define_poi}.
\end{definitions}
Basically, this definition enforces the condition that the ``path of interaction'' is a contour whose shape is only a function of the distance between the points that it connects. The shape of the contour for a given distance is assumed to be dictated by the nature of bonding in the material. Here we assume this shape to be known. (See also footnote~\ref{foot:curveshape} on page~\pageref{foot:curveshape}).

\begin{figure}
\centering
\includegraphics[totalheight=0.2\textheight]{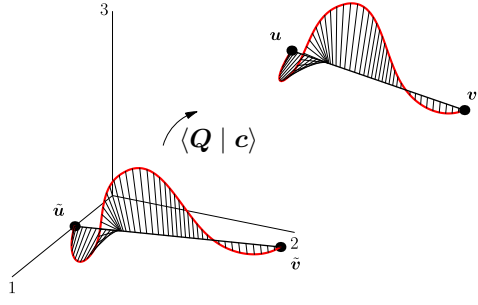}
\caption{The property of a \emph{path of interaction} as mentioned in the definition is illustrated in this figure. Two paths of interaction for the pairs $(\bm{u},\bm{v})$ and $(\tilde{\bm{u}},\tilde{\bm{v}})$ are shown in the figure.}
\label{fig:define_poi}
\end{figure}

Let $(\bm{e}_1,\bm{e}_2,\bm{e}_3)$ be a basis of $\real{3}$. For every $l>0$, let $\bm{\Upsilon}_l:[0,1] \to \real{3}$ be a continuously differentiable contour in $\real{3}$ such that
\begin{subequations}
\begin{equation}
\bm{\Upsilon}_l(s) \cdot \bm{e}_1 = sl, \quad  0 \le s \le 1, \label{eqn:g_prop_1}
\end{equation}
\begin{equation}
\bm{\Upsilon}_l(0) = (0,0,0); \qquad \bm{\Upsilon}_l(1) = (l,0,0). \label{eqn:g_prop_2} 
\end{equation}
\end{subequations}
\fref{fig:contour} describes the properties of the contour $\bm{\Upsilon}_l$ mentioned above.\\

\begin{figure}
\begin{center}
\subfigure[]{\label{fig:adm}\includegraphics[totalheight=0.2\textheight]{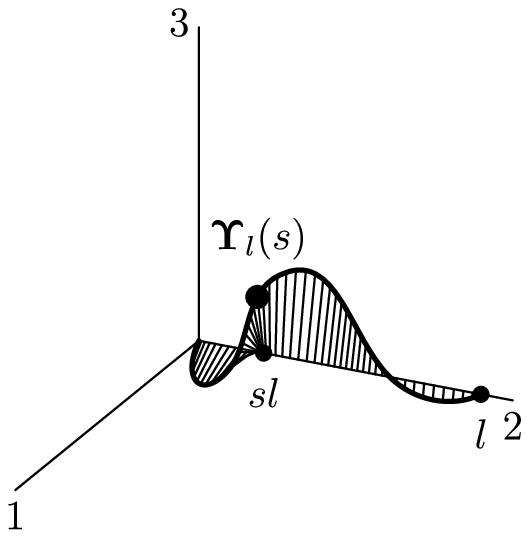}}
\subfigure[]{\label{fig:inadm}\includegraphics[totalheight=0.2\textheight]{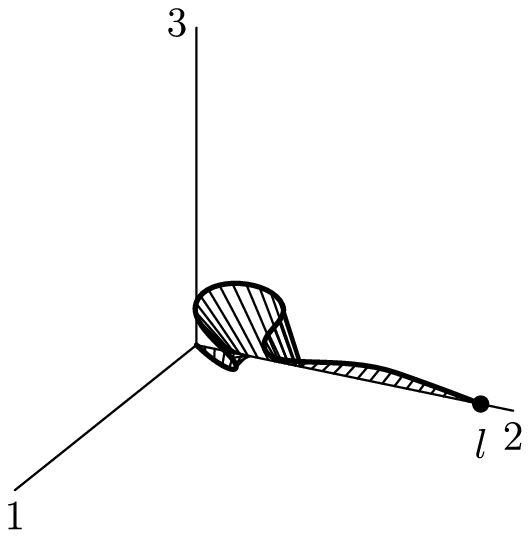}}
\end{center}
\caption{Frame (a) shows an admissible contour and frame (b) shows an inadmissible contour which violates property \eref{eqn:g_prop_1} of $\bm{\Upsilon}_l$.}
\label{fig:contour}
\end{figure}

By the definition of the \emph{path of interaction}, it is clear that the shape of any path of interaction can be described by the contour $\bm{\Upsilon}_l$. Moreover, from its properties it is possible to explicitly define the path of interaction as a function of $\bm{\Upsilon}_l$, a given point $\bm{x}$ through which it passes, and a vector $\bm{z}$ which connects its end points. More precisely, for any $0 \le \bar{s} \le 1$, $\bm{x} \in \real{3}$ and $\bm{z} \in \real{3}$, let $\bm{\bm{\Gamma}}(\cdot \, ;\bar{s},\bm{x},\bm{z}): [0,1] \to \real{3}$ denote a path of interaction in $\real{3}$, such that
\begin{equation}
\bm{\bm{\Gamma}}(s;\bar{s},\bm{x},\bm{z}) = \bm{x} + \Qz \left [\bm{\Upsilon}_{\vnorm{\bm{z}}}(s) -\bm{\Upsilon}_{\vnorm{\bm{z}}}(\bar{s}) \right ],\label{eqn:transform}
\end{equation}
for some $\Qz \in \bm{SO}(3)$, satisfying
\begin{equation}
\Qz\bm{e}_1 = -\frac{\bm{z}}{\vnorm{\bm{z}}}, 
\label{eqn:z_condition}
\end{equation}
as shown in \fref{fig:path}. Let us now verify that $\bm{\Gamma}$ qualifies as a path of interaction. It is clear from \eref{eqn:transform} that the contours $\bm{\Gamma}(s;\bar{s},\bm{x},\bm{z})$ and $\bm{\Upsilon}_{\vnorm{\bm{z}}}(s)$ are related through a rigid body transformation, such that
\begin{subequations}
\label{prop_Gamma}
\begin{equation}
\label{eqn:G_prop_1}
\bm{\bm{\Gamma}}(\bar{s};\bar{s},\bm{x},\bm{z}) = \bm{x},
\end{equation}
\begin{equation}
\label{eqn:G_prop_2}
\bm{\bm{\Gamma}}(1;\bar{s},\bm{x},\bm{z}) - \bm{\bm{\Gamma}}(0;\bar{s},\bm{x},\bm{z}) = -\bm{z}.
\end{equation}
\end{subequations}
\fref{fig:path} describes the properties mentioned above.\\

\begin{figure}
\centering
\includegraphics[totalheight=0.2\textheight]{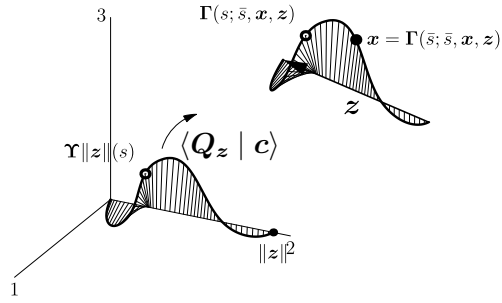}
\caption{The contour $\bm{\Gamma}(.;\bar{s},\bm{x},\bm{z})$ passes through $\bm{x}$ and the vector $\bm{z}$ is the difference between the two end points of the contour. It is related to the contour $\bm{\Upsilon}_{\vnorm{\bm{z}}}$, such that any point $\bm{\Upsilon}_{\vnorm{\bm{z}}}(s)$ is mapped to  $\bm{\Gamma}(s;\bar{s},\bm{x},\bm{z})$ (shown above as a hollow dot) by a rigid body motion $\langle \Qz | \bm{c} \rangle$, where $\Qz$ represents rotation and $\bm{c}= \bm{x} - \Qz\bm{\Upsilon}_{\vnorm{\bm{z}}}(\bar{s})$, represents translation.}
\label{fig:path}
\end{figure}

From \eref{eqn:G_prop_1}, it follows that $\bm{\Gamma}$ passes through the point $\bm{x}$ and from \eref{eqn:G_prop_2} it is clear that $\bm{z}$ is the vector joining its endpoints. Moreover $\Qz$ in \eref{eqn:transform} is made independent of $\bar{s}$ to ensure that the condition $\Qz = \bm{I}$ whenever  $\bm{x}-\bm{y} = \bm{u}-\bm{v}$, in the definition of the path of interaction, is satisfied.  \\

For any $\bm{x} \in \real{3}$, let
\begin{equation}
\mathcal{C}_{\bm{x}} = \{\bm{\bm{\Gamma}}(\cdot;\bar{s},\bm{x},\bm{z}): 0 \le \bar{s} \le 1, \bm{z} \in \real{3} \},
\end{equation}
denote the set of all paths of interaction that pass through $\bm{x}$. Now for every $\bm{y}$ on a contour $\bm{\bm{\Gamma}} \in \mathcal{C}_{\bm{x}}$, let $\bm{y}_{\perp}$ denote the projection of $\bm{y}$ on the line joining the end points of the path. It is easy to see that $\bm{y}_{\perp}(s;\bar{s},\bm{x},\bm{z})$ is given by
\begin{equation}
\bm{y}_{\perp}(s;\bar{s},\bm{x},\bm{z}) = \bm{x} + \Qz \left [ s \vnorm{\bm{z}} \bm{e}_1 - \bm{\Upsilon}_{\vnorm{\bm{z}}}(\bar{s}) \right ].
\end{equation}
Therefore, from \eref{eqn:G_prop_1} we have
\begin{equation}
\bm{x}_{\perp}(\bar{s},\bm{x},\bm{z}) = \bm{y}_{\perp}(\bar{s};\bar{s},\bm{x},\bm{z}) = \bm{x} + \Qz \left [ \bar{s} \vnorm{\bm{z}} \bm{e}_1 - \bm{\Upsilon}_{\vnorm{\bm{z}}}(\bar{s}) \right ].
\end{equation}
Using \eref{eqn:z_condition} this simplifies to 
\begin{equation}
\label{eqn:x_perp}
\bm{x}_{\perp}(\bar{s},\bm{x},\bm{z}) = \bm{x} - \bar{s}\bm{z} - \Qz\bm{\Upsilon}_{\vnorm{\bm{z}}}(\bar{s}).
\end{equation}
We now generalize Noll's Lemma \ref{lem1} to arbitrary paths of interaction.

\begin{lemma}
\label{lem2}
For given paths of interaction on $\real{3} \times \real{3}$, and under conditions \ref{condition_1}-\ref{condition_3} on $\bm{f}(\bm{v},\bm{w})$ given at the start of this appendix, the following equation holds:
\begin{equation}
\label{eqn:lem2}
\int_{\bm{y} \in \real{3}} \bm{f}(\bm{x},\bm{y}) \,d\bm{y} = \frac{1}{2} \divr_{\bm{x}} \int_{\bm{z} \in \real{3}} \left [ \int_{\bar{s}=0}^{1} \bm{f}(\bm{x}_{\perp} + \bar{s} \bm{z}, \bm{x}_{\perp}-(1-\bar{s})\bm{z}) \right ] \otimes \Qz\bm{\Upsilon}'_{\vnorm{\bm{z}}}(\bar{s}) \, d\bar{s}  \, d\bm{z},
\end{equation}
where $\bm{x}_{\perp}(\bar{s},\bm{x},\bm{z})$ is given by \eref{eqn:x_perp}.
\end{lemma}
\begin{proof}
From \eref{eqn:average} we have
\begin{equation}
\label{eqn:average_*}
\int_{\bm{y} \in \real{3}} \bm{f}(\bm{x},\bm{y}) \,d\bm{y} = \frac{1}{2}\int_{\bm{z} \in \real{3}} \left [ \bm{f}(\bm{x},\bm{x} - \bm{z}) - \bm{f}(\bm{x} + \bm{z},\bm{x}) \right ] d\bm{z}.
\end{equation}
Next, from the chain rule, we have
\begin{equation}
\nabla_{\bm{x}} \bm{f}(\bm{x}_{\perp} + \bar{s} \bm{z}, \bm{x}_{\perp}-(1-\bar{s})\bm{z}) = (\nabla_{\bm{v}} \bm{f} + \nabla_{\bm{w}} \bm{f}) \nabla_{\bm{x}} \bm{x}_{\perp},
\end{equation}
where $\bar{s} \in \mathbb{R}$. From \eref{eqn:x_perp} we have $\nabla_{\bm{x}} \bm{x}_{\perp} = \bm{I}$. Therefore
\begin{equation}
\label{eqn:chain_rule_perp}
\nabla_{\bm{x}} \bm{f}(\bm{x}_{\perp} + \bar{s} \bm{z}, \bm{x}_{\perp}-(1-\bar{s})\bm{z}) = \nabla_{\bm{v}} \bm{f} + \nabla_{\bm{w}} \bm{f}.
\end{equation}
Similarly,
\begin{equation}
\label{eqn:differential_f_perp_1}
\frac{d}{d\bar{s}} \bm{f}(\bm{x}_{\perp} + \bar{s} \bm{z}, \bm{x}_{\perp}-(1-\bar{s})\bm{z}) = \left ( \nabla_{\bm{v}} \bm{f} + \nabla_{\bm{w}} \bm{f} \right ) \left (\frac{d \bm{x}_{\perp}}{d\bar{s}} + \bm{z} \right ).
\end{equation}
From \eref{eqn:x_perp}, we have
\begin{equation}
\label{eqn:x_perp_der}
\frac{d \bm{x}_{\perp}}{d\bar{s}} = -\bm{z} - \Qz \bm{\Upsilon}'_{\vnorm{\bm{z}}}(\bar{s}).
\end{equation}
Substituting \eref{eqn:x_perp_der} into \eref{eqn:differential_f_perp_1}, we have
\begin{equation}
\label{eqn:differential_f_perp_2}
\frac{d}{d\bar{s}} \bm{f}(\bm{x}_{\perp} + \bar{s} \bm{z}, \bm{x}_{\perp}-(1-\bar{s})\bm{z}) = -\left ( \nabla_{\bm{v}} \bm{f} + \nabla_{\bm{w}} \bm{f} \right ) \Qz \bm{\Upsilon}'_{\vnorm{\bm{z}}}(\bar{s}).
\end{equation}
Combining \eref{eqn:chain_rule_perp} and \eref{eqn:differential_f_perp_2}, we have
\begin{equation}
\frac{d}{d\bar{s}}\bm{f}(\bm{x}_{\perp} + \bar{s} \bm{z}, \bm{x}_{\perp}-(1-\bar{s})\bm{z}) = -\left [ \nabla_{\bm{x}} \bm{f}(\bm{x}_{\perp} + \bar{s} \bm{z}, \bm{x}_{\perp}-(1-\bar{s})\bm{z}) \right ] \Qz \bm{\Upsilon}'_{\vnorm{\bm{z}}}(\bar{s}).
\end{equation}
Integrating both sides over the interval $\bar{s} \in [0,1]$, we have
\begin{align}
\label{eqn:int_0_1}
\bm{f}(\bm{x}_{\perp}(1,\bm{x},\bm{z}) + \bm{z},\bm{x}_{\perp}(1,\bm{x},\bm{z})) &- \bm{f}(\bm{x}_{\perp}(0,\bm{x},\bm{z}),\bm{x}_{\perp}(0,\bm{x},\bm{z}) - \bm{z}) = \notag \\ 
&-\int_{\bar{s}=0}^{1} [ \nabla_{\bm{x}} \bm{f}(\bm{x}_{\perp} + \bar{s} \bm{z}, \bm{x}_{\perp}-(1-\bar{s})\bm{z}) \Qz \bm{\Upsilon}'_{\vnorm{\bm{z}}}(\bar{s}) \, d\bar{s}.
\end{align}
Using \eref{eqn:x_perp}, \eref{eqn:z_condition}, and property \eref{eqn:g_prop_2} of $\bm{\Upsilon}_l$, we have
\begin{equation}
\label{eqn:x_perp_end}
\bm{x}_{\perp}(0,\bm{x},\bm{z})=\bm{x} \qquad ; \qquad \bm{x}_{\perp}(1,\bm{x},\bm{z})=\bm{x}.
\end{equation}
Substituting, \eref{eqn:x_perp_end} into \eref{eqn:int_0_1} and using the identity \eref{eqn:identity_lemma}, we have
\begin{equation}
\bm{f}(\bm{x}+\bm{z},\bm{x}) - \bm{f}(\bm{x},\bm{x}-\bm{z}) = - \divr_{\bm{x}} \int_{\bar{s}=0}^{1} \bm{f}(\bm{x}_{\perp} + \bar{s} \bm{z}, \bm{x}_{\perp}-(1-\bar{s})\bm{z}) \otimes \Qz \bm{\Upsilon}'_{\vnorm{\bm{z}}}(\bar{s}) \, d\bar{s}.
\end{equation}
If we now substitute the above equation into \eref{eqn:average_*}, the lemma is proved.\qed
\end{proof}

\begin{lemma}
\label{lem3}
Let $\Omega \subset \real{3}$, with a piecewise smooth surface $\mathcal{S}$. For given paths of interaction on $\real{3} \times \real{3}$, and under the conditions \ref{condition_1}-\ref{condition_3} on $\bm{f}(\bm{v},\bm{w})$ given at the start of this appendix, the following equation holds:
\begin{align}
\int_{\bm{x} \in \Omega} & \int_{\bm{y} \in \Omega^{\rm{c}}} \bm{f}(\bm{x},\bm{y}) \,d\bm{y} \, d\bm{x} \notag \\
&= \frac{1}{2} \int_{\mathcal{S}} \int_{\bm{z} \in \real{3}} \int_{\bar{s}=0} ^{1} \bm{f}(\bm{x}_{\perp} + \bar{s} \bm{z}, \bm{x}_{\perp} - (1-\bar{s})\bm{z}) (\Qz\bm{\Upsilon}'_{\vnorm{\bm{z}}}(\bar{s}) \cdot \bm{n} ) \,d\bar{s}\,d\bm{z} \,d\mathcal{S}(\bm{x}),
\label{eqn:lem3}
\end{align}
where $\bm{x}_{\perp}(\bar{s},\bm{x},\bm{z})$ is given by \eref{eqn:x_perp} and $\Omega^{\rm{c}}:=\real{3} \backslash \Omega$.
\end{lemma}
\begin{proof}
We immediately see that due to antisymmetry of $\bm{f}(\bm{v},\bm{w})$, we have
\begin{equation}
\int_{\bm{x} \in \Omega} \int_{\bm{y} \in \Omega} \bm{f}(\bm{x},\bm{y}) \,d\bm{y} \, d\bm{x} = \bm{0}.
\end{equation}
Therefore, we have
\begin{equation}
\label{eqn:anti_sym_2}
\int_{\bm{x} \in \Omega} \int_{\bm{y} \in \Omega^{\rm{c}}} \bm{f}(\bm{x},\bm{y}) \,d\bm{y} \, d\bm{x} = \int_{\bm{x} \in \Omega} \int_{\bm{y} \in \real{3}} \bm{f}(\bm{x},\bm{y}) \,d\bm{y} \, d\bm{x}.
\end{equation}
From Lemma \ref{lem2}, we have
\begin{equation}
\int_{\bm{y} \in \real{3}} \bm{f}(\bm{x},\bm{y}) \, d\bm{y} = \divr_{\bm{x}} \bm{g}(\bm{x}),
\end{equation}
with
\begin{equation}
\bm{g}(\bm{x}) = \frac{1}{2} \int_{\bm{z} \in \real{3}} \int_{\bar{s}=0}^{1} \bm{f}(\bm{x}_{\perp}+\bar{s}\bm{z},\bm{x}_{\perp} - (1-\bar{s}) \bm{z}) \otimes \Qz\bm{\Upsilon}'_{\vnorm{\bm{z}}}(\bar{s}) \, d\bar{s} \, d\bm{z},
\end{equation}
where $\bm{x}_{\perp}(\bar{s},\bm{x},\bm{z})$ is given by \eref{eqn:x_perp}. From the divergence theorem we have
\begin{equation}
\int_{\bm{x} \in \Omega} \divr_{\bm{x}} \bm{g}(\bm{x}) \, d\bm{x} = \int_{\mathcal{S}} \bm{g}(\bm{x}) \cdot \bm{n}(\bm{x}) \, d\mathcal{S}(\bm{x}).
\label{eqn:divergence}
\end{equation}
Using \eref{eqn:anti_sym_2}-\eref{eqn:divergence}, we obtain \eref{eqn:lem3}.\qed
\end{proof}

\section{Derivation of the Hardy stress from the Murdoch--Hardy procedure}
\label{sec:hardy_limit}
The Hardy stress, $\stressmh^{\rm{H}}_{w,\rm{v}}$, was derived in \sref{sec:murdoch_proc} subject to the condition:
\begin{equation}
\label{eqn:condition}
\divr_{\bm{x}} \stressmh^{\rm{H}}_{w,\rm{v}}(\bm{x},t) = \sum_{\substack{\alpha,\beta \\  \alpha \neq \beta}} \bm{f}_{\alpha\beta} w(\bm{x}_\alpha - \bm{x}).
\end{equation}
Recall that the expression for $\stressmh^{\rm{H}}_{w,\rm{v}}$ given in \eref{eqn:hardy_stress} was obtained by using the generator function $\bm{g}^{\rm{H}}$ (see \eref{eqn:g_1}), which is actually a distribution. Due to the inherent obstacles present in this procedure (see footnote~\ref{fn:mollifier} on page~\pageref{fn:mollifier}), we now derive the expression for the Hardy stress given in \eref{eqn:hardy_stress}, by avoiding the use of the Dirac delta distribution.

\begin{definition}
Consider the following definitions for the function $\eta(\bm{x})$ and associated functions taken from \cite{evans}:
\begin{enumerate}
\item
Define  a mollifier $\eta \in C^{\infty}(\real{3})$ by 
\begin{equation}
\eta(\bm{r}) := \left \{ 
\begin{array}{ll}
C \rm{exp} \left ( \frac{1}{\vnorm{\bm{r}}^2 - 1} \right )  &\mbox{if $\vnorm{\bm{r}} < 1$}, \\
0 &\mbox{if $\vnorm{\bm{r}} \ge 1$}, 
\end{array} 
\right.
\end{equation}
where the constant $C>0$ is selected so that $\int_{\real{3}} \eta \, d\bm{r} = 1$. 

\item
For each $\epsilon>0$, set
\begin{equation}
\eta_\epsilon(\bm{r}) := \frac{1}{\epsilon^3} \eta \left ( \frac{\bm{r}}{\epsilon} \right ).
\end{equation}
The family of functions $\eta_\epsilon$ are $C^\infty$ and satisfy
\begin{equation}
\int_{\real{3}} \eta_\epsilon \, d\bm{r} = 1.
\end{equation}
The support of $\eta_\epsilon$ is contained in a ball of radius $\epsilon$ centered at $\bm{0}$.

\item
If the function $h:\real{3} \to \mathbb{R}$ is locally integrable, define its mollification $h_\epsilon(\bm{r})$ as
\begin{equation}
h_\epsilon(\bm{r}) := \int_{\real{3}} \eta_\epsilon (\bm{r} - \bm{y}) h(\bm{y}) \, d\bm{y}.
\label{eqn:heps}
\end{equation}
\end{enumerate}
\end{definition}
We use the following property of mollifiers in our derivation:
\begin{equation}
\label{eqn:property}
h_\epsilon \to h \ \ \text{almost everywhere as} \ \ \epsilon \to 0.
\end{equation}
For the proof of \eref{eqn:property}, refer to \cite{evans}.
Equation \eref{eqn:condition} can be rewritten as
\begin{equation}
\label{eqn:condition_split}
\divr_{\bm{x}} \stressmh^{\rm{H}}_{w,\rm{v}}(\bm{x},t) = \sum_{\substack{\alpha,\beta \\  \alpha \neq \beta}} \bm{f}_{\alpha\beta} \sqrt{w(\bm{x}_\alpha - \bm{x}) w(\bm{x}_\alpha - \bm{x})},
\end{equation}
since $w(\bm{x}_\alpha - \bm{x})>0$. Now, since $\sqrt{w(\bm{x}_\beta - \bm{y})}$ is locally integrable, using property \eref{eqn:property}, we have
\begin{equation}
\label{eqn:property_apply}
\sqrt{w(\bm{x}_\alpha- \bm{x})} = \lim_{\epsilon \to 0} \int_{\real{3}} \sqrt{w(\bm{x}_\beta - \bm{y})} \eta_{\epsilon}(\bm{x}_\beta - \bm{x}_\alpha + \bm{x} - \bm{y}) \, d\bm{y},
\end{equation}
where referring to \eref{eqn:heps}, $h(\bm{y})=\sqrt{w(\bm{x}_\beta-\bm{y})}$ and $\bm{r}=\bm{x}_\beta - \bm{x}_\alpha + \bm{x}$. Using \eref{eqn:property_apply}, \eref{eqn:condition_split} can be rewritten as
\begin{align}
\divr_{\bm{x}} \stressmh^{\rm{H}}_{w,\rm{v}}(\bm{x},t) &= \lim_{\epsilon \to 0} \sum_{\substack{\alpha,\beta \\ \alpha \ne \beta}} \int_{\real{3}} \bm{f}_{\alpha\beta}  \sqrt{w(\bm{x}_\alpha - \bm{x}) w(\bm{x}_\beta - \bm{y})} \, \eta_\epsilon(\bm{x}_\beta - \bm{x}_\alpha + \bm{x} - \bm{y}) \, d\bm{y} \notag \\
&=: \lim_{\epsilon \to 0} \sum_{\substack{\alpha,\beta \\ \alpha \ne \beta}} \int_{\real{3}} \bm{g}^\epsilon_{\alpha\beta}(\bm{x},\bm{y},t) \, d\bm{y}.
\end{align}
Now, note that the function $\bm{g}^\epsilon_{\alpha\beta}$ is anti-symmetric with respect to its arguments, $\bm{x}$ and $\bm{y}$, for each $\epsilon>0$ and satisfies all the necessary conditions for the application of Lemma \ref{lem1} in Appendix \ref{ch:noll}. Therefore, from \ref{eqn:lemma1} it follows that
\begin{align}
\divr_{\bm{x}} \stressmh^{\rm{H}}_{w,\rm{v}}(\bm{x},t) = &\lim_{\epsilon \to 0} \sum_{\substack{\alpha,\beta \\ \alpha \ne \beta}} \left ( -\frac{1}{2} \divr_{\bm{x}} \int_{\real{3}} \left [ \int_{s=0}^{1} \bm{f}_{\alpha\beta} \sqrt{w(\bm{x}_\alpha - \bm{x} - s\bm{z}) w(\bm{x}_\beta - \bm{x} + (1-s)\bm{z})} \right.\right. \notag \\
& \times \, \eta_\epsilon(\bm{x}_\beta - \bm{x}_\alpha + \bm{z}) \, ds \bigg ] \otimes \bm{z} \, d\bm{z} \bigg ) \notag \\
&= \lim_{\epsilon \to 0} \sum_{\substack{\alpha,\beta \\ \alpha \ne \beta}} \left (- \frac{1}{2} \int_{\real{3}} \int_{s=0}^{1} \nabla_{\bm{x}} \left ( \bm{f}_{\alpha\beta} \sqrt{w(\bm{x}_\alpha - \bm{x} - s\bm{z}) w(\bm{x}_\beta - \bm{x} + (1-s)\bm{z})} \right ) \right. \notag \\
& \times \, \eta_\epsilon(\bm{x}_\beta - \bm{x}_\alpha + \bm{z}) \bm{z} \, ds \, d\bm{z} \bigg ),
\end{align}
where in the last equality we have used the identity given in \eref{eqn:identity_lemma} and the fact that $\eta_\epsilon$ is independent of $\bm{x}$ in the above equation. Since $w$ is positive, and if $w$ ,s a continuously differentiable function, it follows that 
\begin{equation}
\nabla_{\bm{x}} \left (\bm{f}_{\alpha\beta} \sqrt{w(\bm{x}_\alpha - \bm{x} - s\bm{z}) w(\bm{x}_\beta - \bm{x} + (1-s)\bm{z})} \right ) \bm{z}
\label{eqn:newh}
\end{equation}
is a locally integrable function of $\bm{z}$. Therefore, by property \ref{eqn:property} and using the property, $\eta_\epsilon(\bm{r})=\eta_\epsilon(-\bm{r})$, we have
\begin{align}
\divr_{\bm{x}} \stressmh^{\rm{H}}_{w,\rm{v}}(\bm{x},t) &= -\frac{1}{2} \sum_{\substack{\alpha,\beta \\ \alpha \ne \beta}} \int_{s=0}^{1} \nabla_{\bm{x}} \left [ \bm{f}_{\alpha\beta} w(\bm{x}_\alpha - \bm{x} - s(\bm{\bm{x}_\alpha - \bm{x}_\beta})) \right ] (\bm{x}_\alpha - \bm{x}_\beta) \, ds \notag \\ 
&= \divr_{\bm{x}} \Bigg [ \frac{1}{2} \sum_{\substack{\alpha,\beta \\ \alpha \ne \beta}} \int_{s=0}^{1} -[ \bm{f}_{\alpha\beta} w( (1-s) \bm{x}_\alpha + s \bm{x}_\beta- \bm{x} ) \otimes (\bm{x}_\alpha - \bm{x}_\beta)] \, ds \Bigg ],\label{eqn:divr_hardy_stress}
\end{align}
where referring to \eref{eqn:heps}, $\bm{r}=\bm{x}_\alpha-\bm{x}_\beta$, $\bm{y}=\bm{z}$, and $h(\bm{z})$ is given in \eref{eqn:newh}.  Comparing both sides of \eref{eqn:divr_hardy_stress}, we arrive at the expression for the Hardy stress tensor given in \eref{eqn:hardy_stress}.

\begin{acknowledgements}
The authors would like to thank Dr. Marcel Arndt for assisting with the translation of Noll's paper \cite{noll1955}. In addition, the authors would like to thank Roger Fosdick, Richard~D.~James and Ryan~S.~Elliott for many stimulating conversations regarding various aspects of this paper and, in particular, the issue of the extension of potential energy functions. Additional helpful conversations with Ronald~E.~Miller, Eliot~Fried, Richard~B.~Lehoucq and Andrew~L.~Ruina are gratefully acknowledged.
\end{acknowledgements}

\bibliographystyle{spmpsci}       
\bibliography{admal_tadmor_references}        
\end{document}